\documentclass{acmsmall}

\newcommand{\figurename}{Figure}
\usepackage[T1]{fontenc}
\usepackage[utf8]{inputenc}
\usepackage{amsmath,amssymb,stmaryrd,xspace,pdfsync,mathrsfs,graphicx,refcount,verbatim,bbm}
\usepackage{microtype}
\usepackage{tikz}
\usepackage{enumitem}
\usepackage{relsize}
\usepackage{caption}
\usepackage{tabularx}

\usepackage{ifxetex}

\ifxetex
\usepackage{mathspec}
\usepackage{fontspec}
\usepackage[no-sscript]{xltxtra}
\usepackage[bookmarks,bookmarksopen,bookmarksdepth=2,xetex]{hyperref}
\else
  \usepackage[T1]{fontenc}
  \usepackage[utf8]{inputenc}
  \usepackage{lmodern}
  \usepackage[unicode,bookmarks,bookmarksopen,bookmarksdepth=2]{hyperref}
  \usepackage{bookmark}
\fi

\usepackage[mathcal]{euscript}

\usetikzlibrary{decorations.text,arrows,snakes,shapes,fit,shadows,calc,patterns,intersections}

\setlist[1]{itemsep=.5ex,leftmargin=*}
\setlist[1,itemize]{label=--}
\setlist[1,enumerate]{label=(\arabic*),ref=(\arabic*)}

\usepackage[textwidth=4cm]{todonotes}
\newcounter{sauvegarde}
\newcommand{\mynote}[3][]{\todo[caption={\sf #3}, color={
    \ifnum#2=0 green!20
    \else\ifnum#2=1 orange!30
    \else\ifnum#2=2 yellow!20
    \else\ifnum#2=3 cyan!20
    \else magenta!20\fi\fi\fi\fi}, size=\tiny, #1]{\renewcommand{\baselinestretch}{1}\selectfont\sf#3}\xspace}

\newcommand\mypar[1]{\medskip\noindent{\subsectionfont #1.}}
\let\subsubsection\mypar

\definecolor{my1}{cmyk}{0,.6,0,0}
\definecolor{my2}{cmyk}{.3,.0,.0,.0}

\newcommand{\efgame}{Ehrenfeucht-Fraïssé\xspace}
\newcommand\nat{\ensuremath{\mathbb{N}}\xspace}

\newcommand\Xs{\ensuremath{\mathcal{X}}\xspace}
\newcommand\Cs{\ensuremath{\mathcal{C}}\xspace}
\newcommand\Csgen[3]{\ensuremath{\Cs_{#1,#3}^{#2}}\xspace}
\newcommand\Cslev[1]{\ensuremath{\Cs_{#1}}\xspace}

\newcommand\Csik{\ensuremath{\Cs_i^k}\xspace}
\newcommand\Cstwo{\ensuremath{\Cs_2}\xspace}

\newcommand\Csi{\ensuremath{\Cs_i}\xspace}
\newcommand\Csikn{\ensuremath{\Cs_{i,n}^k}\xspace}
\newcommand\Csitwo{\ensuremath{\Cs_{i,2}}\xspace}
\newcommand\Cstwolen[1]{\ensuremath{\Cs_{2,#1}}\xspace}
\newcommand\Cstwon{\ensuremath{\Cstwolen{n}}\xspace}
\newcommand\Csin{\ensuremath{\Cs_{i,n}}\xspace}
\newcommand\fCgen[3]{\ensuremath{\fC_{#1,#3}^{#2}}\xspace}
\newcommand\fCik{\ensuremath{\fC_i^k}\xspace}
\newcommand\fCi{\ensuremath{\fC_i}\xspace}
\newcommand\fCtwo{\ensuremath{\fC_2}\xspace}

\newcommand\fCikn{\ensuremath{\fC_{i,n}^k}\xspace}
\newcommand\fCin{\ensuremath{\fC_{i,n}}\xspace}
\newcommand\fCtwolen[1]{\ensuremath{\fC_{2,#1}}\xspace}
\newcommand\fCtwon{\ensuremath{\fC_{2,n}}\xspace}

\newcommand\Gs{\ensuremath{\mathcal{G}}\xspace}

\newcommand\Ps{\ensuremath{\mathcal{P}}\xspace}
\newcommand\Es{\ensuremath{\mathcal{E}}\xspace}
\newcommand\Ss{\ensuremath{\mathcal{S}}\xspace}
\newcommand\Fs{\ensuremath{\mathcal{F}}\xspace}
\newcommand\Ts{\ensuremath{\mathcal{T}}\xspace}
\newcommand\Rs{\ensuremath{\mathcal{R}}\xspace}

\newcommand\Sat{\ensuremath{\mathord{\mathrm{Sat}}}}
\newcommand\ct{\ensuremath{\mathbb{T}}\xspace}
\newcommand\ctc{\ensuremath{\mathbb{T}_{\Cs}}\xspace}
\newcommand\cts{\ensuremath{\mathbb{T}_{\fC}}\xspace}
\newcommand\cs{\ensuremath{\mathbb{S}}\xspace}
\newcommand\crr{\ensuremath{\mathbb{U}}\xspace}

\newcommand\mat{\ensuremath{\mathscr{M}}\xspace}
\newcommand\mnat{\ensuremath{\mathscr{N}}\xspace}

\newcommand\pat{\ensuremath{\mathscr{P}}\xspace}

\newcommand{\dew}[1]{\ensuremath{\Delta_{#1}(<)}\xspace}
\newcommand{\dews}[1]{\ensuremath{\Delta_{#1}(<,+1,min,max)}\xspace}
\newcommand{\sic}[1]{\ensuremath{\Sigma_{#1}}\xspace}
\newcommand{\siw}[1]{\ensuremath{\Sigma_{#1}(<)}\xspace}
\newcommand{\siws}[1]{\ensuremath{\Sigma_{#1}(<,+1,min,max)}\xspace}
\newcommand{\pic}[1]{\ensuremath{\Pi_{#1}}\xspace}
\newcommand{\piw}[1]{\ensuremath{\Pi_{#1}(<)}\xspace}

\newcommand{\bsw}[1]{\ensuremath{\mathcal{B}\Sigma_{#1}(<)}\xspace}
\newcommand{\bsws}[1]{\ensuremath{\mathcal{B}\Sigma_{#1}(<,+1,min,max)}\xspace}

\newcommand{\dewu}{\ensuremath{\Delta_{1}(<)}\xspace}
\newcommand{\dewsu}{\ensuremath{\Delta_{1}(<,+1,min,max)}\xspace}

\newcommand{\siwu}{\ensuremath{\Sigma_{1}(<)}\xspace}
\newcommand{\siwsu}{\ensuremath{\Sigma_{1}(<,+1,min,max)}\xspace}

\newcommand{\piwu}{\ensuremath{\Pi_{1}(<)}\xspace}
\newcommand{\piwsu}{\ensuremath{\Pi_{1}(<,+1,min,max)}\xspace}

\newcommand{\bswu}{\ensuremath{\mathcal{B}\Sigma_{1}(<)}\xspace}
\newcommand{\bswsu}{\ensuremath{\mathcal{B}\Sigma_{1}(<,+1,min,max)}\xspace}

\newcommand{\decd}{\ensuremath{\Delta_{2}}\xspace}
\newcommand{\dewd}{\ensuremath{\Delta_{2}(<)}\xspace}
\newcommand{\dewsd}{\ensuremath{\Delta_{2}(<,+1,min,max)}\xspace}

\newcommand{\sicd}{\ensuremath{\Sigma_{2}}\xspace}
\newcommand{\siwd}{\ensuremath{\Sigma_{2}(<)}\xspace}
\newcommand{\siwsd}{\ensuremath{\Sigma_{2}(<,+1,min,max)}\xspace}

\newcommand{\picd}{\ensuremath{\Pi_{2}}\xspace}
\newcommand{\piwd}{\ensuremath{\Pi_{2}(<)}\xspace}
\newcommand{\piwsd}{\ensuremath{\Pi_{2}(<,+1,min,max)}\xspace}

\newcommand{\bscd}{\ensuremath{\mathcal{B}\Sigma_{2}}\xspace}
\newcommand{\bswd}{\ensuremath{\mathcal{B}\Sigma_{2}(<)}\xspace}
\newcommand{\bswsd}{\ensuremath{\mathcal{B}\Sigma_{2}(<,+1,min,max)}\xspace}

\newcommand{\dect}{\ensuremath{\Delta_{3}}\xspace}
\newcommand{\dewt}{\ensuremath{\Delta_{3}(<)}\xspace}
\newcommand{\dewst}{\ensuremath{\Delta_{3}(<,+1,min,max)}\xspace}

\newcommand{\sict}{\ensuremath{\Sigma_{3}}\xspace}
\newcommand{\siwt}{\ensuremath{\Sigma_{3}(<)}\xspace}
\newcommand{\siwst}{\ensuremath{\Sigma_{3}(<,+1,min,max)}\xspace}

\newcommand{\pict}{\ensuremath{\Pi_{3}}\xspace}
\newcommand{\piwt}{\ensuremath{\Pi_{3}(<)}\xspace}
\newcommand{\piwst}{\ensuremath{\Pi_{3}(<,+1,min,max)}\xspace}

\newcommand{\bswt}{\ensuremath{\mathcal{B}\Sigma_{3}(<)}\xspace}
\newcommand{\bswst}{\ensuremath{\mathcal{B}\Sigma_{3}(<,+1,min,max)}\xspace}

\newcommand{\dewi}{\ensuremath{\Delta_{\lowercase{i}}(<)}\xspace}

\newcommand{\sici}{\ensuremath{\Sigma_{\lowercase{i}}}\xspace}
\newcommand{\siwi}{\ensuremath{\Sigma_{\lowercase{i}}(<)}\xspace}
\newcommand{\siwsi}{\ensuremath{\Sigma_{\lowercase{i}}(<,+1,min,max)}\xspace}

\newcommand{\piwi}{\ensuremath{\Pi_{\lowercase{i}}(<)}\xspace}

\newcommand{\bsci}{\ensuremath{\mathcal{B}\Sigma_{\lowercase{i}}}\xspace}
\newcommand{\bswi}{\ensuremath{\mathcal{B}\Sigma_{\lowercase{i}}(<)}\xspace}
\newcommand{\bswsi}{\ensuremath{\mathcal{B}\Sigma_{\lowercase{i}}(<,+1,min,max)}\xspace}

\newcommand{\mso}{\ensuremath{\textup{MSO}}\xspace}
\newcommand{\msow}{\ensuremath{\textup{MSO}(<)}\xspace}
\newcommand{\fo}{\ensuremath{\textup{FO}}\xspace}
\newcommand{\fow}{\ensuremath{\textup{FO}(<)}\xspace}
\newcommand{\fows}{\ensuremath{\textup{FO}(<,+1,min,max)}\xspace}

\newcommand\sieq[2]{\ensuremath{\lesssim^{#1}_{#2}}\xspace}
\newcommand\ksieq[1]{\sieq{k}{#1}}

\newcommand\bceq[2]{\ensuremath{\cong^{#1}_{#2}}\xspace}
\newcommand\kbceq[1]{\bceq{k}{#1}}

\newcommand\gmo{\ensuremath{\geqslant}\xspace}
\newcommand\lmo{\ensuremath{\leqslant}\xspace}

\let\geq\geqslant

\newcommand\Sep{\ensuremath{\mathcal{C}}\xspace}

\newcommand\content[1]{\ensuremath{\contentmorphism(#1)}}

\newcommand\contentmorphism{\ensuremath{\textsf{alph}}}

\newcommand\valc[1]{\ensuremath{\textsf{val}_{\Cs}(#1)\xspace}}
\newcommand\vals[1]{\ensuremath{\textsf{val}_{\fC}(#1)\xspace}}

\newcommand\inst[1]{\ensuremath{\llfloor{#1}\rrfloor\xspace}}

\newcommand\chain{chain\xspace}
\newcommand\qchain[1]{\ensuremath{\sic{#1}}-chain\xspace}
\newcommand\chains{chains\xspace}
\newcommand\qchains[1]{\ensuremath{\sic{#1}}-chains\xspace}
\newcommand\Chain{Chain\xspace}

\newcommand\Chains{Chains\xspace}
\newcommand\qChains[1]{\ensuremath{\sic{#1}}-Chains\xspace}

\newcommand\qpchains[2]{\ensuremath{\sic{#1}[#2]}-chains\xspace}

\newcommand\ichain{\qchain{\lowercase{i}}}

\newcommand\dchain{\qchain{2}}
\newcommand\ichains{\qchains{\lowercase{i}}}

\newcommand\dchains{\qchains{2}}

\newcommand\iChains{\qChains{\lowercase{i}}}

\newcommand\dChains{\qChains{2}}

\newcommand\ikchains{\qpchains{i}{k}}

\newcommand\jun{juncture\xspace}
\newcommand\juns{junctures\xspace}
\newcommand\Jun{Juncture\xspace}
\newcommand\Juns{Junctures\xspace}
\newcommand\qjun[1]{\ensuremath{\sic{#1}}-juncture\xspace}
\newcommand\qjuns[1]{\ensuremath{\sic{#1}}-junctures\xspace}

\newcommand\qJuns[1]{\ensuremath{\sic{#1}}-Junctures\xspace}
\newcommand\qpjun[2]{\ensuremath{\sic{#1}[#2]}-juncture\xspace}
\newcommand\qpjuns[2]{\ensuremath{\sic{#1}[#2]}-junctures\xspace}

\newcommand\ijun{\qjun{\lowercase{i}}}
\newcommand\ijuns{\qjuns{\lowercase{i}}}

\newcommand\iJuns{\qJuns{\lowercase{i}}}
\newcommand\djun{\qjun{2}}
\newcommand\djuns{\qjuns{2}}

\newcommand\ikjun{\qpjun{\lowercase{i}}{k}}
\newcommand\ikjuns{\qpjuns{\lowercase{i}}{k}}

\newcommand\dkjun{\qpjun{2}{k}}

\newcommand\fI{\ensuremath{\mathcal D}\xspace}

\newcommand\fC{\ensuremath{\mathcal J}\xspace}
\newcommand\fS{\ensuremath{\mathcal S}\xspace}
\newcommand\fR{\ensuremath{\mathcal R}\xspace}

\newcommand\fT{\ensuremath{\mathcal T}\xspace}

\newcommand\fX{\ensuremath{\mathcal X}\xspace}

\newcommand\fg{\ensuremath{g}\xspace}

\newcommand\wfA{\ensuremath{\Abb_\alpha}\xspace}

\newcommand\Abb{\ensuremath{\mathbb{A}}\xspace}

\newcommand\Lbb{\ensuremath{\mathbb{L}}\xspace}

\newcommand\wbb{\ensuremath{\mathbbm{w}}\xspace}

\DeclareMathOperator{\downclos}{\downarrow}

\tikzstyle{nor}=[minimum size=0.35cm,draw,rounded rectangle,inner sep=2pt]
\tikzstyle{nod}=[minimum size=0.35cm,draw,circle,inner sep=2pt]
\tikzstyle{nok}=[minimum size=0.45cm,draw,circle,inner sep=0pt]
\tikzstyle{nof}=[minimum size=0.35cm,draw,circle,double,double
distance=1pt]
\tikzstyle{port}=[minimum size=0.35cm,draw,thick,rectangle,inner sep=2pt]
\tikzstyle{nop}=[minimum size=0.35cm,draw,thick,rectangle,inner sep=1pt,rotate=90]

\tikzstyle{nol}=[minimum size=0.35cm,draw,rounded rectangle,inner sep=1pt,rotate=90]

\tikzstyle{ar}=[line width=0.5pt,->,double]
\tikzstyle{siar}=[line width=1.5pt,->]
\tikzstyle{ars}=[line width=1pt,->,double]

\newtheorem{fact}[theorem]{Fact}

\providecommand*{\donothing}[1]{}

\def\formatline{}
\def\runningfoot{}
\def\firstfoot{}
\makeatletter
\def\@evenfoot{}
\def\@oddfoot{}
\makeatother
\begin{document}

\markboth{T. Place and M. Zeitoun}{Going Higher in First-Order Quantifier Alternation Hierarchies on Words}

\title{Going Higher in First-Order Quantifier Alternation Hierarchies on Words}
\author{Thomas Place and Marc Zeitoun \affil{LaBRI, University of Bordeaux, France}\texttt{firstname.lastname@labri.fr}}

\keywords{First-order logic, Regular languages, Decidable
  characterization, Membership Problem, Separation problem, Quantifier alternation, Logical hierarchies, Dot-depth
  hierarchy, Straubing-Thérien hierarchy.}

\begin{abstract}
  We investigate quantifier alternation hierarchies in first-order
  logic on finite words. Levels in these hierarchies are defined by counting the
  number of quantifier alternations in formulas. We prove that one can
  decide membership of a regular language in the levels \bscd (finite boolean
  combinations of formulas having only one alternation) and \sict
  (formulas having only two alternations and beginning with an existential
  block). Our proofs work by considering a deeper problem, called
  separation, which, once solved for lower levels, allows us to solve
  membership for higher~levels.
\end{abstract}

\maketitle

\section{Introduction}
\label{sec:intro}
The connection between logic and automata theory is well known and has a
fruitful history in theoretical computer science. It was first observed when
\citeN{BuchiMSO}, \citeN{ElgotMSO} and \citeN{TrakhMSO} independently proved
that regular languages of finite words are exactly languages that can be
defined by a monadic second-order logic (\mso) sentence. Since then, many
efforts have been devoted to the investigation and understanding of the
expressive power of relevant fragments of \mso. In this field, the yardstick
result is often to prove a \emph{decidable characterization}, \emph{i.e.}, to
design an algorithm which, given as input a regular language, decides whether
it can be defined within the fragment under investigation. This decision
problem is called the \emph{membership problem}. More than the algorithm
itself, the main motivation for solving it is the insight given by its proof.
Indeed, in order to prove a decidable characterization, one has to consider
and understand \emph{all} properties that can be expressed in the~fragment.

\smallskip The most prominent fragment of \mso is first-order logic (\fow, or
\fo for short) equipped with a predicate ``$<$'' for the linear order. This
logic has first been investigated on finite words by \citeN{mnpfo}, who showed
that a language is \fo definable iff it is \emph{star-free}, that is, iff it
can be defined from singleton languages using boolean operations and
concatenation (but \emph{not} the Kleene star, hence the name). As such, this result
just amounts to a simple syntactic translation, which does not provide any
insight on the expressive power of first-order logic. However, together with
an earlier result from \citeN{sfo}, it yields a decidable characterization.
Indeed, Schützenberger's Theorem states that a regular language is star-free
if and only if its syntactic monoid is \emph{aperiodic}. The syntactic monoid
is a finite algebraic structure that can be effectively computed from any
representation of the language. Moreover, aperiodicity can be rephrased as an
equation that needs to be satisfied by all elements of the monoid. Therefore,
Schützenberger's Theorem together with McNaughton-Papert's result indeed
entails decidability of first-order definability.

\mypar{Quantifier Alternation.} Schützenberger's proof additionally provides
an algorithm which, given a regular language, outputs a first-order sentence
(of course when the input language is first-order definable). However, this
sentence may be unnecessarily complicated. The next natural step consists in
requiring the output sentence to be ``as simple as possible''. To make this
question precise, one needs a meaningful notion of complexity.

\smallskip The most appropriate parameter for classifying first-order definable
languages according to the difficulty of defining them is their
\emph{quantifier alternation}. The quantifier alternation \emph{of a formula}
is simply the maximal number of switches between blocks of $\exists$
quantifiers and blocks of $\forall$ quantifiers in its prenex normal form. The
quantifier alternation \emph{of a language} definable in \fo is the smallest quantifier
alternation of a first-order sentence that defines~it. Observe that the quantifier
alternation \emph{of a language} is, like first-order definability, a semantic
notion (in contrast with the quantifier alternation \emph{of a formula}, which
is a syntactic notion). This explains why it is not straightforward to compute it
from a representation of the language.

\smallskip It is intuitive that formulas involving several alternations are
difficult to grasp---one usually uses only few of them to state mathematical
properties. This intuition is supported by results showing that, indeed, this
parameter is meaningful, \emph{i.e.}, that languages of high quantifier
alternation are ``hard'' to deal with. The algorithmic treatment of
first-order formulas involves an unavoidable non-elementary lower
bound~\cite{Stockmeyer:1973:WPR:800125.804029,Stockmeyer-phd,DBLP:conf/dagstuhl/Reinhardt01}.
This is the case for instance for the satisfiability problem. Likewise, the
number of states of the minimal automaton equivalent to an \fo formula may be
non-elementarily large in the size of the formula. This blowup is due to quantifier
alternation, since restricting these problems to formulas of bounded
quantifier alternation yields elementary decision procedures.

\smallskip This motivates the investigation of what can be expressed with a \emph{fixed} number of quantifier alternations, and, already importantly, with few of them. This is
what we do in this paper: we investigate the hierarchy inside \fo obtained by classifying languages according to their quantifier alternation. More precisely, the hierarchy
involves the classes \siwi, \bswi and \dewi defined as follows:

\begin{itemize}
\item an \fow formula is \siwi if its prenex normal form has $(i-1)$
  quantifier alternations and starts with a block of existential quantifiers,
  or if it has strictly less than $(i-1)$
  quantifier alternations.
  A language is \siwi if it can be defined by a \siwi sentence.
\item a formula is \bswi if it is a finite boolean combination of \siwi
  formulas. A language is \bswi if it can be defined by a \bswi sentence.
\item Finally, a language is \dewi if it can be defined by both a \siwi
sentence and the negation of a \siwi sentence. Note that there is no notion of a ``\dewi formula''.
\end{itemize}
The quantifier alternation hierarchy is known to be strict:
\[
  \dewi \subsetneq \siwi \subsetneq {\bswi} \subsetneq \dew{i+1}.
\]
This well-known hierarchy thus defines a complexity measure of first-order
definable languages: complex ones are those requiring several quantifier
alternations.

\medskip Another motivation for investigating this hierarchy is its ties with
two other famous hierarchies in formal language theory, defined in terms of
regular expressions. Roughly speaking, levels in both of these hierarchies
count the number of alternations between boolean operations and concatenation
product that are necessary to express a language (recall that, by
McNaughton-Papert's Theorem, every first-order definable language can be built
from singleton languages using union, concatenation and boolean operations).
In the first of these hierarchies, the \emph{Straubing-Thérien
  hierarchy}~\cite{StrauConcat,TheConcat}, level~$i$ exactly corresponds to
the class \bswi, as shown by \citeN{PPOrder}. In the second one, the
\emph{dot-depth hierarchy}, which was actually defined
earlier by~\citeN{BrzoDot}, level~$i$ corresponds to augmenting the logic
\bswi with a predicate for the successor relation, as shown by
\citeN{ThomEqu}. These correspondences show that proving
decidability of the membership problem for \bswd immediately entails its
decidability for level 2 in the Straubing-Thérien hierarchy, but also in the
dot-depth hierarchy, thanks to a reduction due to \citeN{StrauVD}. We refer
the reader to Section~\ref{sec:history} for details.

\smallskip Many efforts have been devoted to finding decidable
characterizations for levels in the quantifier alternation hierarchy. Despite these
efforts however, only the lower ones are known to be decidable. The~class
\bswu consists exactly of all piecewise testable languages, \emph{i.e.}, such
that membership of a word only depends on its scattered subwords up to a
fixed~size. These languages were characterized by~\citeN{simon75} as those
whose syntactic monoid is $\mathcal{J}$-trivial. A decidable characterization
of \siwd{}---hence of \dewd as well---was obtained
by~Arfi~\citeyear{arfi87,Arfi_1991}, a problem revisited and clarified by Pin
and Weil~\citeyear{pwdelta_conf,pwdelta}, who also set up a generic algebraic
framework to work with. For~\dewd, the literature is very rich, see the survey
by \citeN{Tesson02diamondsare}. For example, the \dewd definable languages are
exactly the ones definable in the two-variable restriction of
\fow~\cite{twfodeux}. These are also the languages whose syntactic monoid
belongs to the class~$\textup{\sf DA}$, as shown again by Pin and
Weil~\citeyear{pwdelta_conf,pwdelta} (see also~\cite{schul}). For higher
levels in the hierarchy, getting decidable characterizations remained a major
open problem. In particular, the case of \bswd has a very abundant history and
a series of combinatorial, logical, and algebraic conjectures have been
proposed over the years. We refer to Section~\ref{sec:history} and to several
surveys cited in this section for a bibliography. So far, the only known
effective result was partial, working only when the alphabet is of
size~2~\cite{StrauDD2}.

\mypar{Contributions.} In this paper, we establish decidable characterizations
for the fragments \bswd, \dewt and \siwt of first-order logic. These new
results are based on a deeper decision problem than membership: the \emph{separation
problem}. Fix a class \Sep of languages. The \Sep-separation problem amounts to
deciding whether, given two input regular languages, there exists a third
language in \Sep containing the first language while being disjoint from the
second one. Solving the \Sep-separation problem is more general than obtaining
a decidable characterization for the class \Sep. Indeed, since regular
languages are effectively closed under complement, testing membership in \Sep
can be achieved by testing whether the input is \Sep-separable from its
complement. While this reduction immediately transfers decision procedures for
one problem to the other, this is not our primary motivation for looking at
separation. Although intrinsically more challenging, a solution to the
separation problem requires more understanding than just getting a decidable
characterization. This understanding for a given fragment can then be
exploited in order to obtain decidable characterizations for extensions built
on top of this fragment.

\smallskip Historically, the separation problem for regular languages was
first investigated as a special case of a deep problem in semigroup theory,
the \emph{pointlike problem}, solved for several 
cases by relying on purely algebraic and topological
arguments~\cite{henckell:1988,DBLP:journals/ijac/HenckellRS10a,AZ97-J}. It was only
identified as a variant of the separation problem by~\citeN{MR1709911}.
Recently, a research effort has been made to investigate this problem from a
radically different perspective, with the aim of finding new and self-contained proofs
relying on elementary ideas and notions from language theory only. Such proofs
were obtained for several results already known in the algebraic
framework~\cite{martens,pvzmfcs13,pzfo,pvzltt,PZ:FO-Sep16}. This paper is a
continuation of this effort for classes that were not solved even in the
algebraic setting: we solve the separation problem for \siwd, and we use our
solution as a basis to obtain decidable characterizations for the classes
\bswd, \dewt and~\siwt.

\smallskip Our proof works as follows: given two regular languages, one can
easily construct a morphism $\alpha$ from $A^*$ into a finite monoid $M$ that
recognizes both languages. We then design an algorithm that computes, inside
the finite monoid $M$, enough \sicd-related information to answer the
\siwd-separation question for \emph{any} pair of languages that are recognized
by~$\alpha$. It turns out that it is also possible to use this information to
obtain decidability of \dewt, \siwt and \bswd (though this last
characterization is much more difficult). This information amounts to the
notion of \emph{\dchain}, our main tool in the paper. A \dchain is an
\emph{ordered sequence} $s_1,\dots,s_n \in M$ that witnesses a property of
$\alpha$ with respect to\ \siwd.
Let us give some intuition in the case $n = 2$---which is enough to make the
link with \sicd-separation. A sequence $s_1,s_2$ of elements of~$M$ is a
\dchain if any \siwd language containing all words in $\alpha^{-1}(s_1)$
intersects $\alpha^{-1}(s_2)$. In terms of separation, this means that
$\alpha^{-1}(s_1)$ is \emph{not} separable from $\alpha^{-1}(s_2)$ by a
$\siwd$ definable language. 
This notion can actually be extended to any level of the hierarchy.

\medskip
This paper contains three main separate, new and nontrivial results:
\begin{enumerate}[label=(\arabic*),ref=(\arabic*)]
\item\label{item:1} An algorithm to compute \dchains---hence \siwd-separability is decidable.
\item\label{item:2} A transfer result showing that an algorithm to compute
  \ichains of length 2 entails a decidable characterization of $\siw{i+1}$. In
  particular, by~\ref{item:1}, membership in \siwt is decidable. Decidability
  of \piwt, the dual of \siwt, and of \dewt are then~immediate.
\item\label{item:3} A decidable characterization of~\bswd.
\end{enumerate}
For \ref{item:1}, computing \dchains is achieved using a fixpoint algorithm
that starts with trivial \dchains such as $s,s,\dots,s$, and iteratively
computes more \dchains until a fixpoint is reached. For our technique to work,
we actually have to consider a notion slightly more general than \dchains. The
completeness proof of this algorithm relies on the Factorization Forest
Theorem of Simon~\citeyear{simonfacto}. This is not surprising (even though
one can actually bypass its use), as the link between this
theorem and the quantifier alternation hierarchy was already observed
by~\citeN{pwdelta} and \citeN{bfacto}.

\smallskip
For~\ref{item:2}, we establish a characterization of \siwt in terms of an
equation on the syntactic monoid of the language. This equation is
parametrized by the set of \dchains of length~$2$. In other words, we use
\dchains to abstract an infinite set of equations into a single one. The proof
relies again on the Factorization Forest Theorem of Simon~\cite{simonfacto}
and is actually generic to all levels in the hierarchy. This means that for
any level $i$, we define a notion of \ichain and we characterize \siw{i+1}
using an equation parametrized by \ichains of length~$2$. However,
decidability of \siw{i+1} depends on our ability to compute all \ichains of
length~$2$, which we can only do for $i =2$.

\smallskip
Finally for~\ref{item:3}, the decidable characterization of \bswd is the most
difficult result of the paper. As for \siwt, it is presented by two equations
parametrized by \dchains (of length $2$ and $3$). However, the
characterization is this time specific to the case $i = 2$. This is because
most of our proof relies on a careful analysis of our algorithm that computes
\dchains, which only works for $i =2$. The equations share surprising
similarities with the ones used by \citeN{bpopen} to characterize a totally
different formalism: boolean combinations of open sets of infinite trees.
In~\cite{bpopen} also, the authors present their characterization as a set of
equations parametrized by a notion of ``\chain'' for open sets of infinite
trees (although their ``\chains'' are not explicitly identified as a
separation relation). Since the formalisms are of different nature, the way
these \chains and our \dchains are constructed are completely independent,
which means that the proofs are also mostly independent. However, once the
construction analysis of \chains has been done, several combinatorial
arguments used to make the link with equations are analogous. In particular,
we reuse and adapt definitions from~\cite{bpopen} to present these
combinatorial arguments in our proof. One could say that the proofs are both
(very different) setups to apply similar combinatorial arguments in the end.

\smallskip Our results are shown using the ordering relation `$<$' on
positions as the only numerical predicate of the signature in the logic. In
full first-order logic, one can define other natural numerical
predicates, such as the first and last positions, as well as the successor
relation. However, defining these predicates requires an additional
quantification. It is known that enriching the signature with these predicates
indeed increases the expressiveness of each fragment in the quantifier
alternation hierarchy. This yields another hierarchy inside first-order logic,
that has also been investigated in the literature. In particular, it has been
shown by~\citeN{ThomEqu} to correspond to the so-called \emph{dot-depth}
hierarchy defined by~\citeN{BrzoDot} in terms of regular constructs
needed to build a star-free language. In Section~\ref{sec:transfer}, we present
already known results to show that all decidability statements obtained for
the original hierarchy can be lifted to the hierarchy where the additional
predicates are allowed. This works both for decidable
characterizations~\cite{StrauVD,pzsucc} as well as for
separation~\cite{Steinberg:delay-pointlikes:2001,pzsucc}.

\mypar{Organization.} 
Sections~\ref{sec:words} and~\ref{sec:history} are devoted to the presentation
of the problem we investigate. In Section~\ref{sec:words}, we define the
quantifier alternation hierarchies and precisely state this problem.
Section~\ref{sec:history} presents an outline of the rich history about these
problems, viewed from different perspectives.

Then, in Sections~\ref{sec:tools} and~\ref{sec:chains}, we develop the
machinery necessary to the statements and to the proofs of our results.
Section~\ref{sec:tools} is devoted to the presentation of well-known,
classical tools such as \efgame games, monoids and Simon's Factorization
Forest Theorem, while Section~\ref{sec:chains} introduces a new tool, specific
to this paper: \ichains.

The remaining sections present and prove our results. In
Section~\ref{sec:generic}, we reduce the membership and separation problems
for all levels in the hierarchy to the problem of computing \ichains. In the
following sections, we then prove that these problems can be solved for
specific levels. In Section~\ref{sec:comput} we obtain a solution to
separation for \siwd and to membership for \siwt, \piwt and \dewt. Then, in
Section~\ref{sec:caracbc}, we obtain a solution to membership for \bswd.
Sections~\ref{app:ctrees} to~\ref{app:width} are then devoted to the difficult
proof of the decidable characterization of \bswd. Finally, in the last
section, Section~\ref{sec:transfer}, we lift up our results to the hierarchy
with successor, using previously known transfer results.

\smallskip
This paper is the full version of~\cite{pzqalt}.

\section{Quantifier Alternation Hierarchies}
\label{sec:words}
As explained in the introduction, we study two decisions problems, called
membership and separation, to investigate two famous
hierarchies of classes of languages. In this section, we precisely
define these hierarchies and decision problems. Note that the section is
devoted to definitions only. We shall also present
in~Section~\ref{sec:history} the history about these hierarchies.

\smallskip The section is organized in two parts. We begin by giving a logical
definition of our two hierarchies: they classify first-order definable
languages by counting the number of quantifier alternations that are needed
for defining these languages. Equivalent combinatorial definitions
in terms of star-free languages will be presented in
Section~\ref{sec:history}. In the second part, we define the membership
problem and the separation problem.

\subsection{Quantifier Alternation Hierarchies of First-Order Logic}
\label{sec:quant-altern-hier}

Throughout the paper, we assume fixed a finite alphabet $A$. We denote by
$A^{*}$ the set of all words over $A$ (including the empty word 
$\varepsilon$) and by $A^+$ the set of all nonempty words over $A$. If
$u,v \in A^*$ are words over $A$, we denote by $u \cdot v$ or $uv$ the word
obtained by concatenation of $u$ and $v$ and by \content{u} the alphabet of
$u$, \emph{i.e.}, the smallest subset $B$ of $A$ such that $u \in B^*$. A
\emph{language over $A$} is a subset of $A^*$. In this paper we work with
regular languages. These languages have several equivalent characterizations,
as they can be defined by either:
\begin{itemize}
\item \emph{monadic second-order logic},
\item \emph{finite automata},
\item \emph{regular expressions},
\item \emph{finite monoids}.
\end{itemize}
The two hierarchies we investigate in the paper are contained within a strict subclass of regular languages that we define now: the class of first-order definable languages.

\mypar{First-Order Logic.} We view words as logical structures made of
a sequence of positions. Each position has a label in the alphabet~$A$ and can
be quantified. We denote by `$<$' the linear order over the positions. We work
with first-order logic, \fow, using the following predicates:
\begin{itemize}
\item for each $a \in A$, a unary predicate $P_a$ that selects positions labeled with an $a$.
\item a binary predicate `$<$' for the linear order.
\end{itemize}
To every first-order sentence $\varphi$, one can associate the language
$\{w \in A^* \mid w \models \varphi\}$ of words that satisfy $\varphi$. For
instance, the sentence $\exists x P_a(x)$ defines the language of all words
having at least one `$a$'. Hence, \fow defines a class of languages: the class
of all languages that can be defined by an \fow sentence. For the sake of
simplifying the presentation, we will abuse notation and use \fow to denote
both the logic and the associated class of languages.

\mypar{Order Hierarchy.} One classifies first-order formulas by
counting the number of alternations between $\exists$ and $\forall$
quantifiers in the prenex normal form of the formula. For $i \in \nat$, a
formula is said to be \siw{i} (resp. \piw{i}) if its prenex normal form has
either
\begin{itemize}
\item \emph{exactly} $(i -1)$ quantifier alternations (\emph{i.e.}, exactly
  $i$ quantifier blocks) and starts with an $\exists$~quantifier (resp.\ with
  a $\forall$ quantifier), or
\item \emph{strictly less} than $(i -1)$ quantifier
  alternations (\emph{i.e.}, strictly less than $i$ quantifier
  blocks).
\end{itemize}
For example, a formula whose prenex normal form is
\[
  \forall x_1 \forall x_2 \exists x_3 \forall x_4
  \ \varphi(x_1,x_2,x_3,x_4) \quad \text{(with $\varphi$ quantifier-free)}
\]
\noindent
is \piwt. Observe that a \piw{i} formula is by definition the negation of a
\siw{i} formula. Finally, a \bsw{i} formula is a finite boolean combination of
\siw{i} formulas. As for full first-order logic, we will abuse notations and
use \siw{i}, \piw{i} and \bsw{i} to denote both the logics and the associated
classes of languages. Finally, we denote by \dew{i} the class of languages
that can be defined by \emph{both} a \siw{i} and a \piw{i}
formula\footnote{Note that, strictly speaking, \dew{i} is not a logic: there
  is no notion of a ``\dew{i} formula.''}. It is known~\cite{PPOrder} that
this gives a strict infinite hierarchy of classes of languages as represented
in Figure~\ref{fig:hiera3}. In the paper, we call this hierarchy the
\emph{order hierarchy}. It turns out that quantifier alternation can be used
to define another natural hierarchy within first-order logic, which we now
describe.

\mypar{Enriched Hierarchy.} Observe that in full first-order logic,
several natural relations can be defined using the linear order:
\begin{itemize}
\item Position $x$ is the first one: ${min}(x)\stackrel{\text{def}}{=} \forall y\ \neg (y < x)$.
\item Position $x$ is the last one: ${max}(x)\stackrel{\text{def}}{=} \forall y\ \neg (x < y)$.
\item Position $y$ is the successor of position $x$: $(y = x+1)\stackrel{\text{def}}{=} x < y \wedge \neg (\exists z\ x< z \wedge z <y)$.
\end{itemize}
Therefore, adding these relations as predicates in the signature of first-order
logic does not increase its expressive power: \fow and the enriched logic
\fows define the same class of languages. However, observe that replacing the
predicates $min,max$ or $+1$ with their definitions may increase the
quantifier alternation of the formula. For example,
\[
  \begin{array}{ll}
    \exists x \exists y\quad y = x + 1 \wedge P_a(x) \wedge P_b(y) &
                                                                   \text{has no alternation, while}\\[1ex]
    \exists x \exists y\quad (x < y \wedge \neg (\exists z\ x< z \wedge z
    <y)) \wedge P_a(x) \wedge P_b(y) & \text{has one alternation.}

  \end{array}
\]
Hence, it is not immediate whether fragments of the order hierarchy have the
same expressive power as their enriched counterpart. In fact, it is known that
the predicate `$+1$' cannot be freely defined in any logic of the order
hierarchy. Hence, we get a second hierarchy, also depicted in
Figure~\ref{fig:hiera3}. That this hierarchy is also strict follows from the
work of \citeN{BroKnaStrict} and Thomas~\citeNN{ThomEqu,ThomStrict}. In the
paper, we call it the \emph{enriched hierarchy}.

\tikzstyle{non}=[inner sep=1pt]
\tikzstyle{tag}=[draw,fill=white,sloped,circle,inner sep=1pt]
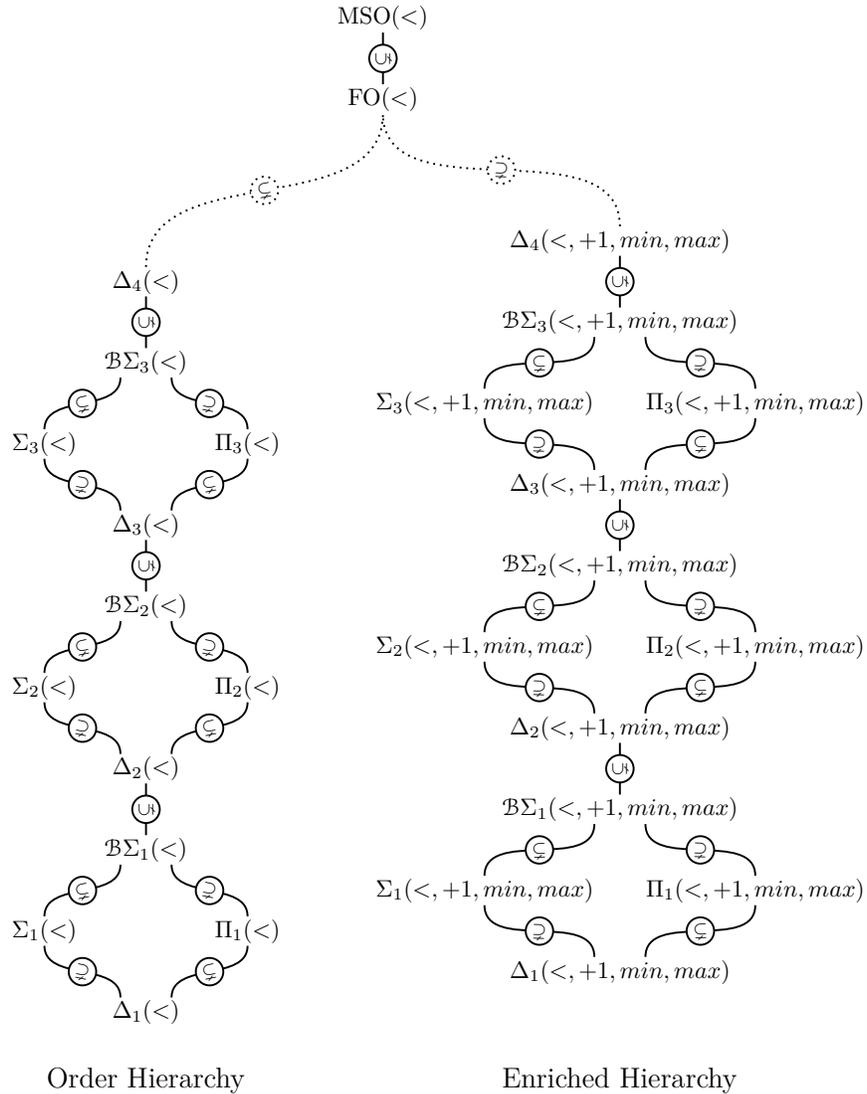
\begin{figure}[ht]
  \begin{center}
    \scalebox{.9}{
      \begin{tikzpicture}
        \node[non] at (0.0,-1.0) {{\large Order Hierarchy}};
        \node[non] at (7.0,-1.0) {{\large Enriched Hierarchy}};

        \node[non] (d1) at (0.0,0.0) {\dewu};
        \node[non] (s1) at ($(d1)+(-1.5,1.2)$) {\siwu};
        \node[non] (p1) at ($(d1)+(1.5,1.2)$) {\piwu};
        \node[non] (b1) at ($(d1)+(0.0,2.4)$) {\bswu};

        \node[non] (d2) at (0.0,3.6) {\dewd};
        \node[non] (s2) at ($(d2)+(-1.5,1.2)$) {\siwd};
        \node[non] (p2) at ($(d2)+(1.5,1.2)$) {\piwd};
        \node[non] (b2) at ($(d2)+(0.0,2.4)$) {\bswd};

        \node[non] (d3) at (0.0,7.2) {\dewt};
        \node[non] (s3) at ($(d3)+(-1.5,1.2)$) {\siwt};
        \node[non] (p3) at ($(d3)+(1.5,1.2)$) {\piwt};
        \node[non] (b3) at ($(d3)+(0.0,2.4)$) {\bswt};

        \node[non] (d4) at (0.0,10.8) {\dew{4}};

        \begin{scope}[xshift=7cm,yshift=0.6cm]

          \node[non] (d1b) at (0.0,0.0) {\dewsu};
          \node[non] (s1b) at ($(d1b)+(-2.0,1.2)$) {\siwsu};
          \node[non] (p1b) at ($(d1b)+(2.0,1.2)$) {\piwsu};
          \node[non] (b1b) at ($(d1b)+(0.0,2.4)$) {\bswsu};

          \node[non] (d2b) at (0.0,3.6) {\dewsd};
          \node[non] (s2b) at ($(d2b)+(-2.0,1.2)$) {\siwsd};
          \node[non] (p2b) at ($(d2b)+(2.0,1.2)$) {\piwsd};
          \node[non] (b2b) at ($(d2b)+(0.0,2.4)$) {\bswsd};

          \node[non] (d3b) at (0.0,7.2) {\dewst};
          \node[non] (s3b) at ($(d3b)+(-2.0,1.2)$) {\siwst};
          \node[non] (p3b) at ($(d3b)+(2.0,1.2)$) {\piwst};
          \node[non] (b3b) at ($(d3b)+(0.0,2.4)$) {\bswst};

          \node[non] (d4b) at (0.0,10.8) {\dews{4}};

        \end{scope}

        \draw[thick] (d1.150) to [out=90,in=-90] node[tag] {\scriptsize
          $\supsetneq$} (s1.south);
        \draw[thick] (d1.30) to [out=90,in=-90] node[tag] {\scriptsize
          $\subsetneq$} (p1.south);
        \draw[thick] (s1.north) to [out=90,in=-90] node[tag] {\scriptsize
          $\subsetneq$} (b1.-150);
        \draw[thick] (p1.north) to [out=90,in=-90] node[tag] {\scriptsize
          $\supsetneq$} (b1.-30);
        \draw[thick] (b1.north) to [out=90,in=-90] node[tag] {\scriptsize
          $\subsetneq$} (d2.south);

        \draw[thick] (d2.150) to [out=90,in=-90] node[tag] {\scriptsize
          $\supsetneq$} (s2.south);
        \draw[thick] (d2.30) to [out=90,in=-90] node[tag] {\scriptsize
          $\subsetneq$} (p2.south);
        \draw[thick] (s2.north) to [out=90,in=-90] node[tag] {\scriptsize
          $\subsetneq$} (b2.-150);
        \draw[thick] (p2.north) to [out=90,in=-90] node[tag] {\scriptsize
          $\supsetneq$} (b2.-30);
        \draw[thick] (b2.north) to [out=90,in=-90] node[tag] {\scriptsize
          $\subsetneq$} (d3.south);

        \draw[thick] (d3.150) to [out=90,in=-90] node[tag] {\scriptsize
          $\supsetneq$} (s3.south);
        \draw[thick] (d3.30) to [out=90,in=-90] node[tag] {\scriptsize
          $\subsetneq$} (p3.south);
        \draw[thick] (s3.north) to [out=90,in=-90] node[tag] {\scriptsize
          $\subsetneq$} (b3.-150);
        \draw[thick] (p3.north) to [out=90,in=-90] node[tag] {\scriptsize
          $\supsetneq$} (b3.-30);
        \draw[thick] (b3.north) to [out=90,in=-90] node[tag] {\scriptsize
          $\subsetneq$} (d4.south);

        \draw[thick] (d1b.150) to [out=90,in=-90] node[tag] {\scriptsize
          $\supsetneq$} (s1b.south);
        \draw[thick] (d1b.30) to [out=90,in=-90] node[tag] {\scriptsize
          $\subsetneq$} (p1b.south);
        \draw[thick] (s1b.north) to [out=90,in=-90] node[tag] {\scriptsize
          $\subsetneq$} (b1b.-150);
        \draw[thick] (p1b.north) to [out=90,in=-90] node[tag] {\scriptsize
          $\supsetneq$} (b1b.-30);
        \draw[thick] (b1b.north) to [out=90,in=-90] node[tag] {\scriptsize
          $\subsetneq$} (d2b.south);

        \draw[thick] (d2b.150) to [out=90,in=-90] node[tag] {\scriptsize
          $\supsetneq$} (s2b.south);
        \draw[thick] (d2b.30) to [out=90,in=-90] node[tag] {\scriptsize
          $\subsetneq$} (p2b.south);
        \draw[thick] (s2b.north) to [out=90,in=-90] node[tag] {\scriptsize
          $\subsetneq$} (b2b.-150);
        \draw[thick] (p2b.north) to [out=90,in=-90] node[tag] {\scriptsize
          $\supsetneq$} (b2b.-30);
        \draw[thick] (b2b.north) to [out=90,in=-90] node[tag] {\scriptsize
          $\subsetneq$} (d3b.south);

        \draw[thick] (d3b.150) to [out=90,in=-90] node[tag] {\scriptsize
          $\supsetneq$} (s3b.south);
        \draw[thick] (d3b.30) to [out=90,in=-90] node[tag] {\scriptsize
          $\subsetneq$} (p3b.south);
        \draw[thick] (s3b.north) to [out=90,in=-90] node[tag] {\scriptsize
          $\subsetneq$} (b3b.-150);
        \draw[thick] (p3b.north) to [out=90,in=-90] node[tag] {\scriptsize
          $\supsetneq$} (b3b.-30);
        \draw[thick] (b3b.north) to [out=90,in=-90] node[tag] {\scriptsize
          $\subsetneq$} (d4b.south);

        \node[non] (fo) at (3.5,13.5) {\fow};
        \node[non] (mso) at (3.5,14.7) {\msow};

        \draw[thick] (fo.north) to [out=90,in=-90] node[tag] {\scriptsize
          $\subsetneq$} (mso.south);

        \draw[thick,dotted] (d4.north) to [out=90,in=-90] node[tag] {\scriptsize
          $\subsetneq$} (fo.south);
        \draw[thick,dotted] (d4b.north) to [out=90,in=-90] node[tag] {\scriptsize
          $\supsetneq$} (fo.south);
      \end{tikzpicture}}
  \end{center}
  \caption{Quantifier Alternation Hierarchies}
  \label{fig:hiera3}
\end{figure}

\subsection{The Membership and Separation Problems}

We now present the two decisions problems investigated in the paper, called membership and separation. Both
problems can be defined for any class of languages, and therefore in
particular for any
class corresponding to a level in either the order or the enriched hierarchy.

\mypar{The membership problem.} Fix a class of languages~$\Cs$. The \emph{membership problem} for $\Cs$ is as follows:

\smallskip
\begin{tabular}{rl}
  {\bf INPUT:}  &  A regular language $L$. \\[.5ex]
  {\bf OUTPUT:} &  Does $L$ belong to $\Cs$?
\end{tabular}

\smallskip\noindent Usually, an algorithm solving the membership problem for
$\Cs$ is called a \emph{decidable characterization} of $\Cs$. Note that, in
general, there is no guarantee that there exists such an algorithm. In fact,
one can actually build classes of regular languages having an undecidable membership
problem from decidable ones using standard
operators~\cite{ABR:Undec-Identity:92,Rhodes99-undec,DBLP:journals/ijac/Auinger10}. However, such classes
are usually \emph{ad hoc} and we have yet to find a natural class of
regular languages having an undecidable membership problem.

Decidable characterizations are known for \fow~\cite{sfo,mnpfo} and up to the
\sicd level in both
hierarchies~\cite{simon75,knast83,arfi87,pwdelta_conf,pwdelta,gssig2} (see
Section~\ref{sec:history} for more details). In this paper, we expand this
knowledge and prove decidable characterizations for the levels \bscd, \dect,
\sict and \pict in both hierarchies. These new results rely on the
investigation of a deeper problem that we now define: the separation problem.

\mypar{The separation problem.} Let
$L,L_0,L_1$ be languages. We say that $L$ \emph{separates} $L_0$ from $L_1$ if
$$L_0 \subseteq L \text{ and } L_1 \cap L = \emptyset.$$
For a class $\Cs$ of languages, we say that $L_0$ is \emph{$\Cs$-separable}
from $L_1$ if some language in $\Cs$ separates $L_0$ from $L_1$. Note that when
$\Cs$ is closed under complement, then $L\in\Cs$ separates $L_0$ from $L_1$
iff $A^*\setminus L$ (which also belongs to \Cs) separates $L_1$ from $L_0$. Observe however that
when $\Cs$ is not closed under complement (for instance when $\Cs = \sic{i}$ or
$\Cs = \pic{i}$), the definition is not symmetrical: it may be the case that
$L_0$ is $\Cs$-separable from $L_1$, while $L_1$ is not $\Cs$-separable from
$L_0$. The separation problem for $\Cs$ is as follows: \smallskip

\begin{tabular}{rl}
  {\bf INPUT:}  &  Two regular languages $L_0$ and $L_1$. \\[.5ex]
  {\bf OUTPUT:} &  Is $L_0$ $\Cs$-separable from $L_1$?
\end{tabular}

\smallskip\noindent The separation problem is a refinement of the membership
problem. Indeed, observe that asking whether a language $L$ is $\Cs$-separable
from its complement is equivalent to asking whether $L \in \Cs$, since the
only potential separator is $L$ itself. Hence, since regular languages are
effectively closed under complement, membership immediately reduces to separation.

The separation problem is known to be decidable for
full \fow~\cite{henckell:1988,DBLP:journals/ijac/HenckellRS10a} thanks to a result of
\citeN{MR1709911}, who proved that the problems solved in these papers are
equivalent to separation. A direct proof for \fow has been obtained recently
by the authors~\citeyear{pzfo,PZ:FO-Sep16}. Separation is also known decidable up to \decd in both
hierarchies~\cite{martens,pvzmfcs13,pzsucc}. In this paper, we present a
solution for \sicd and \picd in both hierarchies.

\smallskip

Note that while we obtain results for both the order and the enriched
hierarchies, we mostly work with the order hierarchy. For the enriched
hierarchy, it is known that for each level, both the
membership~\cite{StrauVD,pinweilVD} and the separation
problem~\cite{Steinberg:delay-pointlikes:2001,pzsucc} can be reduced to the
same problem for the level's counterpart in the order hierarchy. We present
these reductions in Section~\ref{sec:transfer}. In other sections, we work
with the order hierarchy only.

\section{History}
\label{sec:history}
We presented in Section~\ref{sec:quant-altern-hier} two hierarchies within
first-order logic, defined in purely logical terms. Historically, the very
same hierarchies were first considered in a language theoretic framework and
were given combinatorial definitions. In this section, we review the history
related to these hierarchies, starting with this language
theoretic point of view.

We face a compromise between two natural approaches: the first
one would be to present results in a purely chronological order---at the risk
of getting bogged down in details and thereby missing central threads---while
the second one would be to simply highlight major ideas that have emerged all
along the years---at the cost of possibly loosing the time~line. For the sake
of (hopeful) readability, we choose a hybrid approach: we shall review the
main trends and milestones, but we adopt a chronological view for each of them.

We organize the section as follows. In Section~\ref{sec:motivation}, we
motivate why such hierarchies have been considered, and we present their
combinatorial definitions. 
Next in Section~\ref{sec:hierarchies}, we connect the combinatorial and
logical definitions. Finally, in Section~\ref{sec:connections-algebra}, we
will focus on developments that lead to solutions of membership problems for
fragments of these hierarchies. We shall explain along the way how research about
these hierarchies actually influenced a wide scientific domain.

The literature about these hierarchies is abundant. In this paper, we only
focus on some specific aspects. For more details and a complete bibliography,
we invite the reader to refer to the papers surveying the subject,
\emph{e.g.}, by \citeN{DBLP:journals/ita/Brzozowski76}, \citeN{EilenbergB},
\citeN{Weil-ConcatSurvey1989}, \citeN{Thomas:Languages-automata-logic:1997:a}
and
Pin~\citeNN{Pin_1995,Pin_1997,pinbridges,Pin-ThemeVar2011,PZ:Siglog15,Pin:WSPC16}
and to the literature cited in these papers.

\subsection{From Schützenberger's Theorem to Concatenation Hierarchies}
\label{sec:motivation}

We first introduce two concatenation hierarchies defined in combinatorial
terms with the motivation of classifying regular languages. Note that we only
recall their definition in this subsection. In the rest of the section, we
shall present connections between these hierarchies and logical ones, and focus
on tools that were developed to investigate them.

\smallskip The definitions of these hierarchies have their source in
Schützenberger's Theorem~\citeyear{sfo}, which provides an algorithm to decide
whether a regular language is star-free. Recall that star-free languages are built
from singleton languages using a finite number of times
\begin{itemize}
\item \emph{Concatenation products}: if $K$ and $L$ are star-free, then
  so is $KL=\{xy\mid x\in K,\ y\in L\}$,
\item \emph{Boolean combinations}: any finite boolean combination of
  star-free languages is star-free.
\end{itemize}

Schützenberger's Theorem~\citeyear{sfo} states that a language is star-free
if and only if its syntactic monoid is \emph{aperiodic}. The key point is that
aperiodicity of a finite monoid is a decidable property. All proofs of this
result, either close to the original one
\cite{LuchesiSimonKowaltowskiBook,perrinfo,pippengerbook,Higgins:sf:2000,conf/lata/Colcombet11,Pin:Mathematical-Foundations-Automata-Theory:2015:a}
or using alternate ideas, like \cite{MeyerFO,EilenbergB} or \cite{wfo,DGfo},
build a star-free expression from an aperiodic language. However, as explained
in the introduction, this expression may be unnecessarily complicated, in
particular it may involve avoidable interleavings between the complement and
concatenation~operations.

\subsubsection{The Dot-Depth Hierarchy} The question addressed by
\citeN{BrzoDot} when they defined the dot-depth hierarchy was to
classify star-free languages according to this complexity: the level assigned
to a language is the minimal nesting between complement and concatenation that
is necessary to express it with a star-free expression (hence the name:
`dot' means `concatenation'). Its definition is motivated by
understanding the interplay between boolean operations and one of the
fundamental operations involved in the definition of rationality, namely the
concatenation product of languages, as defined~above.

\smallskip

\smallskip
Defining the hierarchy amounts to (1) defining a base level, numbered 0,
consisting of ``simple'' languages, and (2) defining how to build level $i+1$
from level~$i$, for each natural integer $i$. This step can be decomposed in
two sub-steps:
\begin{itemize}
\item Level $i+\frac12$ is the closure of level $i$ under finite unions and (possibly marked)
  products.
\item Level $i$ is the closure of level $i+\frac12$ under finite boolean combinations.
\end{itemize}
Levels of the form $i+\frac12$ for an integer $i$ are called \emph{half
  levels}. They were missing in the original definition, but introduced later
by \citeN{PPOrder}. There are actually several variations of the dot-depth
hierarchy in the literature, see Table~\ref{tab:2}. These variants
consist in choosing the interpretation domain ($A^+$ or $A^*$), or the
base level, or the precise way to go from an integer level to the next
half level. To define half levels, several closure operators have been
considered in addition to finite unions, such as closure under usual
product of languages, or marked product instead of product, defined as 
follows for $a\in A$,
\[
  KaL=\{xay\mid x\in K,\ y\in L\}.
\]
These minor adjustments were motivated by the needs of each paper. For
instance, the definition of \citeN{ThomEqu} is convenient to establish
a correspondence between this hierarchy and the enriched hierarchy at
all levels, including level 0 (whose logical definition differs also
slightly from ours). Likewise, \citeN{pwdelta} only consider languages
of nonempty words to elegantly formulate a correspondence with
algebraic classes. It is easy to get lost in all these variations, but
what the reader should remember is that these changes are harmless:
the definitions coincide on all levels, except~possibly on level 0
(with the restriction that levels of hierarchies over $A^+$ consist in
traces over $A^+$ of languages belonging to hierarchies over~$A^*$).
In particular, for each level, all variants have the same decidability
status with respect to the problems we consider.
\begin{table}[h]
  \centering
    \def\enl{\large\strut}
    \begin{tabular}{|@{\ \enl}c@{ }|c@{ }|c|@{}c@{}|}
      \cline{2-4}
      \multicolumn{1}{c}{}& \multicolumn{1}{|c|}{Domain} & Level 0
      &\begin{tabular}{c}Closure from level\large\strut\\ $i\in\nat$ to level
         $i+\frac12$\\ Union and:\\[.5ex]\end{tabular}\\
      \hline
      \cite{BrzoDot}& $A^*$&Finite or co-finite&$K,L\mapsto KL$\\
      \hline
      \cite{ThomEqu}& $A^+$ &Bool\{$uA^*v\mid u,v\in A^*$\}&$K,L\mapsto KL$\\
      \hline
      \cite{pwdelta}& $A^+$&$\{\emptyset,A^+\}$&\large{\strut}$\begin{array}{c}K,L\mapsto uKvLw,\\ u,v,w\in A^*\end{array}$\\
      \hline
      \cite{Pin-ThemeVar2011}& $A^*$ &Bool\{$uA^*v\mid u,v\in A^*$\}&$K,L\mapsto KaL,\ a\in A$\\
      \hline
    \end{tabular}
    \caption{Some variations in the definition of the dot-depth hierarchy}
  \label{tab:2}
\end{table}

\subsubsection{The Polynomial Closure} Historically, the 
most investigated concatenation operator is the marked product used,
\emph{e.g.}, in the definition of~\citeN{Pin-ThemeVar2011} of the
dot-depth hierarchy (last line of Table~\ref{tab:2}). The operation
that associates to a class of languages its closure under finite
unions and marked products is called  \emph{polynomial
  closure}~\cite{sfo}. It is the common operation employed for going
from level~$i$ to level $i+\frac12$ in the dot-depth and in another hierarchy that
we now present. In other words, this new hierarchy differs from the
dot-depth only by the choice of the base level.

\subsubsection{The Straubing-Thérien Hierarchy} Ten years after the dot-depth
was defined, Straubing \citeNN{StrauConcat,StrauVD} and \citeN{TheConcat}
independently considered a similar and also natural hierarchy. As before, its
definition is by induction.
\begin{itemize}
\item The class of languages of level $0$ is $\{\emptyset,A^*\}$.
\item For any integer $i \geqslant 0$, level
  $i + \frac{1}{2}$ is the polynomial closure of level~$i$.
\item Languages of level $i + 1$ are the finite boolean
  combinations of languages of level~$i+\frac{1}{2}$.
\end{itemize}

Comparing the definition of the hierarchies (last line of Table~\ref{tab:2}
for the dot-depth) yields inductively that each level in the Straubing-Thérien
hierarchy is contained in the corresponding level of the dot-depth hierarchy.
The containment is actually strict, and this makes it natural to investigate
the exact relationship between the two hierarchies. Also clearly, both
hierarchies fully cover all star-free languages.

\subsubsection{Strictness of the Hierarchies}
\label{sec:strictness}
\noindent
The first natural question is whether these definitions actually yield strict
(or infinite, this is equivalent in this case) hierarchies  or whether they
collapse. The dot-depth hierarchy was shown to be strict
by~\citeN{BroKnaStrict} for alphabets of size at least 2 on integer levels:
one can show that $L_n$ defined inductively by $L_0=\varepsilon$ and
$L_n=(aL_{n-1}b)^*$ is at level~$n$ in the dot-depth hierarchy. Another proof
of the fact that the hierarchy is strict based on algebra was given by
\citeN{StrauConcat}. Yet other proofs were presented by Thomas
\citeNN{ThomStrict2,ThomStrict}, using arguments based on \efgame
games. 
All these proofs easily imply that the hierarchy is strict on \emph{all}
levels, including half levels.

Regarding the Straubing-Thérien hierarchy, strictness was established by
\citeN{Margolis&Pin:Products-group-languages:1985:a} (see also
\cite{therien:powersurvey} for a short proof). Strictness actually follows
from a more general result of~\citeN{StrauVD} that connects both
hierarchies, see below.

The fact that both hierarchies are strict makes it relevant to investigate the
membership problem at each level of each of these hierarchies. Relatively few
results are known, but this question motivated a wealth of fruitful ideas. We
shall describe progress in this line of research in
Section~\ref{sec:connections-algebra}. Before, let us
connect the combinatorial definitions with the ones relying on first-order
logic, which we presented in Section~\ref{sec:quant-altern-hier}.

\subsection{Connections with Logic}
\label{sec:hierarchies}

The interest in the dot-depth and Straubing-Thérien hierarchies increased
after relationships were discovered in the eighties, first by Thomas, then by
Perrin and Pin between them and logical hierarchies. Recall that we defined 
two alternation hierarchies within first-order logic in Section~\ref{sec:words}:
the order hierarchy, which counts alternations between blocks of $\exists$ and
$\forall$ quantifiers for formulas in the signature
$\{{<},\ P_a\mid a\in A\}$, and the enriched hierarchy, which counts the same
alternations for formulas in the signature
$\{{<},\ {+}1,\ \mathit{min},\ \mathit{max},\ P_a\mid a\in A\}$.

\smallskip Recall also that \citeN{sfo} proved that star-free languages are
exactly first-order definable ones. \citeN{ThomEqu} discovered a more precise
correspondence, level by level, between the dot-depth hierarchy of star-free
languages and the enriched quantifier alternation hierarchy within \fo. Note
that~\citeN{ThomEqu} did not actually state the result for half levels, as
they were not considered. However, it can be easily derived from the arguments
of the paper.

\begin{theorem}[\citeN{ThomEqu}] \label{thm:ddequi}
  Let $i \geqslant 0$. Then,
  \begin{itemize}
  \item A language has dot-depth $i$ if and only it is definable in \bswsi.
  \item A language has dot-depth $i+\frac{1}{2}$ if and only it is definable in \siws{i+1}.
  \end{itemize}
\end{theorem}

This connection with finite model theory and descriptive complexity sustained
an earlier informal statement by~\citeN{DBLP:journals/ita/Brzozowski76}
arguing that dot-depth is a relevant complexity parameter. The argument was
based on the fact that star-free expressions can express feedback-free
circuits, and that concatenation increases the depth of such circuits:
\emph{since concatenation (or ``dot'' operator) is linked to the sequential
  rather than the combinational nature of a language, the number of
  concatenation levels required to express a given aperiodic language should
  provide a useful measure of complexity.} Since it was known that the
nonelementary complexity of standard problems for \fo is tied to quantifier
alternation~\cite{Stockmeyer:1973:WPR:800125.804029}, Theorem~\ref{thm:ddequi}
brought mathematical evidence that the level in the dot-depth hierarchy of a language
is indeed a meaningful complexity measure, thus supporting Brzozowski's intuition.

\smallskip
A statement similar to Theorem~\ref{thm:ddequi} was established by~\citeN{PPOrder} for the
Straubing-Thérien hierarchy, which corresponds to the order hierarchy.

\begin{theorem}[\citeN{PPOrder}] \label{thm:stequi} Let
  $i \geqslant 0$ be an integer. Then,
  \begin{itemize}
  \item A language has level $i$ in the Straubing-Thérien hierarchy if and only if it is
    definable in \bswi.
  \item A language has level $i+\frac{1}{2}$ in the Straubing-Thérien hierarchy if and only if
    it is definable in \siw{i+1}.
  \end{itemize}
\end{theorem}

In addition, \citeN{PPOrder} introduced half levels as the closure under
finite unions and intersections of marked products of the preceding level (it
turns out that intersection is actually useless, see~\cite{Arfi_1991,PinIntersection13}). Finally, they extended the
correspondence to infinite~words.

\medskip The results obtained during the 70s and the 80s fostered many
connections among several communities of researchers, working in automata
theory, semigroup theory or finite model theory, and laid the ground of a
clean framework, with tools from these different fields. The research effort
continued in the 90s, in particular with the developments of algebraic methods
to investigate membership problems.

\subsection{Connections with Algebra: the Syntactic Approach}
\label{sec:connections-algebra}

\medskip
Knowing that both hierarchies are strict and that they
capture a meaningful complexity measure, the most natural question is
whether we can compute the level in each of these hierarchies of an
input regular language. This corresponds to solving membership for
each level. Even though the membership problem is standard nowadays,
this is only after Schützenberger's work that it was identified as the
salient problem to look at. Moreover, \citeN{sfo} also proposed a
convenient tool to solve this problem, namely the syntactic monoid.
See \cite{Pin_1997} for a comprehensive survey on this topic. 

\subsubsection{Syntactic Monoids: Definition and Seminal Result} The \emph{syntactic congruence}
$\sim_L$ of a language~$L$, defined by \citeN{schsynt}, relates those words
that cannot be distinguished by the language when embedded in the same
context. Formally,
\[
  u\sim_L v\Longleftrightarrow (\forall x,y\in A^*,\ xuy\in L\Leftrightarrow xvy\in L).
\]
The key result of Myhill and Nerode \citeyear{mn58} implies that a language is
regular if and only if this congruence has finite index. Hence, in this case, the
quotient set $A^*/{\sim_L}$ is a computable finite monoid, called the
\emph{syntactic monoid} of the language. An easy to check but important
property is that $L$ is a union of $\sim_L$-classes, so that the so-called
\emph{syntactic morphism} from $A^*$ to $A^*/{\sim_L}$ that maps a word to its
$\sim_L$-class recognizes $L$ (in the sense that $L$ is a union of
$\sim_L$-classes, and so it is the preimage of a subset of the syntactic monoid under
the syntactic morphism).

\smallskip
Schützenberger's Theorem precisely states that a language is
star-free if and only if it is \emph{aperiodic}, \emph{i.e.}, its syntactic monoid
satisfies the equation
\[
  x^\omega=x^{\omega+1},
\]
where $\omega$ represents some large integer, which can be computed from the
language as well. This means that for every element $x$ of the syntactic
monoid of the language, the equality $x^\omega=x^{\omega+1}$ has to hold.
Since the syntactic monoid of the input language is finite and computable from
any representation of the language, checking whether it is is aperiodic is a
\emph{decidable} property. To sum up, Schützenberger's Theorem~\citeyear{sfo} reduces
a nontrivial semantic property (to be definable in some fragment for a
language) into a purely syntactic, easily testable condition (to satisfy an
equation for a finite, computable algebra).

\smallskip The importance of this result stems from two reasons:

\begin{itemize}
\item First, Schützenberger established membership as the standard problem
  that is worth investigating in order to understand a class of regular
  languages. This is justified, since obtaining a decidable characterization
  requires a deep insight about the class, as this amounts to capturing in a
  \emph{single} algorithm \emph{all} properties that can be expressed within
  the class.

\item Schützenberger also proposed a methodology which proved successful in
  solving other membership problems. Let us briefly explain the core of his
  strategy. The hardest direction is to build a star-free expression for a
  language whose syntactic monoid is aperiodic. The key observation is that
  either \emph{all} languages recognized by a syntactic monoid are star-free,
  or \emph{none} of them~is. Hence, instead of building a star-free expression
  for a single language, one may rather do so for all languages recognized by
  its syntactic morphism. The payoff of this approach may not be immediate, as
  the goal is more demanding than the original~one. Yet, the languages
  recognized by the syntactic monoid are connected one another, which
  makes the method amenable to induction as soon as one can decompose each
  language into simpler ones using only star-free operations.
\end{itemize}

Despite the impact that is acknowledged nowadays to Schützenberger's
methodology, about ten more years were necessary to cement it as a fundamental
approach.

\subsubsection{Validation of the Syntactic Approach}
Notable breakthroughs after Schützenberger's Theorem were obtained
by Simon, a student of Brzozowski, in his PhD~\citeyear{SimonPhD} shortly
after the dot-depth hierarchy was defined. His results had a major impact on
research in the theoretical computer science community, particularly two of
them, characterizing important subclasses of level one [\citeNP{BSlocalConf};
\citeyearNP{BSlocal}],~\cite{simon75}:
\begin{enumerate}[label=$\alph*)$]
\item The class of locally testable languages, \emph{i.e.}, such that
  membership of a word in such a language is determined only by looking at
  infixes, prefixes, infixes up to a given length. This result was also
  obtained independently by
  \citeN{McNaughton:Algebraic-decision-procedures-local:1974:a} and
  \citeN{Zalcstein:Locally-testable-languages:1972:a}.
  It is easy to check that these languages form a subclass of dot-depth one.

\item The class of piecewise testable languages, \emph{i.e.}, such that
  membership of a word in such a language is determined only by looking at its
  scattered subwords up to a given length. It is the boolean algebra generated
  by languages of the form $A^*a_1A^*\cdots A^*a_nA^*$. This is exactly the
  first level of the Straubing-Thérien hierarchy, corresponding to the
  class~$\bsw1$. Note however that this hierarchy was not already defined at
  that time.
\end{enumerate}

Before presenting other results about the hierarchies, let
us comment these results and explain why they deeply influenced the theoretical computer
science landscape.

\begin{itemize}
\item The main reason why Simon's results were recognized as important is that
  they supported Schützenberger's methodology as the ``right'' one to tackle
  membership questions, by underlining the key role played by the syntactic
  monoid in automata theory. Indeed, Schützenberger and Simon both used the
  same strategy in order to obtain their decidable characterizations, by
  reducing membership to checking equations on the syntactic monoid. This
  common approach was further validated by \citeN{EilenbergB} who established
  a one-to-one correspondence between varieties of regular languages and
  varieties of finite monoids. It was complemented by a theorem of
  \citeN{Reiterman:Birkhoff-theorem-finite-algebras:1982:a}, which shows that
  these algebraic classes can be described by a (possibly infinite)
  set of equations (such as
  $x^\omega=x^{\omega+1}$ for aperiodic monoids, which characterize star-free
  languages). Note however that Eilenberg's and Reiterman's theorems
  are generic results, useless for actual characterizations. They do
  \emph{not} provide a uniform solution to all membership problems:
  the actual decidable characterization depends of course on the class
  under investigation (the topic of this paper is precisely to
  establish such characterizations for levels in the hierarchies).
  Yet, Eilenberg's and Reiterman's theorems entail that any 
  class of regular languages that forms a variety can be characterized by
  equations satisfied by all syntactic monoids of languages of the class.
  Since all integer levels in the hierarchies are indeed varieties, a major
  research direction was to understand the variety of finite monoids
  associated to them. In this paper, we provide equations for \bsw2.

\item  A second reason why the study on the hierarchies in general and
  Simon's results in particular were recognized as important is that they
  connected several areas: automata theory, finite semigroup theory, but also
  combinatorics on words (see \cite{SakaSimon} or more recently
  \cite{pk-phs-mk-indexSimon}) and finite model theory. They received a number
  of proofs, either reminiscent of the original ones
  \cite{Lallement:1979:SCA:539871,Pin:Varietes-langages-formels:1984:a,Howie91,Pin:Mathematical-Foundations-Automata-Theory:2015:a},
  or using arguments of different flavors. For instance, just for the case of
  piecewise testable languages,
  \citeN{Straubing&Therien:Partially-ordered-finite-monoids:1988:a} gave an
  alternate proof based on an early use of ordered monoids in automata theory.
  Ordered monoids turned out to be a key notion in the study of the
  hierarchies (see below), and the result was reproved by
  \citeN{HenckellPinJordered}. Almeida~\citeNN{AlmeidaBS1,AlmeidaBook}
  presented a proof based on profinite topology, \citeN{Higgins_1997} a proof
  using representations by transformation semigroups,
  \citeN{Klima:Piecewise-testable-languages-combinatorics:2011:a} a purely
  algebraic one. Simon's decision criteria for both classes were refined to
  understand the computational hardness of the associated membership
  problems~\cite{Stern:Complexity-some-problems-from:1985:a,DFA-sf-PSPACE} and
  to improve the complexity of original
  algorithms~\cite{Stern:Characterizations,PinLTT,PinLTTjournal,Trahtman:LTTDFA,Trahtman:PT-LTT2001,KlimaPTComb}.

\item  At last, Simon's work contains ingredients that inspired several
  researchers to solve other membership problems. For instance, the result on
  locally testable languages \cite{BSlocal} introduces the notion of graph
  congruence, reused by \citeN{KnastGraph83} to give a membership algorithm
  for level~1 of the dot-depth hierarchy. Simon's result was also influential
  for characterizing the class of locally threshold testable languages, where
  membership of a word depends not only on the set of infixes, but also on the
  number of such infixes counted up to a threshold. This class was
  characterized by Beauquier and Pin~\citeNN{Beauquier_1989,Beauquier_1991} by
  relying on a deep paper of
  \citeN{Therien&Weiss:Graph-congruences-wreath-products:1985:a} that used
  again graph congruences (a completely different proof by \citeN{bojLTT}
  relies on the decidability of Presburger logic and Parikh's theorem). Graph
  congruences in turn are the premises of the framework developed by
  \citeN{TilsonCat}, motivated by difficult decision problems in semigroup
  theory (in particular, the decidability of the well-known Krohn-Rhodes
  hierarchy, which classifies languages according to serial~decomposition).
\end{itemize}
\medskip
Following Simon's results, level~1 in the dot-depth hierarchy was successfully
characterized by Knast. The proof is, however, much more technical. To sum up,
\begin{itemize}
\item \citeN{simon75} characterized level 1 in Straubing-Thérien
  hierarchy, or equivalently the fragment $\bsw1$ of first-order logic.
\item Knast~\citeyear{knast83,KnastGraph83} characterized level 1 in the dot-depth
  hierarchy, or equivalently the fragment $\bsws1$ of first-order logic.
\end{itemize}
While these results and others for classes outside the hierarchies gather
evidence that the syntactic approach is relevant to tackle membership
problems, the time intervals between significant contributions regarding
levels in the hierarchy show that the problem is difficult. Despite a wealth
of results towards a solution for level 2, the last complete statement until
the conference version of the present paper~\cite{pzqalt} regarding integer
levels went back to Knast~\citeyear{knast83,KnastGraph83}. 

\subsubsection{Connections between the Hierarchies}
It should be noted that after Knast's result, researchers became primarily
focused on the Straubing-Thérien hierarchy, rather than on the dot-depth one.
The reason is that \citeN{StrauVD} proved that it is the most fundamental of
the two from the membership point of view. More precisely, for any integer
level $i\geq2$, membership for level $i$ in the dot-depth hierarchy can be
effectively reduced to membership for the same level in the Straubing-Thérien
hierarchy. This was generalized
to half levels by \citeN{pinweilVD}. This explains why we also work with the Straubing-Thérien hierarchy
in the present paper.  We shall detail the actual
reductions in Section~\ref{sec:transfer}.

\subsubsection{Limits of the Syntactic Approach}
In view of Straubing's result, the principal objective of researchers became
to solve membership for level~2 in the Straubing-Thérien hierarchy (at that
time, half levels were not already defined). While a lot
of effort was devoted to solving this problem, this proved very difficult.
Over the years, several attempts were made:

\begin{itemize}
\item First, partial results were obtained by restricting the set of possible
  input languages. For example, level~2 was characterized by \citeN{StrauDD2}
  for languages over an alphabet of size 2. Other partial results were
  obtained by \citeN{Cowan_1993}, building on results of \citeN{Weil_1989}
  and~\citeN{Straubing_1992}.

\item Second, many upper bounds of the actual level~2 were introduced. Usually
  defined by a set of equations and having a decidable membership
  problem, these upper bounds were often presented as conjectures. When such a
  conjecture was disproved, a new one was proposed to tighten the gap between
  the proposed candidate and the actual level~2. For instance, Straubing
  \citeNN{StrauDD2Conf,StrauDD2} proposed such a candidate and proved that it
  holds in some particular cases~\cite{Straubing_1992}. Another version was
  proposed by~\citeN{pin:hal-00143953}, and refined by themselves
  in~\citeyear{Pin_2001}. More recently, \citeN{AK2009} disproved the
  conjecture of Straubing, and proposed a new candidate~\cite{AK2010}. All
  these conjectures actually provided strict upper bounds for level 2.

\item A third approach was to reduce the decidability of level~2 to distinct
  mathematical problems. A remarkable example is the relationship between the
  decidability of level 2 and a purely algebraic problem. This connection was
  discovered by \citeN{pin-straubing:upper}: they considered the variety
  generated by all finite monoids of upper triangular boolean matrices, and
  proved that it corresponds exactly to level~2 in the Straubing-Thérien
  hierarchy. Unfortunately, this problem turned out to be as hard as the
  original one. 
\end{itemize}

\subsubsection{Half Levels and Ordered Monoids} All these attempts underlined
that level~2 was difficult to attack directly. This motivated the
investigation of the half levels, introduced by~\citeN{PPOrder}. At first
glance, they may seem to be just an additional refinement, but this is not the
case. First, half levels are arguably more fundamental than integer levels, since
each integer level can be reconstructed from the preceding half level by
closure under boolean operations. Also importantly, half levels are simpler to
deal with, and understanding them is a first step towards membership
algorithms for integer levels. For instance, gathering enough information
about level $\frac32$ is crucial in our approach to the solution of the
membership problem at level~2.

The main issue with half levels is that they are not closed under complement. This
is a problem for generalizing Schützenberger's methodology, which translates
the semantic membership problem into a property of the syntactic monoid.
Indeed, the reason why the syntactic approach works for ``varieties'' is
that, for such a class~\Cs, either all or none of the languages recognized by
a syntactic monoid belong to \Cs. This is precisely what fails for half
levels, since a language is recognized by a monoid if and only if so is its complement. In other
words, the syntactic monoid is not well suited to capture classes that are not
closed under complement, and has therefore to be adapted if one wants
to generalize Eilenberg's Theorem.

Nonetheless, Arfi \citeNN{arfi87,Arfi_1991} managed to show that levels
$\frac12$ and $\frac32$ of the Straubing-Thérien hierarchy have decidable
membership, and to describe the associated classes of languages. This is very
easy for level $\frac12$. A downside of this approach for level $\frac32$ is
that it relied on involved results of~\citeN{Hashiguchi_1983}, thus hiding
the core of the argument. Note also that Straubing's transfer result did not
apply to half levels, since it relied on the correspondence between varieties
of languages and varieties of finite semigroups. This made it relevant to
investigate level $\frac32$ in the dot-depth hierarchy as well, a task
successfully achieved by \citeN{gssig2}. However, this combinatorial and
technical proof is not easily amenable to generalization.

This made it crucial to understand what could be saved from Schützenberger's
approach. In fact, Arfi's characterization for level $\frac12$ is explicitly
stated as a property to be satisfied by the syntactic monoid. This property is
not an equation, but a closure property of the accepting set of the language.
This led \citeN{pinordered} to develop an Eilenberg-Schützenberger's
methodology for classes that are not closed under complement. Pin's idea was
to equip monoids with a partial order relation compatible with multiplication
and to constrain accepting sets to be upward closed. This yields an adapted
notion of recognizability, for which the set of languages recognized by an
ordered syntactic monoid is not closed anymore under complement, but still
carries enough structure and information to recover the generic
methodology~\cite{pinordered}, including equational descriptions of such
classes~\cite{pin:hal-00143951}. Cleaner decidability membership algorithms
were subsequently re-obtained for level~$\frac32$ \cite{pwdelta_conf,pwdelta,bfacto}.
Instead of Hashiguchi's black box result, the proofs in these papers rely on a
simple tool that we shall also use: the factorization forest Theorem of
\citeN{simonfacto}. Finally, Straubing's results connecting the
Straubing-Thérien and the dot-depth hierarchies were also generalized to this
new setting~\cite{pinweilVD}, thus giving an alternate proof of the
decidability of level $\frac32$ in the dot-depth hierarchy.

\subsubsection{Level 2 and Above} This paper continues this research
effort. As explained in the introduction, a key ingredient in our
approach is to consider the \emph{separation problem}, which is more
demanding than membership. The core of our results is a solution to
this problem for \siwd, i.e., the level~$\frac32$ in the
Straubing-Thérien hierarchy (for which membership is already known to
be decidable). We are then able to obtain membership algorithms for
both \siwt and \bswd by building upon this first separation algorithm.
This highlights the fact that a solution for the separation problem
associated to some class carries information that can be exploited to
tackle weaker problems (such as membership) for more complicated
classes. In particular, a good illustration of this is the fact that
our membership algorithm for \siwt follows from a generic connection
between separion and membership: for \emph{any integer $i$}, a separation algorithm for
\siwi yields a membership algorithm for $\siw{i+1}$. On the
other hand, our membership algorithm for \bswd results from a specific
and detailed analysis of the separation algorithm for level~$\siwd$.

Finally, note that while we work with the order hierarchy (i.e., the
Straubing-Thérien hierarchy) in most sections, we come back to the
enriched hierarchy (i.e., the dot-depth hierarchy) at the end of the
paper. Using previously known transfer theorems, we are able to lift
all results that we have proved for levels in the order hierarchy to
the same levels in the enriched hierarchy.

\section{Tools}
\label{sec:tools}
In this section, we recall the definitions of two well-known combinatorial
tools used several times in the paper:

\begin{itemize}
\item the \emph{\efgame game} variant corresponding to levels of the order
  hierarchy, which are a mean to capture their expressive power in terms of
  games. For more on \efgame games, see for instance \cite{bookstraub,Immerman:DC:1999,LibkinFMT:2004}.
\item the definition of regular languages in terms of \emph{monoids}. This
  definition makes it possible to use convenient combinatorial results, in
  particular Simon's Factorization Forests Theorem, which we also present in
  this section.
\end{itemize}

\subsection{Logical Tools: \efgame Games}

It is usual to classify first-order formulas according to their
\emph{quantifier rank}, \emph{i.e.}, the length of the longest sequence of
nested quantifiers in the formula. For example, the following formula
\[
  \forall x\ P_a(x) \Rightarrow ((\exists y\ (y < x \wedge P_c(y)) \wedge (\exists y\exists z\ (x < y < z) \wedge P_b(y)))
\]
\noindent
has quantifier rank $3$. We use the quantifier rank to associate to our logics
binary relations over the set $A^*$. We begin with the logics \siwi. Set
$k,i \in \nat$ and $w,w' \in A^*$. We write
\[
  w \ksieq{i} w'
\]
if \emph{any} \siw{i} formula of quantifier rank $k$ satisfied by $w$ is also
satisfied by $w'$. Observe that since a \piw{i} formula is the negation of a
\siw{i} formula, we have $w \ksieq{i} w'$ if and only if any \piw{i} formula
of quantifier rank $k$ satisfied by $w'$ is also satisfied by $w$. Moreover,
the following facts are immediate from the definition.

\begin{fact} \label{fct:preorder}
  For all $k,i \in \nat$, $\ksieq{i}$ is a preorder.
\end{fact}

\begin{fact} \label{fct:definition} For all $k,i \in \nat$, a language $L
  \subseteq A^*$ can be defined by a \siw{i} formula of rank $k$ if and only if
  $L$ is saturated by $\ksieq{i}$, \emph{i.e.}, if and only if
  \[
    L = \{w \mid \exists w' \in L \text{ s.t. } w' \ksieq{i} w\}.
  \]
\end{fact}

We now extend this definition to the logics \bswi. Set $k,i \in \nat$ and
$w,w' \in A^*$. We write $w \kbceq{i} w'$ if $w$ and $w'$ satisfy the same
\bswi formulas of quantifier rank $k$. By definition, \bswi formulas are
finite boolean combinations of \siwi formulas. We thus obtain the following
two facts:
\begin{fact} \label{fct:equivalence}
  For all $k,i \in \nat$, $\kbceq{i}$ is the equivalence relation
  induced by $\ksieq{i}$, \emph{i.e.},
  \[
    w \kbceq{i} w' \text{ if and only if } w \ksieq{i} w' \text{ and } w' \ksieq{i} w.
  \]
  Moreover, for all fixed $k,i \in \nat$, $\kbceq{i}$ has finite index.
\end{fact} 
\begin{fact} \label{fct:definition2}
  For all $k,i \in \nat$, a language $L \subseteq A^*$ can be defined by
  a \bsw{i} formula of rank $k$ if and only if $L$ is a union of
  equivalence classes of $\kbceq{i}$, that is, if and only if
  \[
    L = \{w \mid \exists w' \in L \text{ s.t. } w' \kbceq{i} w\}.
  \]
\end{fact}
We can now define \efgame games. A specific \efgame game can be associated to
every logic. Here, we define the game tailored to the logics \siwi in the
quantifier alternation hierarchy. This means that these games characterize the
preorders \ksieq{i} (and therefore, by Fact~\ref{fct:equivalence}, also the
equivalence \kbceq{i}).

\medskip
\noindent
{\bf \efgame Games.} Before giving the definition, a remark is in
order. There are actually two ways to define the class of
\siw{i}-definable languages.  First, one can consider \emph{all} first-order
formulas and say that a formula is \siw{i} if it has at most $i$
blocks of quantifiers once rewritten in prenex normal form. 
However, one could also restrict the set of
allowed formulas to only those that are already in prenex form and have at
most $i$ blocks of quantifiers. While this does not change the class
of \siw{i}-definable languages as a whole, this changes the set of
formulas of quantifier rank $k$ for a fixed $k$. Therefore, this
changes the preorder $\ksieq{i}$. This means that there is a version
of the \efgame game for each definition. In this paper, we use the
version corresponding to the definition that considers \emph{all}
first-order~formulas.

\medskip
Set $i \geqslant 1$. We define the game associated to \siw{i}. The board of
the game consists of two words $w,w'\in A^*$ and there are two players
called \emph{Spoiler and Duplicator}. Moreover, initially, there
exists a distinguished word among $w,w'$ that we call the \emph{active 
  word} (this word may change as the play progresses). The game is set
to last a predefined number $k$ of rounds. When the play starts, both
players have $k$ pebbles. Finally, there is a parameter that gets
updated during the game, a counter $c$ called the \emph{alternation
  counter}. Initially, $c$ is set to $0$ and has to be bounded by
$i-1$.

At the start of each round $j$, Spoiler chooses a word, either $w$ or
$w'$. Spoiler can always choose the active word, in which case both $c$
and the active word remain unchanged. However, Spoiler can only choose
the word that is not active when $c < i - 1$, in which case the active
word is switched and $c$ is incremented by $1$ (in particular this
means that the active word can be switched at most $i-1$ times). If
Spoiler chooses $w$ (resp. $w'$), he puts a pebble on a position $x_j$ in
$w$ (resp. $x'_j$ in $w'$).

Duplicator must answer by putting a pebble at a position $x'_j$ in
$w'$ (resp. $x_j$ in $w$). Moreover, Duplicator must ensure that all
pebbles that have been placed up to this point satisfy the following
condition: for all  $\ell_1,\ell_2 \lmo j$, the labels at positions
$x^{}_{\ell_1},x'_{\ell_1}$ are the same, and $x_{\ell_1} < x_{\ell_2}$
if and only if $x'_{\ell_1} < x'_{\ell_2}$.

Duplicator wins if she manages to play for all $k$ rounds, and Spoiler
wins as soon as Duplicator is unable to play.

\begin{lemma}[Game definition of $\ksieq{i}$, folklore] \label{lem:efgame}
  For all $k,i \in \nat$ and $w,w' \in A^{*}$, $w \ksieq{i} w'$ if and
  only if Duplicator has a winning strategy for playing $k$ rounds in
  the \siw{i} game played over $w,w'$ with $w$ as the initial active word.
\end{lemma}

Note that we will often use Lemma~\ref{lem:efgame} implicitly and
alternate between the original and the game definition of
$\ksieq{i}$. We now give a few classical lemmas on \efgame games
that we reuse several times in our proofs. We begin with a
lemma stating that $\ksieq{i}$ is a pre-congruence, \emph{i.e.}, that
it is compatible with the concatenation product.

\begin{lemma}[Pre-congruence Lemma]
  \label{lem:efconcat}
  Let $i \in \nat$ and let $w^{}_1,w'_1,w^{}_2,w'_2 \in A^*$. Then
  \begin{equation*}
    (w^{}_1 \ksieq{i} w'_1\text{ and } w^{}_2 \ksieq{i} w'_2) \implies w^{}_1w^{}_2 \ksieq{i}
    w'_1w'_2.
  \end{equation*}
\end{lemma}

\begin{proof}
  By Lemma~\ref{lem:efgame}, Duplicator has winning strategies in the
  \siw{i} games over $w^{}_1,w'_1$ and $w^{}_2,w'_2$, with $w^{}_1,w^{}_2$
  as initial active words respectively. These strategies can be
  easily combined into a strategy for the \siw{i} game over
  $w^{}_1w^{}_2$ and $w'_1w'_2$, with $w^{}_1w^{}_2$ as initial active word. We
  conclude that $w^{}_1w^{}_2 \ksieq{i} w'_1w'_2$.
\end{proof}

The second lemma is a well-known property of full first-order
logic, which implies that, unlike monadic second order logic, first-order
logic cannot express modulo counting. This property is called \emph{aperiodicity}.

\begin{lemma}[Aperiodicity Lemma]
  \label{lem:aperiodic}
  Let $k,k_1,k_2 \in \nat$ be such that $k_1,k_2 \gmo 2^{k}-1$. Let $v
  \in A^*$. Then
  \[
    \forall i\in\nat,\quad v^{k_1} \ksieq{i} v^{k_2}
  \]
\end{lemma}

\begin{proof}
  This is well known for full first-order logic and easy to prove by induction
  on~$k$ (see~\cite{bookstraub} for details).
\end{proof}

We finish with another classical property, which we call the
\emph{$\sic i$-property}. Contrary to the pre-congruence or aperiodicity
properties, the $\sic i$-property is specific to \siw{i}. It will be central
in the~proofs.

\begin{lemma}[$\sic i$-property Lemma] \label{lem:siprop}
  Let $i \in \nat$, and let $k,\ell,r,\ell',r' \in \nat$ be such that
  $\ell,r,\ell',r' \gmo 2^k$ and let $u,v \in A^*$ such that $v \ksieq{i} u$.
  Then we have:
  \[
    u^{\ell}u^{r} \ksieq{i+1} u^{\ell'}vu^{r'}.
  \]
\end{lemma}

\begin{proof}
  Set $w = u^\ell u^{r}$ and $w' = u^{\ell'}vu^{r'}$. We prove that
  $w \ksieq{i+1} w'$ using an \efgame argument: we prove that Duplicator has a
  winning strategy for the game in $k$ rounds for \siw{i+1} played on $w,w'$
  with $w$ as initial active~word. The proof goes by induction on $k$. We
  distinguish two cases depending on the value, 0 or 1, of the alternation
  counter $c$ after Spoiler has played the first round.

  \medskip
  \noindent
  {\bf Case 1: $c=1$.} In this case, by definition of the game, it
  suffices to prove that $w' \ksieq{i} w$. From our hypothesis we
  already know that $v \ksieq{i} u$. Moreover, it follows from
  Lemma~\ref{lem:aperiodic} that $u^{\ell'} \ksieq{i} u^\ell$
  and $u^{r'} \ksieq{i} u^{r - 1}$. It then follows from
  Lemma~\ref{lem:efconcat} that $w' \ksieq{i} w$.

  \medskip
  \noindent
  {\bf Case 2: $c=0$.} By definition, this means that Spoiler has played on some
  position $x$ in $w$. Therefore $x$ is inside a copy of the word $u$. Since
  $w$ contains at least $2^{k+1}$ copies of $u$, by symmetry we can assume
  that there are at least $2^{k}$ copies of $u$ to the right of $x$. We now
  define a position $x'$ inside $w'$ that will serve as Duplicator's
  answer. We choose $x'$ so that it belongs to a copy of $u$ inside $w'$ and
  is at the same relative position inside this copy as $x$ is in its own copy
  of~$u$. Therefore, to fully define $x'$, it only remains to define the copy
  of $u$ in which we choose~$x'$. Let $n$ be the number of copies of $u$ to
  the \emph{left} of $x$ in $w$, that is, $x$ belongs to the
  $(n+1)$-th copy of $u$   starting from the left of $w$. If $n <
  2^{k-1}-1$, then $x'$ is chosen inside the $(n+1)$-th copy of $u$
  starting from the left of $w'$. Otherwise, $x'$ is chosen inside the
  $2^{k-1}$-th copy of $u$ starting from the left of $w'$. Observe
  that these copies always exist and occur before the factor $v$, since $\ell'\gmo 2^k$.

  Set $w=w_puw_q$ and $w'=w'_puw'_q$, where the two distinguished `$u$' factors
  are the copies containing positions $x,x'$. By definition of
  the game, it suffices to prove that $w_p \sieq{k-1}{i+1}
  w'_p$ and $w_q \sieq{k-1}{i+1} w'_q$ to conclude that Duplicator
  can play for the remaining $k-1$ rounds. If $n < 2^{k-1}-1$, then by
  definition, $w_p=w'_p$, therefore it is immediate that
  $w_p \sieq{k-1}{i+1} w'_p$. Otherwise, both $w_p$ and
  $w_p'$ are concatenations of at least $2^{k-1}-1$ copies of
  $u$. Therefore $w_p \sieq{k-1}{i+1} w'_p$ follows
  Lemma~\ref{lem:aperiodic}. Finally observe that by definition $w_q$ and $w'_q$
  are of the form $w_q =u^{\ell_1}u^{r}$ and $w'_q=u^{\ell'_1}vu^{r'}$ for some
  $\ell_1$ and $\ell'_1$ such that $\ell_1 + r \gmo
  2^k$ (by the assumption made at the beginning of Case 2) and $\ell'_1,r' \gmo
  2^{k-1}$ (by the choice made by Duplicator and hypothesis on $r'$).
  Therefore, it is immediate by induction on $k$ that $w_q \sieq{k-1}{i+1} w'_q$.
\end{proof}

\subsection{Algebraic Tools: Monoids and Simon's Factorization Forests  Theorem}

A \emph{semigroup} is a set $S$ equipped with an associative
multiplication denoted by '$\cdot$'. A \emph{monoid}~$M$
is a semigroup in which there exists a neutral element denoted
$1_M$. Observe that $A^*$ is a monoid with concatenation as the
multiplication and $\varepsilon$ as the neutral element.

An element $e$ of a semigroup is \emph{idempotent} if $e\cdot e=e$. Given any
finite semigroup $S$, it is well known that there is a number $\omega(S)$,
denoted by $\omega$ when $S$ is understood from the context, such that
$s^\omega$ is an idempotent for each element $s$ of $S$:
$s^\omega = s^\omega \cdot s^\omega$.

Monoids are a standard tool to recognize regular languages.  Let $L$ be a
language and~$M$ be a monoid. We say that \emph{$L$ is recognized by $M$} if
there exists a monoid morphism $\alpha : A^* \rightarrow M$ and an
\emph{accepting set} $F \subseteq M$ such that $L=\alpha^{-1}(F)$. Kleene's
theorem states that a language is regular if and only if it can be recognized
by a \emph{finite monoid}.

The usual approach to characterize a class of regular languages is to abstract
it as a class of monoids, each recognizing only languages in the class, and
such that conversely any language is recognized by one of these monoids. For
such an approach to work, the class of languages has to fulfill some
properties. In particular, since any monoid recognizing a language also
recognizes its complement, this approach only makes sense when the class of
languages is closed under complement (among other operations).

In the paper however, we investigate classes of languages, such as \siwi, that are \emph{not}
closed under complement. For such classes, one needs to use \emph{ordered
  monoids} as recognizing structures. An ordered monoid is a monoid endowed
with a partial order '$\lmo$' which is compatible with multiplication:
$s\lmo t$ and $s'\lmo t'$ imply~$ss'\lmo tt'$.

We say that $L$ is recognized by an ordered monoid $M$ if there exist a
monoid morphism $\alpha : A^* \rightarrow M$ and an \emph{upward closed}
accepting set $F \subseteq M$ such that $L=\alpha^{-1}(F)$. One also says that
$\alpha$ recognizes $L$. The condition for
$F$ of being upward closed means that if $s\in F$ and $s\lmo t$, then also
$t\in F$. Note that if $\alpha$ recognizes $L$, although $A^*\setminus L = \alpha^{-1}(M\setminus F)$, the set
$M\setminus F$ is not necessarily upward closed, hence
$A^*\setminus L$ is not necessarily recognized by~$\alpha$.

\medskip
\noindent {\textbf{Syntactic Ordered Monoid of a Language.}}
Given a regular language $L$, one can compute a canonical finite ordered
monoid that recognizes it as follows. The \emph{syntactic preorder} $\lmo_L$ of
a language $L$ is defined on pairs of words in $A^*$ by $w \lmo_L w'$ if for all
$u,v \in A^*$, $uwv \in L \Rightarrow uw'v \in L$. Similarly, we define
$\equiv_L$, the \emph{syntactic equivalence} of $L$ as follows:
$w \equiv_L w'$ if $w \lmo_L w'$ and $w' \lmo_L w$. One can verify that
$\lmo_L$ and $\equiv_L$ are compatible with multiplication. Therefore, the
quotient $M_L$ of $A^*$ by $\equiv_L$ is an ordered monoid for the partial order
induced by the preorder~$\lmo_L$. One can check that $M_L$ can be effectively
computed from $L$. Moreover, the ordered monoid $M_L$ recognizes $L$.
See~\cite{pinordered} for details. We call $M_L$ the \emph{syntactic ordered
  monoid of $L$} and the associated morphism the \emph{syntactic morphism}.

\medskip
\noindent
{\bf Morphisms and Separation.} When working on separation, we
consider as input two regular languages $L_0,L_1$. It will be
convenient to have a \emph{single} monoid recognizing both of them,
rather than having to deal with two objects. This can always be
assumed without loss of generality as such a monoid can easily be
constructed as follows. Let $M_0,M_1$ be monoids recognizing $L_0,L_1$ together
with the morphisms $\alpha_0,\alpha_1$, respectively. Then, $M_0
\times M_1$ equipped with the componentwise multiplication $(s_0,s_1)
\cdot (t_0,t_1)=(s_0 t_0,s_1 t_1)$ is a monoid that recognizes both
$L_0$ and $L_1$ with the morphism $\alpha : w \mapsto
(\alpha_0(w),\alpha_1(w))$.

\medskip
\noindent
{\bf Alphabet Compatible Morphisms.} In our \sicd-separation
algorithm, it will be convenient to work with morphisms that satisfy
an additional property. A morphism $\alpha: A^* \rightarrow M$ is said
to be \emph{alphabet compatible} if for all $u,v \in A^*$, $\alpha(u)
= \alpha(v) $ implies $\content{u} = \content{v}$. Note that when
$\alpha$ is alphabet compatible, $\content{s}$ is
well defined for all $s \in M$ as the unique subset $B$ of $A$ such that for all $u \in
\alpha^{-1}(s)$, we have $\content{u} = B$ (if $s$ has no preimage then we
simply set $\content{s} = \emptyset$).

To any morphism $\alpha: A^* \rightarrow M$ into a finite monoid $M$,
we associate a morphism $\beta$, called the \emph{alphabet completion}
of $\alpha$, that recognizes all languages recognized by $\alpha$ and
is alphabet compatible. If $\alpha$ is already alphabet compatible,
then $\beta = \alpha$. Otherwise, observe that $2^A$ is a monoid with
union as the multiplication. Hence, we can define $\beta$ as the
morphism:
\[
  \begin{array}{rlcl}
    \beta: & A^* & \to     & M \times 2^A \\
           & w   & \mapsto & (\alpha(w),\content{w}).
  \end{array}
\]
It is straightforward to verify that any language recognized by a morphism
into a finite (ordered) monoid is also recognized by its alphabet completion.

\medskip
\noindent
{\bf Simon's Factorization Forests Theorem.} In several of our proofs,
we make use of a combinatorial result on monoids: Simon's Factorization
Forests Theorem~\cite{simonfacto}. We state this theorem here. For more details on
factorization forests and a proof of the theorem, we refer the reader
to~\cite{bfacto,cfacto,kfacto,tc-handbook15}.

Let $M$ be a finite monoid and $\alpha: A^* \rightarrow M$ a
morphism. An \emph{$\alpha$-factorization forest} is an ordered
unranked tree whose nodes are labeled by words in $A^*$ and such that
for any inner node $x$ with label $w$, if $x_1,\dots,x_n$ are its
children listed from left to right with labels $w_1,\dots,w_n$, then
$w = w_1\cdots w_n$. Moreover, any node in the forest must be of one of
the three following kinds:

\begin{itemize}
\item \emph{leaf nodes}, which are labeled by either a single letter or
  the empty word.
\item \emph{binary nodes}, which have exactly two children.
\item \emph{idempotent nodes}, which have an arbitrary number of
  children whose labels $w_1,\dots,w_n$ satisfy $\alpha(w_1) = \cdots
  = \alpha(w_n) = e$ for some idempotent $e \in M$.
\end{itemize}
If $w \in A^*$, an \emph{$\alpha$-factorization forest for $w$} is an
$\alpha$-factorization forest whose root is labeled by $w$.

\begin{theorem}[\cite{simonfacto,kfacto}] \label{thm:facto}
  For all words $w \in A^*$, there exists an $\alpha$-factorization forest for
  $w$ of height at most $3|M|-1$.
\end{theorem}

\section{\texorpdfstring{\iChains and \iJuns}{Σ\textiinferior-Chains and Σ\textiinferior-Junctures}}
\label{sec:chains}
In this section, we introduce our last tool, the set of \emph{\ichains}. It is
specific to the paper and is central to all our results. Such a set can be
associated to any morphism $\alpha: A^* \rightarrow M$, and the notion is
designed with the separation problem for \sici and \bsci in mind: both
problems can be reduced to the computation of this set.

In this section, we only give the definition of \ichains. We postpone the link
with separation and membership to Section~\ref{sec:generic}. We split the
presentation in two parts. In the first part, we define \emph{\ichains}. In
the second part, we define a refined notion: \emph{\ijuns}. This second notion
carries more information than standard \ichains and is actually more than we
need to make the link with separation. However, we will have to work with this
stronger notion in order to be able to compute \dchains in
Section~\ref{sec:comput}.

\subsection{\texorpdfstring{\iChains}{Σ\textiinferior-Chains}}

\noindent
{\bf \Chains.} Set $M$ as a finite monoid. A \emph{\chain} for $M$ is
a word over the alphabet $M$, \emph{i.e.}, an element of $M^*$. A
remark about notation is in order here. A word is usually denoted as
the concatenation of its letters. Since $M$ is a monoid, this would be
ambiguous here since $st$ could either mean a word with 2 letters $s$
and $t$, or the product of $s$ and $t$ in $M$. To avoid confusion, we
will write $(s_1,\dots,s_n)$ a \chain for $M$ of length $n$. Note that
when $M$ is clear from the context, we will simply speak of \chains,
leaving $M$ implicit.

For all $n \in \nat$, observe that $M^n$, the set of \chains of length
$n$, is a monoid when equipped with the componentwise
multiplication. In the paper, we denote \chains by
$\bar{s},\bar{t},\dots$ and sets of \chains by $\Ss,\Ts,\dots$. As
explained above, given a monoid $M$, we are not interested in all
\chains for $M$, but only in those that carry information with
respect to the logic \siw{i} and some morphism $\alpha: A^*
\rightarrow M$, which we call the \emph{\ichains for $\alpha$}.

\medskip
\noindent
{\bf \iChains.} Fix $i \in \nat$, we begin by defining a set of
\emph{\ikchains} for each fixed quantifier rank $k$. The set of
\ichains will then be the intersection of all sets of \ikchains. 

\smallskip
When $i = 0$, we set by convention $\Csik[\alpha] = M^*$ for all
$k$. Otherwise, when $i \geqslant 1$, we let
\[
  (s_1,\dots,s_n) \in
  \Csik[\alpha] \text{~~if~~} \exists w_1,\dots,w_n \in A^* \text{ with }
  \begin{cases}
    w_1
    \ksieq{i} \cdots \ksieq{i} w_n \text{ and}\\
    \forall j,\ \alpha(w_j)=s_j.
  \end{cases}
\]
We
can now define the \emph{set of 
  \ichains for $\alpha$} as the set $$\Csi[\alpha] =
\bigcap_k\Csik[\alpha].$$

\begin{remark}
  Observe that the set of \ikchains of length $2$ can be viewed as an
  abstraction of \ksieq{i} over the set $M$ with respect to $\alpha$. An
  important observation is that this abstraction is no longer a preorder: in
  general, this is a non-transitive relation. This is because $(r,s)$ and
  $(s,t)$ are \ikchains of length~2 iff there are words $u,v$ mapped to $r,s$
  and $v',w$ mapped to $s,t$, respectively, such that $u\ksieq{i}v$ and
  $v'\ksieq{i}w$. However, $v$ and $v'$ may be completely unrelated.  In
  particular, this means that the whole set of \ikchains carries more
  information than the set of \ikchains of length $2$ only.
\end{remark}

It will often be convenient to speak only of \ichains of a given fixed
length.  For any fixed $n \in \nat$, we let $\Csikn[\alpha]$ be the
set of \ikchains of length~$n$ for $\alpha$, \emph{i.e.},
$\Csikn[\alpha] = \Csik[\alpha] \cap M^n$. We define $\Csin[\alpha]$
similarly. We have the following lemma.

\begin{lemma} \label{fct:chainref}
  For any $i,k,n \in \nat$,
  \[
    \begin{array}{lclcl}
      \Cs_i[\alpha] &\subseteq &\Cs_i^{k+1}[\alpha] &\subseteq
      &\Cs_i^k[\alpha]. \\[.5ex]
      \Csin[\alpha] &\subseteq &\Cs_{i,n}^{k+1}[\alpha] &\subseteq
      &\Csikn[\alpha].
    \end{array}
  \]
  Moreover, for all $i,n \in \nat$, there exists $\kappa_{i,n} \in
  \nat$ such that for any $k \geqslant \kappa_{i,n}$,
  \[
    \begin{array}{lclcl}
      \Csin[\alpha] & = & \Cs_{i,n}^{\kappa_{i,n}}[\alpha] & = & \Csikn[\alpha]
    \end{array}
  \]
\end{lemma}
\begin{proof}
  The first property is immediate from the definitions. The existence of
  $\kappa_{i,n}$ follows from the first property and the fact that for
  all fixed $n$, $M^n$ is a finite set.
\end{proof}

Notice that for a given $k$, the set $\Csikn[\alpha]$ can be computed by brute force,
by calculating all $\kbceq{i}$-classes in $A^*$ (which can be done by enumerating
all the finitely many nonequivalent formulas of rank~$k$ in \siwi). Therefore,
computing (an upper bound on) $\kappa_{i,n}$ immediately yields computability of
$\Csin[\alpha]$. However, while the existence of $\kappa_{i,n}$
is easy to prove, its computation in nontrivial. It may happen
that $\Cs_{i,n}^{k}[\alpha]=\Cs_{i,n}^{k+1}[\alpha]$, but
$\Cs_{i,n}^{k+1}[\alpha]\supsetneq\Cs_{i,n}^{k+2}[\alpha]$. We will obtain
a bound on $\kappa_{i,n}$ as a byproduct of our algorithm for computing \ichains
presented in Section~\ref{sec:comput}.

\medskip
\noindent {\bf Closure Properties.} We finish the definitions by stating
simple closure properties of the sets $\Csik[\alpha]$ and $\Csi[\alpha]$:
closure under subwords, closure under stutter and closure under product. These
three properties are illustrated on an example in Figure~\ref{fig:clos}, where
\chains are represented pictorially: we draw the \chain $(s_1,s_2,\ldots,s_n)$
as
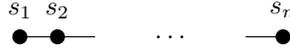
\begin{center}
  \begin{tikzpicture}[every node/.style={circle,anchor=center,inner sep=0pt,minimum size=.5mm,label distance=1.2mm}]
    \fill (-2,0) circle (1mm) node (s1)[label=above:$s_1$]{};
    \fill (-1.5,0) circle (1mm) node (s2)[label=above:$s_2$]{};
    \draw (0,0) node (s3){$\ldots$};
    \fill (1.5,0) circle (1mm) node (sn)[label=above:$s_n$]{};
    \draw (s1) -- (s2) -- ($(s2)+(.5,0)$);
    \draw ($(sn)-(.5,0)$)  -- (sn) ;
  \end{tikzpicture}
\end{center}

\begin{figure}[b]
  \begin{center}
    \begin{tikzpicture}[every node/.style={circle,anchor=center,inner sep=0pt,minimum size=.5mm,label distance=1.2mm}]
      \fill (-1,.5) circle (1mm) node (s1)[label=above:$s_1$]{};
      \fill (-.5,.5) circle (1mm) node (s2)[label=above:$s_2$]{};
      \fill (0,.5) circle (1mm) node (s3)[label=above:$s_3$]{};
      \fill (.5,.5) circle (1mm) node (s4)[label=above:$s_4$]{};
      \fill (1,.5) circle (1mm) node (s5)[label=above:$s_5$]{};
      \fill (1.5,.5) circle (1mm) node (s6)[label=above:$s_6$]{};
      \node[anchor=east] (c1) at (3.5,.5) {$ \in \Csi[\alpha]\quad$};
      \draw (s1) -- (s2) -- (s3) -- (s4) -- (s5)  -- (s6) ;

      \fill (-1,-.5) circle (1mm) node (t1)[label={90:$t_1$}]{};
      \fill (-.5,-.5) circle (1mm) node (t2)[label=above:$t_2$]{};
      \fill (0,-.5) circle (1mm) node (t3)[label=above:$t_3$]{};
      \fill (.5,-.5) circle (1mm) node (t4)[label=above:$t_4$]{};
      \fill (1,-.5) circle (1mm) node (t5)[label=above:$t_5$]{};
      \fill (1.5,-.5) circle (1mm) node (t6)[label=above:$t_6$]{};
      \node[non,anchor=east] (c2) at (3.5,-.5) {$ \in \Csi[\alpha]\quad$};
      \draw (t1) -- (t2) -- (t3) -- (t4) -- (t5)  -- (t6) ;

      \draw[very thick,decorate,decoration={brace}] (c1.north east) to coordinate (cc) (c2.south east);

      \draw[ars] ($(cc)+(0.5,0.0)$) to ($(cc)+(1.8,0.0)$);

      \fill (8.,1.5)  circle (1mm) node (s1)[label=above:$s_1$]{};
      \fill (8.5,1.5) circle (1mm) node (s2)[label=above:$s_2$]{};
      \fill (9.,1.5) circle (1mm) node (s4)[label=above:$s_4$]{};
      \node (c3) at (10.,1.5) {$ \in \Csi[\alpha]\quad$};
      \node[anchor=center] (c3b) at (9,2.3) {\emph{Closure under subwords:}};
      \draw (s1) -- (s2) -- (s4);

      \fill (6.25,-.4) circle (1mm) node (t1)[label=above:$s_1$]{};
      \fill (6.75,-.4) circle (1mm) node (t2)[label=above:$s_2$]{};
      \fill (7.25,-.4) circle (1mm) node (t3)[label=above:$s_2$]{};
      \fill (7.75,-.4) circle (1mm) node (t4)[label=above:$s_2$]{};
      \fill (8.25,-.4) circle (1mm) node (t5)[label=above:$s_3$]{};
      \fill (8.75,-.4) circle (1mm) node (t6)[label=above:$s_4$]{};
      \fill (9.25,-.4) circle (1mm) node (t7)[label=above:$s_5$]{};
      \fill (9.75,-.4) circle (1mm) node (t8)[label=above:$s_6$]{};
      \fill (10.25,-.4) circle (1mm) node (t9)[label=above:$s_6$]{};
      \node (c3) at (11.25,-.4) {$ \in \Csi[\alpha]\quad$};
      \node (c3b) at (9,.4) {\emph{Closure under stutter:}};
      \draw (t1) -- (t2) -- (t3) -- (t4) -- (t5) -- (t6) -- (t7) -- (t8) --(t9);

      \fill (6.2,-2.4)  circle (1mm) node (r1)[label=above:$s_1t_1$]{};
      \fill (7,-2.4)  circle (1mm) node (r2)[label=above:$s_2t_2$]{};
      \fill (7.8,-2.4)  circle (1mm) node (r3)[label=above:$s_3t_3$]{};
      \fill (8.6,-2.4)  circle (1mm) node (r4)[label=above:$s_4t_4$]{};
      \fill (9.4,-2.4)  circle (1mm) node (r5)[label=above:$s_5t_5$]{};
      \fill (10.2,-2.4)  circle (1mm) node (r6)[label=above:$s_6t_6$]{};
      \node (c5) at (11.25,-2.4) {$ \in \Csi[\alpha]\quad$};
      \draw node (c5b) at (9,-1.5) {\emph{Closure under product:}};
      \draw (r1) -- (r2) -- (r3) -- (r4) -- (r5) -- (r6) ;
      \draw[very thick,decorate,decoration={brace}] (5.75,-2.85) to (5.75,2.85);
    \end{tikzpicture}
  \end{center}
  \caption{Closure properties of \ichains (example on $\Csi[\alpha]$)}
  \label{fig:clos}
\end{figure}

Observe first that since the relation $\ksieq{i}$ is transitive for
all $i,k$, the sets $\Csik[\alpha]$ and $\Csi[\alpha]$ are closed
under subwords.

\begin{fact} \label{fct:high}
  Let $i,k \in \nat$ and let $\Xs = \Csi[\alpha]$ or $\Xs = \Csik[\alpha]$. Then $\Xs$
  is closed under subwords. That is, for all $(s_1,\dots,s_n) \in \Xs$ and all $j
  \leqslant n$, we have $(s_1,\dots,s_{j-1},s_{j+1},\dots,s_n) \in \Xs$.
\end{fact}
An interesting consequence of Fact~\ref{fct:high} is that, by
Higman's lemma, $\Csi[\alpha]$ and $\Csik[\alpha]$ are both regular
languages over the alphabet $M$. However, this observation is
essentially useless in our argument as Higman's lemma provides no way
for actually computing a recognizing device for the language
$\Csi[\alpha]$.

Another immediate property of $\Csi[\alpha]$ and $\Csik[\alpha]$ is closure
under duplication of letters (also called stutter).

\begin{fact} \label{fct:dupli}
  Let $i,k \in \nat$ and let $\Xs = \Csi[\alpha]$ or $\Xs = \Csik[\alpha]$. Then,
  $\Xs$ is closed under stutter. That is, for all $(s_1,\dots,s_n) \in \Xs$ and all
  $j \leqslant n$, we have $(s_1,\dots,s_{j},s_{j},\dots,s_n) \in \Xs$.
\end{fact}

Finally, since $\ksieq{i}$ is compatible with the concatenation
operation for any $k$ (see Lemma~\ref{lem:efconcat}, the pre-congruence Lemma), it is immediate
that \ichains of length $n$ are closed under product (\emph{i.e.}, componentwise
multiplication).
\begin{fact} \label{fct:chaincomp}
  For all $i,k,n \in \nat$, both $\Csin[\alpha]$ and $\Csikn[\alpha]$ are
  submonoids of $M^n$. 
\end{fact}

\noindent
This ends the definition of \ichains. This leaves two issues.
\begin{itemize}
\item First, we need to explain the link between the computation of \ichains and
  our decisions problems. We establish this link in Section~\ref{sec:generic}. For
  example, we show that the separation problem for \siwi reduces to the
  computation of all \ichains of length $2$.
\item The second issue is finding an algorithm, which, given a morphism
  $\alpha$, computes the set of associated \ichains. We will present such an
  algorithm for \dchains in Section~\ref{sec:comput}. However, this algorithm
  has to work with a refined notion called \emph{``\/\ijuns''}. We now define this
  notion.
\end{itemize}

\subsection{\texorpdfstring{\iJuns}{Σ\textiinferior-Junctures}}

Our algorithm computes more than we actually need to solve separation. The
crucial information is to determine when several \ichains with the same first
element can be ``synchronized''. To explain what we mean, consider two
\ichains $(s,t_1)$ and $(s,t_2)$ of length $2$. By definition, for all $k$
there exist words $w,w_1,w',w_2$ whose images under $\alpha$ are
$s,t_1,s,t_2$ respectively, and such that $w \ksieq{i} w_1$ and
$w' \ksieq{i} w_2$. In some cases (but not all), it will be possible to
choose $w = w'$ for all $k$. The goal of the notion of \ijuns is to record
the cases in which this is true. The reason why we need to capture this extra
information is that (1) it can be computed inductively, which is not clear for
\ichains, and (2) it contains more information than \ichains do.

\smallskip We first define the generic notion of \emph{\jun}, and then a
specific notion, dedicated to our problem, called \emph{\ijun}.

\medskip\noindent {\bf \Juns.} Let $M$ be a finite monoid. A \emph{\jun} for
$M$ is a pair $(s,\Ss)$ where $s \in M$ and $\Ss \subseteq M^*$ is a set of
\chains. 
If
$(s,\Ss)$ is a \jun and $\bar{t} = (t_1,\dots,t_n)$ is a \chain, we write
$\bar{t} \in (s,\Ss)$ if
\[
  t_1 = s \quad \text{\it and} \quad (t_2,\dots,t_n) \in \Ss.
\]
Thus, a \jun $(s,\Ss)$ abstracts a set of \chains all having the same first
element, namely~$s$. Although we will not use it, it is convenient to view a
\jun as an $M$-labeled tree, where only the root is branching.
\figurename~\ref{fig:juncture0} pictures the \jun
$\bigl(s_0,\big\{(r_1,r_2),(s_1), (t_1,t_2,t_3)\big\}\bigr)$, which abstracts the
set of \chains $\big\{(s_0,r_1,r_2),(s_0,s_1), (s_0,t_1,t_2,t_3)\big\}$ ``synchronized''
at $s_0$.
\begin{figure}[htpb]
  \centering
  \begin{tikzpicture}[every node/.style={circle,inner sep=0pt,minimum size=2mm}]
    \fill (0,0) circle (0.1cm) node (s0)[label=above:$s_0$]{};
    \fill (2,0) circle (0.1cm) node (s1)[label=above:$s_1$]{};
    \fill (2,1) circle (0.1cm) node (r1)[label=above:$r_1$]{};
    \fill (4,1) circle (0.1cm) node (r2)[label=above:$r_2$]{};
    \fill (2,-1) circle (0.1cm) node (t1)[label=above:$t_1$]{};
    \fill (4,-1) circle (0.1cm) node (t2)[label=above:$t_2$]{};
    \fill (6,-1) circle (0.1cm) node (t3)[label=above:$t_3$]{};

    \draw  (s0) -- (r1) -- (r2);
    \draw  (s0) -- (s1) ;
    \draw  (s0) -- (t1) -- (t2) -- (t3);
  \end{tikzpicture}
  \caption{\Jun $\bigl(s_0,\big\{(r_1,r_2),(s_1), (t_1,t_2,t_3)\big\}\bigr)$}
  \label{fig:juncture0}
\end{figure}

If $(s,\Ss)$ and $(t,\Ts)$ are \juns, we write $(s,\Ss) \subseteq (t,\Ts)$
when $s=t$ and $\Ss \subseteq \Ts$, \emph{i.e.}, when
$\{\bar{s} \mid \bar{s} \in (s,\Ss)\} \subseteq \{\bar{t} \mid \bar{t} \in
(t,\Ts)\}$.
In other words, this means that the tree representing $(s,\Ss)$ is obtained
from the tree representing $(t,\Ts)$ by simply removing some branches from the
root.

Note that a \chain is in particular a \jun. For this reason, we use the same
notation for sets of \chains and sets of \juns: $\fR,\fS,\fT,\dots$ If $\fT$ is a set \juns,
we define $\downclos \fT$, the \emph{downset} of $\fT$, as the set:
\[
  \downclos \fT=\big\{(r,\Rs) \mid \exists (s,\Ss) \in \fT,\ (r,\Rs)\subseteq (s,\Ss)\big\}.
\]
In other words, $\downclos \fT$ is the set of \juns represented by trees
obtained by possibly removing some branches to trees in $\fT$.

Finally, for any $n \geqslant 1$, a \jun $(s,\Ss)$ is said to have \emph{length
  $n$} when $\Ss \subseteq M^{n-1}$, \emph{i.e.}, when \chains $\bar{s} \in
(s,\Ss)$ all have the same length $n$. In this paper, we shall only use such \juns,
such as the one pictured in
\figurename~\ref{fig:juncture}.
\begin{figure}[htpb]
  \centering
  \begin{tikzpicture}[every node/.style={circle,inner sep=0pt,minimum size=2mm}]
    \fill (0,0) circle (0.1cm) node (s0)[label=above:$s_0$]{};
    \fill (2,0) circle (0.1cm) node (s1)[label=above:$s_1$]{};
    \fill (4,0) circle (0.1cm) node (s2)[label=above:$s_2$]{};
    \fill (6,0) circle (0.1cm) node (s3)[label=above:$s_3$]{};
    \fill (2,1) circle (0.1cm) node (r1)[label=above:$r_1$]{};
    \fill (4,1) circle (0.1cm) node (r2)[label=above:$r_2$]{};
    \fill (6,1) circle (0.1cm) node (r3)[label=above:$r_3$]{};
    \fill (2,-1) circle (0.1cm) node (t1)[label=above:$t_1$]{};
    \fill (4,-1) circle (0.1cm) node (t2)[label=above:$t_2$]{};
    \fill (6,-1) circle (0.1cm) node (t3)[label=above:$t_3$]{};

    \draw  (s0) -- (r1) -- (r2) -- (r3);
    \draw  (s0) -- (s1) -- (s2) -- (s3);
    \draw  (s0) -- (t1) -- (t2) -- (t3);
  \end{tikzpicture}
  \caption{The \jun
    $\big\{s_0,\big\{(r_1,r_2,r_3),(s_1,s_2,s_3),(t_1,t_2,t_3)\big\}\big\}\in M\times2^{M^3}$}
  \label{fig:juncture}
\end{figure}

Observe that for all $n \geqslant 1$, the set $M \times 2^{M^{n-1}}$ of \juns of
length~$n$ is a monoid for the operation:
\[
  (s,\Ss) \cdot (t,\Ts) = (s\cdot t,\,\Ss \cdot \Ts) = \big(st,\{\bar{s}\bar{t} \in M^{n-1} \mid \bar{s} \in \Ss, \quad \bar{t} \in \Ts\}\big).
\]
As for \chains, we are not interested in all \juns but only in those
that carry information with respect to \siwi for some $i$.
We call \emph{\ijuns} these particular \juns.

\medskip
\noindent {\bf \iJuns.} To define \ijuns, we mimic the definition of
\ichains. Fix $i \geqslant 1$. We begin by defining a set of \ikjuns for each fixed
quantifier rank $k$. For all $k \in \nat$, we define the set $\fCik[\alpha]$
of \emph{\ikjuns for $\alpha$}. Let $(t,\Ts)$ be a \jun. We let
$(t,\Ts) \in \fCik[\alpha]$~if
\begin{itemize}
\item all \chains in $\Ts$ have the same length, say $n-1$, and
\item there exists $w \in A^*$ such that $\alpha(w) = t$, and for
  all \chains $(t_2,\dots,t_n) \in \Ts$, there exist $w_2,\dots,w_n
  \in A^*$ satisfying
  \begin{equation}
    \label{eq:sigma2chain}
    w \ksieq{i} w_2 \ksieq{i} \cdots \ksieq{i} w_n,
  \end{equation}
  and for all $j=2,\dots,n$,
  \begin{equation}
    \label{eq:12}
    \alpha(w_j)=t_j.
  \end{equation}
\end{itemize}
We call such a word $w$ a \emph{$k$-witness} of the \ikjun. With the tree
representation of \ikjuns, as in~\figurename~\ref{fig:juncture}, this means
that we can actually label each node by \emph{two} values: one in the finite
monoid $M$ (the same value as in \figurename~\ref{fig:juncture}), and one in
$A^*$, such that
\begin{itemize}
\item for any node whose labeling in $M$ is $s$ and whose labeling in $A^*$ is
  $u$, we have $\alpha(u)=s$,
\item for any edge from a node labeled $u\in A^*$ to one of its children
  labeled $v\in A^*$, we have $u\ksieq i v$.
\end{itemize}
Thus, a $k$-witness of the \ikjun is a possible word-labeling of the root.

\medskip\noindent
We finally define $\fCi[\alpha]$, the \emph{set of \ijuns for $\alpha$} as
the set
\begin{equation*}
  \fCi[\alpha] = \bigcap_k \fCik[\alpha].
\end{equation*}
An immediate observation, already mentioned above, is that \ijuns carry at
least as much information as \ichains, as stated in the following fact.

\begin{fact}
  \label{fct:recoverchains}
  Let $i \geqslant 1$ and let $(s_1,\dots,s_n)$ be a
  \chain. Then,
  \[
    \begin{array}{rcl}
      (s_1,\dots,s_n) \in \Csi[\alpha] & \text{if and only if} & \big(s_1,\big\{(s_2,\dots,s_n)\big\}\big) \in
                                                               \fCi[\alpha].
    \end{array}
  \]
\end{fact}

Recall that we will restrict ourselves to \juns of a fixed
length $n \geqslant 1$. We denote by $\fCikn[\alpha]$ and $\fCin[\alpha]$
the corresponding restrictions:
\begin{align*}
  \fCikn[\alpha] &= \fCik[\alpha]\cap (M \times 2^{M^{n-1}}), \\
  \fCin[\alpha] &= \fCi[\alpha] \cap (M \times 2^{M^{n-1}}).
\end{align*}
For example, the \ijun of \figurename~\ref{fig:juncture} belongs to~$\fC_{i,4}[\alpha]$.
Observe that Lemma~\ref{fct:chainref} can be generalized to \ijuns, as stated
in the following lemma.

\begin{lemma} \label{fct:chainref2}
  For any $i \geqslant 1$ and $k,n \in \nat$, we have
  \[
    \begin{array}{lclcl}
      \fC_i[\alpha] &\subseteq &\fC_i^{k+1}[\alpha] &\subseteq
      &\fC_i^k[\alpha]. \\[1.2ex]
      \fC_{i,n}[\alpha] &\subseteq &\fC_{i,n}^{k+1}[\alpha] &\subseteq
      &\fC_{i,n}^k[\alpha].
    \end{array}
  \]
  Moreover, for all $i,n \in \nat$, there exists $\ell_{i,n} \in
  \nat$ such that for any $k \geqslant \ell_{i,n}$,
  \[
    \begin{array}{lclcl}
      \fC_{i,n}[\alpha] & = & \fC_{i,n}^{\ell_{i,n}}[\alpha] & = & \fC_{i,n}^k[\alpha].
    \end{array}
  \]
\end{lemma}

Note that it is immediate from the definitions that the bound $\ell_{i,n}$ in
Lemma~\ref{fct:chainref2} is also an upper bound on $\kappa_{i,n}$ in
Lemma~\ref{fct:chainref}. We will obtain an upper bound on $\ell_{2,n}$ when
proving the completeness of our algorithm computing \djuns in
Section~\ref{sec:comput}.

\medskip\noindent
{\bf Closure Properties.} We finish the section by generalizing the
closure properties of \ichains to \ijuns. An illustration of all four
closure properties can be found in Figure~\ref{fig:clos2}.

\begin{figure}[ht]
  \begin{center}
    \begin{tikzpicture}[scale=.8]
      \node[non,anchor=east,align=center] (c1) at (.6,-1)
      {\begin{tikzpicture}[scale=.9,every node/.style={circle}]
          \fill (0,0) circle (0.1cm) node (s0)[label=above:$s_0$]{};
          \fill (1,0.5) circle (0.1cm) node (s1)[label=above:$s_1$]{};
          \fill (2,0.5) circle (0.1cm) node (s2)[label=above:$s_2$]{};
          \fill (3,0.5) circle (0.1cm) node (s3)[label=above:$s_3$]{};
          \fill (4,0.5) circle (0.1cm) node (s4)[label=above:$s_4$]{};
          \fill (1,-.5) circle (0.1cm) node (r1)[label=below:$r_1$]{};
          \fill (2,-.5) circle (0.1cm) node (r2)[label=below:$r_2$]{};
          \fill (3,-.5) circle (0.1cm) node (r3)[label=below:$r_3$]{};
          \fill (4,-.5) circle (0.1cm) node (r4)[label=below:$r_4$]{};

          \node (t) at (5.5,0) {$\in \fCi[\alpha]$};
          \draw  (s0) -- (r1) -- (r2) -- (r3)  -- (r4); 
          \draw  (s0) -- (s1) -- (s2) -- (s3) -- (s4); 
        \end{tikzpicture}};

      \node[non,anchor=east,align=center] (c2) at ($(c1.east)-(0,4)$)
      {\begin{tikzpicture}[scale=.9,every node/.style={circle}]
          \fill (0,0) circle (0.1cm) node (t0)[label=below:$t_0$]{};
          \fill (1,0) circle (0.1cm) node (t1)[label=below:$t_1$]{};
          \fill (2,0) circle (0.1cm) node (t2)[label=below:$t_2$]{};
          \fill (3,0) circle (0.1cm) node (t3)[label=below:$t_3$]{};
          \fill (4,0) circle (0.1cm) node (t4)[label=below:$t_4$]{};

          \node (t) at (5.5,0) {$\in \fCi[\alpha]$};
          \draw  (t0) -- (t1) -- (t2) -- (t3) -- (t4); 
        \end{tikzpicture}};

      \draw[very thick,decorate,decoration={brace}] (c1.north east) to
      coordinate (cc) (c2.south east);

      \draw[ars] ($(cc)+(0.4,0.0)$) to ($(cc)+(1,0.0)$);

      \node[non,anchor=south west,align=left] (cl0) at (1.9,3.5)
      {\hspace*{11ex}\emph{Closure under subsets:}\\[2ex]\begin{tikzpicture}
          \fill (0,0) circle (0.1cm) node (s0)[label=above:$s_0$]{};
          \fill (1,0) circle (0.1cm) node (s1)[label=above:$s_1$]{};
          \fill (2,0) circle (0.1cm) node (s2)[label=above:$s_2$]{};
          \fill (3,0) circle (0.1cm) node (s3)[label=above:$s_3$]{};
          \fill (4,0) circle (0.1cm) node (s4)[label=above:$s_4$]{};
          \node (t) at (4.3,-0.2) {$\in \fCi[\alpha]$};
          \draw  (s0) -- (s1) -- (s2) -- (s3) -- (s4); 
        \end{tikzpicture}};

      \node[non,anchor=south west,align=left] (cl1) at ($(cl0.west)-(0,4.5)$)
      {\hspace*{10ex}\emph{Closure under subwords:}\\[2ex]
        \begin{tikzpicture}[scale=1,every node/.style={circle}]
          \fill (0,0) circle (0.1cm) node (s0)[label=above:$s_0$]{};
          \fill (1,0.5) circle (0.1cm) node (s1)[label=above:$s_2$]{};
          \fill (2,0.5) circle (0.1cm) node (s2)[label=above:$s_4$]{};
          \fill (1,-.5) circle (0.1cm) node (r1)[label=below:$r_1$]{};
          \fill (2,-.5) circle (0.1cm) node (r2)[label=below:$r_2$]{};

          \node (t) at (2.3,-0.4) {$\in \fCi[\alpha]$};
          \draw  (s0) -- (r1) -- (r2);
          \draw  (s0) -- (s1) -- (s2);
        \end{tikzpicture}};

      \node[non,anchor=west,align=left] (cl2) at ($(cl1.west)-(0,3.5)$)
      {\hspace*{12ex}\emph{Closure under stutter:}\\[2ex]
        \begin{tikzpicture}[scale=1,every node/.style={circle}]
          \fill (0,0) circle (0.1cm) node (s0)[label=above:$s_0$]{};
          \fill (1,0.5) circle (0.1cm) node (s1)[label=above:$s_0$]{};
          \fill (2,0.5) circle (0.1cm) node (s2)[label=above:$s_1$]{};
          \fill (3,0.5) circle (0.1cm) node (s3)[label=above:$s_2$]{};
          \fill (4,0.5) circle (0.1cm) node (s4)[label=above:$s_3$]{};
          \fill (5,0.5) circle (0.1cm) node (s5)[label=above:$s_3$]{};
          \fill (6,0.5) circle (0.1cm) node (s6)[label=above:$s_4$]{};
          \fill (1,-.5) circle (0.1cm) node (r1)[label=below:$r_1$]{};
          \fill (2,-.5) circle (0.1cm) node (r2)[label=below:$r_1$]{};
          \fill (3,-.5) circle (0.1cm) node (r3)[label=below:$r_1$]{};
          \fill (4,-.5) circle (0.1cm) node (r4)[label=below:$r_2$]{};
          \fill (5,-.5) circle (0.1cm) node (r5)[label=below:$r_3$]{};
          \fill (6,-.5) circle (0.1cm) node (r6)[label=below:$r_4$]{};

          \node (t) at (6.1,0) {$\in \fCi[\alpha]$};
          \draw  (s0) -- (r1) -- (r2) -- (r3)  -- (r4) -- (r5) -- (r6); 
          \draw  (s0) -- (s1) -- (s2) -- (s3) -- (s4) -- (s5) -- (s6); 
        \end{tikzpicture}};

      \node[non,anchor=west,align=left] (cl3) at ($(cl2.west)-(0,4)$)
      {\hspace*{11ex}\emph{Closure under product:}\\[2ex]
        \begin{tikzpicture}[scale=1,every node/.style={circle}]
          \fill (0,0) circle (0.1cm) node (s0)[label=above:$s_0t_0$]{};
          \fill (1,0.5) circle (0.1cm) node (s1)[label=above:$s_1t_1$]{};
          \fill (2,0.5) circle (0.1cm) node (s2)[label=above:$s_2t_2$]{};
          \fill (3,0.5) circle (0.1cm) node (s3)[label=above:$s_3t_3$]{};
          \fill (4,0.5) circle (0.1cm) node (s4)[label=above:$s_4t_4$]{};
          \fill (1,-.5) circle (0.1cm) node (r1)[label=below:$r_1t_1$]{};
          \fill (2,-.5) circle (0.1cm) node (r2)[label=below:$r_2t_2$]{};
          \fill (3,-.5) circle (0.1cm) node (r3)[label=below:$r_3t_3$]{};
          \fill (4,-.5) circle (0.1cm) node (r4)[label=below:$r_4t_4$]{};

          \node (t) at (4.1,0) {$\in \fCi[\alpha]$};
          \draw  (s0) -- (r1) -- (r2) -- (r3)  -- (r4); 
          \draw  (s0) -- (s1) -- (s2) -- (s3) -- (s4); 
        \end{tikzpicture}};
      \draw[very thick,decorate,decoration={brace}]  to (cl0.north west)
      coordinate (cc)  (cl3.south west);
    \end{tikzpicture}
  \end{center}
  \caption{Closure properties of \ijuns (example on $\fCi[\alpha]$)}
  \label{fig:clos2}
\end{figure}

The first property we state is closure under subsets.
\begin{fact}\label{fct:subsets} Let $i,k \geqslant 1$ and let $\fX = \fCi[\alpha]$
  or $\fX = \fCik[\alpha]$. Then, $\fX$ is closed under the following operation:
  for all $(r,\Rs) \in \fX$ and all $\Rs' \subseteq \Rs$, we have $(r,\Rs') \in \fX$. In
  other words, we have $\downclos \fX = \fX$.
\end{fact}

We now generalize closure under subwords to \ijuns.

\begin{fact}\label{fct:high2} Let $i,k \geqslant 1$ and let $\fX = \fCi[\alpha]$
  or $\fX = \fCik[\alpha]$. Then, $\fX$ is closed under the following
  operation: let $(r,\Rs) \in \fX$, and let $\Rs'$ be a set of \chains of
  the same length that are all subwords of \chains in $\Rs$. Then
  $(r,\Rs') \in \fX$.
\end{fact}

We next generalize closure under stutter to \ijuns.

\begin{fact}\label{fct:dupli2}
  Let $i,k \geqslant 1$ and let $\fX = \fCi[\alpha]$ or $\fX = \fCik[\alpha]$. Then
  $\fX$ is closed under the following operation: let $(r,\Rs) \in \fX$, and
  let $\Rs'$ be a set of \chains of the same length, each
  of the form $r^j\bar{r}'$, where $j \geqslant 0$ and $\bar{r}'$ is a stutter
  of some \chain in $\Rs$. Then $(r,\Rs') \in \fX$.
\end{fact}

It remains to generalize closure under product.

\begin{fact}\label{fct:chaincomp2}
  For all $k,n \in \nat$,   $\fCin[\alpha]$ and $\fCikn[\alpha]$ are
  submonoids of $M \times 2^{M^{n-1}}$.
\end{fact}

\section{\texorpdfstring{Generic Results: From \iChains to Separation and Membership}{Generic Results: From {Σ\textiinferior-Chains} to Separation and Membership}}
\label{sec:generic}
In this section, we make explicit the connection between \ichains and our two
decision problems: membership and separation. For all levels in the hierarchy,
we prove that both problems can be reduced to the computation of specific
information about the set of \ichains associated to a morphism recognizing both
input languages. Of course, the amount of required information depends on
whether we consider \siwi or \bswi, and on whether we consider membership or
separation. Note that in order to be stated and proved, all these theorems
only require \ichains: \ijuns are not needed. The section is organized into three~parts.
\begin{itemize}
\item In the first one, we explain the most immediate link: separation for
  \siwi reduces to the computing all \ichains of length $2$.
\item In the second part, we prove that deciding membership for \siwi requires
  less information: only the \qchains{i-1} of length~$2$ are needed.
\item Finally, in the last part we prove reductions for \bswi.
\end{itemize}

\subsection{\texorpdfstring{The Separation Problem for \siwi}{The Separation
    Problem for {Σ\textiinferior(<)}}}

\begin{theorem} \label{thm:sep}
  Let $L_1,L_2$ be regular languages that are both recognized by a
  morphism $\alpha: A^* \rightarrow M$ into a finite monoid $M$ and let
  $F_1,F_2 \subseteq M$ be the corresponding accepting sets. Set $i \in \nat$. Then,
  the following properties hold:
  \begin{enumerate}
  \item\label{item:11} $L_1$ is $\siwi$-separable from $L_2$ iff for all $s_1,s_2 \in
    F_1,F_2$, we have $(s_1,s_2) \not\in \Csi[\alpha]$.
  \item\label{item:12} $L_1$ is $\piwi$-separable from $L_2$ iff for all $s_1,s_2 \in
    F_1,F_2$, we have $(s_2,s_1) \not\in \Csi[\alpha]$.
  \end{enumerate}
\end{theorem}

Theorem~\ref{thm:sep} reduces \siwi-separation to finding an algorithm that,
given a morphism~$\alpha$, computes all the associated \ichains of length $2$.
This computation is simple when $i=1$, and actually already
known~\cite{pvzmfcs13}. In Section~\ref{sec:comput}, we present
an algorithm for the case $i=2$. In fact, we do not compute \ichains directly:
our algorithm computes the more general set of \ijuns, $\fCi[\alpha]$, and
\ichains are then recovered from this set using Fact~\ref{fct:recoverchains}.
This makes Theorem~\ref{thm:sep} effective for $i \leqslant 2$. The problem has
also been solved recently for $i \geqslant 3$, although the proof is much more
involved~\cite{pseps3}. We finish this section with the proof of
Theorem~\ref{thm:sep}.

\begin{proof}[of Theorem~\ref{thm:sep}]
  We prove Item~\ref{item:11}. Item~\ref{item:12} is obtained by symmetry. Assume first that
  $L_1$ is $\siwi$-separable from $L_2$ and let $K$ be a separator.
  By contradiction, suppose that there exist $s_1,s_2 \in F_1,F_2$ such
  that $(s_1,s_2) \in \Csi[\alpha]$. By definition, we know that $K$ can be defined by
  a \siwi formula. Let $k$ be its quantifier rank. By hypothesis, we have
  $(s_1,s_2) \in \Csik[\alpha]$ so that  there exist $w_1,w_2$ mapped by $\alpha$ to
  $s_1,s_2$ respectively, such that $w_1 \ksieq{i} w_2$. In particular, we obtain
  $w_1 \in L_1 \subseteq K$ and $w_2 \in L_2$. Moreover, by choice of
  $k$ and since $w_1 \ksieq{i} w_2$, we also have $w_2 \in K$. This is a
  contradiction since $K$ is by hypothesis a separator, so it cannot
  intersect $L_2$.

  It remains to prove the other direction. Assume that for all
  $s_1,s_2 \in F_1,F_2$, we have $(s_1,s_2) \not\in \Csi[\alpha]$ and let
  $\ell=\kappa_{i,2}$ be as defined in Lemma~\ref{fct:chainref}, that is, such
  that $\Csitwo[\alpha] = \Cs_{i,2}^{\kappa_{i,2}}[\alpha]$. We claim that the
  language
  $$K = \{w \mid \exists w_ 1\in L_1 \text{ s.t. } w_ 1\sieq{\ell}{i} w\},$$ which
  is \siwi-definable by Fact~\ref{fct:definition}, is a separator. Indeed, $K$
  clearly contains $L_1$. If $K$ intersects $L_2$, then by definition of
  $\ell$, there would exist $s_1,s_2 \in F_1,F_2$ such that
  $(s_1,s_2) \in \Csi[\alpha]$, which is false by hypothesis.
\end{proof}

\begin{remark}
  \label{rq:rank-sep}
  Note that the above proof of Theorem~\ref{thm:sep} shows that if two
  languages recognized by~$\alpha$ are \siwi-separable, then they are
  separable by a \siwi formula of rank at most~$\kappa_{i,2}$. In other words,
  the rank $k$ at which the sets $\Csitwo^k[\alpha]$ stabilize is an upper bound
  for the rank of possible separators of languages recognized by $\alpha$.
\end{remark}
\subsection{\texorpdfstring{The Membership Problem for \siwi}{The Membership Problem for {Σ\textiinferior}(<)}}

We now prove that solving membership for \siw{i} requires less
information than separation: only the \qchains{i-1} of length~$2$ need to
be computed.

\begin{theorem} \label{thm:caracsig}
  Let $i \geqslant 1$ and let $L$ be a regular language and $\alpha: A^*
  \rightarrow M$ be its syntactic morphism. For all $i \geqslant 1$, $L$ is
  definable in \siw{i} if and only if $M$ satisfies the following
  property:
  \begin{equation}
    s^{\omega} \lmo s^{\omega}ts^{\omega} \quad \text{for all $(t,s) \in \Cslev{i-1}[\alpha]$}. \label{eq:sig}
  \end{equation}
\end{theorem}

It follows from Theorem~\ref{thm:caracsig} that it suffices to compute the
\qchains{i-1} of length $2$ in order to decide whether a language is definable
in \siw{i}. Also observe that when $i=1$, by definition we have
$(t,1_M) \in \Cslev0[\alpha]$ for all $t \in M$. Therefore,
Equation~\eqref{eq:sig} implies that $1_M \lmo t$ for all $t\in M$.
Conversely, multiplying this inequality on the left and on the right by
$s^\omega$ yields back~\eqref{eq:sig} for all $s,t\in \Cslev0[\alpha]$.
Consequently, Equation~\eqref{eq:sig} may be rephrased as $1_M \lmo t$ for all
$t \in M$, which is the already known equation for \siwu~\cite{pwdelta}.
Similarly, when $i=2$, \eqref{eq:sig} can be rephrased as
$s^{\omega} \lmo s^{\omega}ts^{\omega}$ whenever $t$ is a `subword' of $s$,
which is the previously known equation for \siwd~\cite{pwdelta,bfacto}.

\smallskip Observe that by definition of \piw{i} and \dew{i}, we get
characterizations for these classes as immediate corollaries: recall that a
language is \piw{i}-definable if its complement is \siw{i}-definable, and that
it is \dew{i}-definable if it is both \siw{i}-definable and \piw{i}-definable.

\begin{corollary} \label{cor:caracpi}
  Let $L$ be a regular language and let $\alpha: A^* \rightarrow M$ be its
  syntactic morphism. For all $i \gmo 1$, the following properties hold:
  \begin{itemize}
  \item $L$ is definable in \piw{i} iff $M$ satisfies $s^{\omega}
    \gmo s^{\omega}ts^{\omega}$ for all $(t,s) \in \Cslev{i-1}[\alpha]$.
  \item $L$ is definable in \dew{i} iff $M$ satisfies $s^{\omega}
    = s^{\omega}ts^{\omega}$ for all $(t,s) \in \Cslev{i-1}[\alpha]$.
  \end{itemize}
\end{corollary}

It now remains to prove Theorem~\ref{thm:caracsig}. For the proof, we
assume that $i \geqslant 2$ (a proof for the case $i = 1$ can be
found in~\cite{pwdelta}). We begin with the simpler `only if'
direction, which is an application of Lemma~\ref{lem:siprop} and is
stated in the next proposition.

\begin{proposition} \label{prop:signec}
  Let $L$ be a \siw{i}-definable language and let $\alpha : A^*
  \rightarrow M$ be its syntactic morphism. Then $\alpha$
  satisfies~\eqref{eq:sig}.
\end{proposition}

\begin{proof}
  By hypothesis, $L$ is defined by some \siw{i} formula $\varphi$. Let $k$ be
  its quantifier rank. Let $(t,s) \in \Cs_{i-1}[\alpha]$. We need to prove
  that $s^{\omega} \lmo s^{\omega}ts^{\omega}$. Since $(t,s) \in
  \Cs_{i-1}[\alpha]$, by definition, there exist $v,u$ such that $\alpha(v) =
  t$, $\alpha(u) = s$ and $v \ksieq{i-1} u$.  By the \sici-Property Lemma (Lemma~\ref{lem:siprop}), we
  immediately obtain
  \[
    u^{2^k\omega}\cdot u^{2^k\omega} \ksieq{i}u^{2^k\omega} \cdot v \cdot
    u^{2^k\omega}.
  \]
  It follows from the Pre-congruence Lemma (Lemma~\ref{lem:efconcat}) that for any $w_1,w_2 \in A^*$
  we have:
  \begin{equation*} w_1 \cdot u^{2^k\omega}\cdot
    u^{2^k\omega} \cdot w_2 \ \ \ksieq{i}\ \  w_1 \cdot u^{2^k\omega} \cdot v
    \cdot u^{2^k\omega} \cdot w_2.
  \end{equation*}
  By choice of $k$ and definition of $\ksieq{i}$, this means that $w_1 \cdot u^{2^k\omega} \cdot w_2 \in L$
  implies that $w_1 \cdot u^{2^k\omega} v u^{2^k\omega} \cdot w_2 \in
  L$. By definition of the syntactic preorder, this means that $s^{\omega}
  \lmo s^{\omega}ts^{\omega}$.
\end{proof}

It now remains to prove the harder `if' direction of
Theorem~\ref{thm:caracsig}. We use induction to construct a formula
for the language $L$. We rely on Simon's Factorization Forest Theorem
for the induction, which we state in the following proposition.

\begin{proposition} \label{prop:signec2}
  Let $i\geqslant 2$ and let $\alpha: A^* \rightarrow M$ be a morphism into a finite monoid $M$
  that satisfies~\eqref{eq:sig}. Then for all $h \geqslant 1$ and all $s \in
  M$, there exists a $\siwi$ formula $\varphi$ such that for all $w \in A^*$:
  \begin{itemize}
  \item if $w \models \varphi$ then $s \leqslant \alpha(w)$.
  \item if $\alpha(w) = s$ and $w$ admits an $\alpha$-factorization
    forest of height at most $h$ then $w \models \varphi$.
  \end{itemize}
\end{proposition}

Assume for now that Proposition~\ref{prop:signec2} holds and let $L$ be
a regular language whose syntactic morphism $\alpha: A^*
\rightarrow M$ satisfies~\eqref{eq:sig}. Given $h = 3|M|-1$, for all
$s \in M$, we denote by $\varphi_s$ the \sici formula associated to
$s$ by Proposition~\ref{prop:signec2}. Since, by
Theorem~\ref{thm:facto}, all words admit an $\alpha$-factorization
forest of height at most $3|M| - 1$, we have
\begin{enumerate}[ref=(\arabic*)]
\item\label{item:4} if $w \models \varphi_s$ then $s \leqslant \alpha(w)$.
\item\label{item:50} if $\alpha(w) = s$ then $w \models \varphi_s$.
\end{enumerate}
Let $F$ be the accepting set of $L$ and define $\varphi = \bigvee_{s
  \in F} \varphi_s$. By Item~\ref{item:50} above, we have $L \subseteq \{w \mid w
\models \varphi\}$. Moreover, by definition of recognizability by an
ordered monoid, the set $F$ is upward closed, that is, if $s\in F$ and $s\leqslant t$
then $t\in F$. Hence, Item~\ref{item:4} above implies that $ \{w \mid w
\models \varphi\} \subseteq L$. We conclude that $\varphi$ defines $L$.
This finishes the proof of Theorem~\ref{thm:caracsig}. It now remains
to prove Proposition~\ref{prop:signec2}.

\begin{proof}[of Proposition~\ref{prop:signec2}]
  Set $h \geqslant 1$ and $s \in M$. We construct the formula by induction on
  $h$. Assume first that $h = 1$. Note that the words having an
  $\alpha$-factorization forest of height at most $1$ are either single
  letters or the empty word. Consider the language
  $L_s = \{w \mid |w| \leqslant 1 \text{ and } \alpha(w) = s\}$. Since $L_s$ is
  finite, it can be defined by a \siwi formula $\varphi$ (indeed, since
  $i \geqslant 2$, for any word $w$ one can easily define a $\siwd$ formula whose
  only model is $w$). By definition, $\varphi$ satisfies the conditions of
  Proposition~\ref{prop:signec2}.

  \smallskip
  Assume now that $h > 1$. There are two cases depending on whether $s$
  is idempotent or not. We treat the idempotent case (the other case is
  essentially a simpler version of this proof). Hence we assume that $s$ is an
  idempotent, that we denote by $e$. We first construct $\varphi$ and then prove that it
  satisfies the conditions of the proposition. It is defined as the
  disjunction of several formulas that we define first.

  \medskip\noindent
  {\bf Using Induction.} For all $t \in M$, one can use induction to
  construct a \siwi formula $\psi_t$ such that for all $w \in A^*$,
  \begin{itemize}
  \item if $w \models \psi_t$ then $t \leqslant \alpha(w)$.
  \item if $\alpha(w) = t$ and $w$ admits an $\alpha$-factorization
    forest of height at most $(h-1)$, then $w \models \psi_t$.
  \end{itemize}
  By restricting quantifications, one can modify each of these formulas
  to construct two other formulas $\psi^\ell_t(x)$ and $\psi^r_t(x)$
  both having a single free variable $x$ and such that:
  \begin{itemize}
  \item $w,x \models \psi^\ell_t(x)$ iff the prefix $u$ of $w$
    obtained by keeping only positions $y < x$ satisfies $\psi_t$.
  \item $w,x \models \psi^r_t(x)$ iff the suffix $v$ of $w$
    obtained by keeping only positions $y \geqslant x$ satisfies $\psi_t$.
  \end{itemize}
  Note that these formulas do not have extra quantifiers, so that they also
  belong to $\siwi$.

  \medskip
  \noindent
  {\bf Using \pic{i-1}.} Recall that by Lemma~\ref{fct:chainref}, there
  exists an integer $\kappa$ such that for all $k \geqslant \kappa$:
  \[
    \Csgen{i-1}{k}{2}[\alpha] = \Cs_{i-1,2}[\alpha]
  \]
  Consider the language
  $$K = \bigcup_{w \in \alpha^{-1}(e)} \{u \mid u \sieq{\kappa}{i-1} w\}.$$
  By choice of $\kappa$, for any $u \in K$, we have
  $(\alpha(u),e) \in \Csgen{i-1}{\kappa}{2}[\alpha]=\Cs_{i-1,2}[\alpha]$. Since
  $e = e^2$, one may use Equation~\eqref{eq:sig} to obtain that for all $u \in K$:
  \begin{equation}
    \label{eq:eq}
    e \leqslant e\cdot \alpha(u)\cdot e.
  \end{equation}
  Moreover,  by the dual version of Fact~\ref{fct:definition}, $K$ can be defined by a \piw{i-1} formula
  $\Gamma$ (in particular $\Gamma$ is \siwi).
  We define $\Gamma(x,y)$ as the formula with two free variables $x,y$
  such that $w,x,y \models \Gamma(x,y)$ if and only if $x < y$ and the infix
  $u$ obtained by keeping all positions $z$ in $w$ such that $x \leqslant z < y$
  satisfies $\Gamma$. Note again that this formula can be chosen in \siwi.

  \medskip
  \noindent
  {\bf Definition of $\varphi$.} Finally, we can define the desired
  formula. It is the disjunction of three subformulas. Intuitively, the first
  one captures words having an $\alpha$-factorization forest of height at most
  $h-1$, the second one, words having an $\alpha$-factorization forest of
  height $h$ and whose root is a binary node, and the third one, words with an
  $\alpha$-factorization forest of height $h$ and whose root is an idempotent
  node.
  \[
    \varphi = \psi_e \vee \left(\bigvee_{{t_1}{t_2} = e} \exists x\ \psi_{t_1}^\ell(x)
      \wedge \psi_{t_2}^r(x)\right)  \vee \left(\exists x\exists y\ x < y \wedge
      \psi_e^\ell(x) \wedge
      \Gamma(x,y) \wedge \psi_e^r(y)\right)
  \]
  Note that by definition, $\varphi$ is a
  \siwi formula. We need to prove that it satisfies the conditions of the
  proposition.

  Choose some $w \in A^*$ and assume first that $w \models \varphi$. We need to
  prove that $e \leqslant \alpha(w)$.
  \begin{itemize}
  \item If $w \models \psi_e$, then this is by definition of $\psi_e$.
  \item If
    $w \models \exists x\ \psi^{\ell}_{t_1}(x) \wedge \psi^r_{t_2}(x)$ for
    $t_1{t_2} = e$, then by definition, $w = w_1w_2$ with
    $t_1 \leqslant \alpha(w_1)$ and ${t_2} \leqslant \alpha(w_2)$. It follows that
    $e = t_1{t_2} \leqslant \alpha(w_1w_2) = \alpha(w)$.
  \item Finally, if
    $w \models \exists x\exists y\ x < y \wedge \psi_e(x) \wedge \Gamma(x,y)
    \wedge \psi_e(y)$,
    we obtain that $w = w_1uw_2$ with $e \leqslant \alpha(w_1)$, $u \in K$ and
    $e \leqslant \alpha(w_2)$. By~\eqref{eq:eq}, we know that
    $e \leqslant e\alpha(u)e \leqslant \alpha(w_1uw_2) = \alpha(w)$, which terminates
    this direction.
  \end{itemize}

  Conversely, assume that $\alpha(w) = e$ and that $w$ admits an
  $\alpha$-factorization forest of height at most $h$. We have to prove that
  $w$ satisfies $\varphi$. There are again three cases.
  \begin{itemize}
  \item First, if $w$ has an $\alpha$-factorization forest of height at most
    $h-1$, then $w \models \psi_e$, so $w \models \varphi$.
  \item Second, if $w$ admits an $\alpha$-factorization forest of height $h$
    whose root is a binary node, then $w = w_1w_2$ with $w_1,w_2$ admitting
    forests of height at most $h-1$. Set $t_1= \alpha(w_1)$ and
    ${t_2} = \alpha(w_2)$. Observe that $t_1t_2 = \alpha(w) = e$. By the
    induction hypothesis and definition of the formulas $\psi_t$, we have
    $w_1 \models \psi_{t_1}$ and $w_2 \models \psi_{t_2}$, hence
    $w \models \exists x\ \psi_{t_1}^\ell(x) \wedge \psi_{t_2}^r(x)$. It
    follows that $w \models \varphi$ since $t_1t_2 = e$.
  \item Finally, if $w$ admits an $\alpha$-factorization forest of height $h$
    whose root is an idempotent node, then $w = w_1uw_2$ with
    $\alpha(w_1) = \alpha(u) = \alpha(w_2) = e$ and $w_1,w_2$ admitting
    forests of height at most $h-1$. It follows that $w_1 \models \psi_{e}$
    and $w_2 \models \psi_e$. Moreover, since $\alpha(u) = e$, it is
    immediate that $u \in K$, hence $u \models \Gamma$. We conclude that
    $w \models \exists x\exists y\ x < y \wedge \psi_e^\ell(x) \wedge
    \Gamma(x,y) \wedge \psi_e^r(y)$, whence $w \models \varphi$.
  \end{itemize}
  This concludes the proof of Proposition~\ref{prop:signec2}.
\end{proof}

\subsection{\texorpdfstring{Separation and Membership for \bswi}{Separation and Membership for B{Σ\textiinferior}(<)}}

In this last part, we prove that being able to compute more
information about the set of \ichains yields solutions to both
separation and membership for \bswi. What is needed is a
property called \emph{alternation} that we define now.

\medskip
\noindent
{\bf Alternation.} Let $M$ be a finite monoid. We say that a \chain
$(s_1,\dots,s_n) \in M^*$ has \emph{alternation} $\ell$ if there are
exactly $\ell$ indices $i$ such that $s_i \neq s_{i+1}$. We say that a
set of \chains $\Ss$ has \emph{bounded alternation} if there exists a
bound $\ell \in \nat$ such that all \chains in $\Ss$ have alternation
at most $\ell$.

\begin{theorem} \label{thm:sepbc}
  Let $L_1,L_2$ be regular languages, both recognized by the same
  morphism $\alpha: A^* \rightarrow M$ into a finite monoid $M$ and let
  $F_1,F_2 \subseteq M$ be their respective accepting sets. Let $i \in \nat$. Then
  $L_1$ is $\bswi$-separable from $L_2$ if and only if for all $s_1, s_2
  \in F_1,F_2$, $s_1 \neq s_2$ and $\Csi[\alpha] \cap \{s_1,s_2\}^*$ has
  bounded alternation.
\end{theorem}

Theorem~\ref{thm:sepbc} reduces the separation problem for \bswi to
finding an algorithm which, given a morphism $\alpha$, computes all
pairs $(s_1,s_2) \in M^2$ such that $\Csi[\alpha] \cap \{s_1,s_2\}^*$
has bounded alternation. The problem has been solved when $i=1$
in~\cite{pvzmfcs13}. Above $i=1$, the problem remains open, even when
$i=2$. Note that due to closure of $\Csi[\alpha]$ under subwords,
$\Csi[\alpha] \cap \{s_1,s_2\}^*$ has unbounded alternation if and only if it
contains the language of all \chains  $(s_1,s_2,s_1,s_2,\ldots,s_1,s_2)$, that we
denote by~$(s_1,s_2)^*$.

\smallskip
Before proving Theorem~\ref{thm:sepbc}, we establish
a simple corollary
which states that solving \emph{membership} for \bswi requires slightly less
information. This statement will allow us to solve membership for \bswd in
Section~\ref{sec:caracbc}.

\begin{corollary} \label{cor:membc}
  Let $L$ be a regular language and let $\alpha: A^* \rightarrow M$ be
  its syntactic morphism. Then $L$ is definable in \bswi if and only if
  $\Csi[\alpha]$ has bounded alternation.
\end{corollary}

\begin{proof}
  Recall that $L$ is \bswi-definable  iff $L$ is \bswi-separable
  from its complement. We prove both directions by contrapositive. Let
  $F=\alpha(L)$ be the accepting set of~$L$.

  Assume first that $L$ is not definable in \bswi. By
  Theorem~\ref{thm:sepbc}, this means that there exist $s,t \in M$
  such that $s \in F$ and $t \not\in F$ and $\Csi[\alpha] \cap
  \{s,t\}^*$ has unbounded alternation. Hence $\Csi[\alpha]$ has
  unbounded alternation.

  For the converse, we use the fact that $\alpha$ is the syntactic morphism
  of~$L$. Assume that $\Csi[\alpha]$ has unbounded alternation. By definition
  and since $\Csi[\alpha]$ is closed under subwords, this means that there
  exist $s,t \in M$ such that $\Csi[\alpha] \cap \{s,t\}^*$ has unbounded
  alternation. Since $\alpha$ is the syntactic morphism of $L$, there exist
  $r,r' \in M$ such that either $rsr' \in F$ and $rtr' \not\in F$ or
  $rtr' \in F$ and $rsr' \not\in F$. In both cases,
  $\Csi[\alpha] \cap \{rsr',rtr'\}^*$ has unbounded alternation, since \ichains
  are closed under product. By Theorem~\ref{thm:sepbc}, it follows that $L$ is
  not \bswi-separable from its complement, whence it is not definable
  in~\bswi.
\end{proof}

It remains to prove Theorem~\ref{thm:sepbc}, which we do in the rest of this section.

\begin{proof}[of Theorem~\ref{thm:sepbc}]
  There are two directions, both proved by
  contrapositive.

  Assume first that $L_1$ is not \bswi-separable from $L_2$. We have to find
  $s_1,s_2\in F_1,F_2$ such that $s_1=s_2$, or such that
  $\Csi[\alpha] \cap \{s_1,s_2\}^*$ has unbounded alternation. Using
  Fact~\ref{fct:definition2}, for all $k$, one can find $w_{1,k} \in L_1$ and
  $w_{2,k} \in L_2$ such that $w_{1,k} \kbceq{i} w_{2,k}$. Since $M$ is finite, we
  may assume without loss of generality that there exist $s_1,s_2 \in M$ such
  that for all $k$, $\alpha(w_{1,k}) = s_1$ and $\alpha(w_{2,k}) = s_2$. Observe
  that by definition $s_1 \in F_1$ and $s_2 \in F_2$. If $s_1=s_2$, then we are
  done. Otherwise, $s_1\neq s_2$ and we prove that
  $\Csi[\alpha] \cap \{s_1,s_2\}^*$ has unbounded alternation. Indeed, for all
  $k$, we have
  \[
    w_{1,k} \ksieq{i} w_{2,k} \ksieq{i} w_{1,k} \ksieq{i} w_{2,k} \ksieq{i} w_{1,k}
    \ksieq{i} w_{2,k} \ksieq{i} \cdots
  \]
  Hence by definition, $(s_1,s_2)^* \subseteq \Csi[\alpha]$ which terminates the
  proof of this direction.

  Conversely, assume that there exist $s_1 \in F_1$ and $s_2 \in F_2$ such
  that $\Csi[\alpha] \cap \{s_1,s_2\}^*$ has unbounded alternation. We prove
  that $L_1$ and $L_2$ are not \bswi-separable. More precisely, we show that
  for all $k \in \nat$ there exist $w_1 \in L_1$ and $w_2 \in L_2$ such that
  $w_{1} \kbceq{i} w_{2}$. The result will then follow from
  Fact~\ref{fct:definition2} again.

  Set $k \in \nat$ and set $n$ as the number of equivalence classes of
  $\kbceq{i}$ (recall $\kbceq{i}$ has finite index). Consider the \chain
  $(s_1,s_2)^{n+1} \in \Csik[\alpha]$, that is, the \chain
  $(s_1,s_2,s_1,s_2,\ldots,s_1,s_2)$ of length $2(n+1)$. By definition there exist words
  $u_1,\dots,u_{n+1}$ mapped to $s_1$ under $\alpha$ and $v_1,\dots,v_{n+1}$ mapped
  to $s_2$ under $\alpha$, such that
  \[
    u_1 \ksieq{i} v_1 \ksieq{i} u_2 \ksieq{i} v_2 \ksieq{i} \cdots \ksieq{i} u_{n+1}
    \ksieq{i} v_{n+1}
  \]
  By choice of $n$ and by the pigeonhole principle, we get $j < j'$
  such that $u_j \kbceq{i} u_{j'}$. Hence,
  \[
    u_j \ksieq{i} v_j \ksieq{i} u_{j'} \ksieq{i} u_j
  \]
  It follows that $u_j \kbceq{i} v_j$ and it suffices to set $w_1 = u_j$ and
  $w_2 = v_j$ to terminate the proof.
\end{proof}

\section{\texorpdfstring{Computing \dChains}{Computing Σ\texttwoinferior-Chains}}
\label{sec:comput}
In this section, we present an algorithm which, given a morphism and a integer
$n \geqslant 1$ as input, computes all associated \dchains of length $n$. We
already know by Theorems~\ref{thm:sep} and~\ref{thm:caracsig} that achieving
this for $n=2$ yields an algorithm deciding the separation problem for \siwd
and \piwd and algorithms deciding the membership problem for \siwt, \piwt and
\dewt. In Section~\ref{sec:caracbc}, we will obtain as well an algorithm deciding
the membership problem for \bswd.

Note that our algorithm is designed to work with \emph{alphabet compatible}
morphisms only. As shown in the next lemma, this not restrictive: the
problem of computing the \dchains associated to any morphism can
always be reduced to this case.

\begin{lemma} \label{lem:extchains}
  Set $i \geqslant 1$ and $n \geqslant 1$. Given $\alpha: A^* \to M$ a morphism into a finite
  monoid $M$ and $\beta: A^* \to M \times 2^A$ its alphabet completion,
  we have the following property:
  \[
    \Csin[\alpha] = \big\{(s_1,\dots,s_n) \mid \exists B_1,\dots,B_n \in 2^A
    \text{ s.t. } ((s_1,B_1),\dots,(s_n,B_n)) \in \Csin[\beta]\big\}.
  \]
\end{lemma}

\begin{proof}
  It is immediate from the definitions that for all $k>0$, we have
  \[
    \Csikn[\alpha] = \big\{(s_1,\dots,s_n) \mid \exists B_1,\dots,B_n \in 2^A
    \text{ s.t. } ((s_1,B_1),\dots,(s_n,B_n)) \in \Csikn[\beta]\big\}.
  \]
  Now from Lemma~\ref{fct:chainref}, there exists some $k \in \nat$ such
  that $\Csikn[\alpha] = \Csin[\alpha]$ and $\Csikn[\beta] = \Csin[\beta]$.
\end{proof}

We can now present the algorithm. We organize the section into three
parts. In the first one, we describe the separation algorithm itself. The two
remaining parts are devoted to the proofs of its soundness and
completeness.

\subsection{\texorpdfstring{An algorithm that computes \dchains}{An algorithm that computes Σ\texttwoinferior-chains}}

For the remainder of this section, we fix an \emph{alphabet
  compatible} morphism $\alpha: A^* \rightarrow M$ into a finite
monoid $M$. Recall that this means that for any $s \in M$, \content{s}
is well-defined as $\content{w}$, for any $w\in\alpha^{-1}(s)$.
For any fixed $n \geqslant 1$, we explain how to compute
the following two sets:
\begin{enumerate}
\item the set $\Cstwon[\alpha]$ of \dchains of length $n$ for
  $\alpha$.
\item the set $\fCtwon[\alpha]$ of \djuns of length $n$ for
  $\alpha$.
\end{enumerate}
In fact, our algorithm directly computes the second item,
\emph{i.e.}, $\fCtwon[\alpha]$. Recall that by
Fact~\ref{fct:recoverchains}, this is enough to obtain the first item
as well. Note that considering \djuns is necessary for the technique
to work, even if we are only interested in computing \dchains.

\medskip
\noindent {\bf Outline.} We begin by explaining what our algorithm
does. For this outline, assume $n = 2$. Observe that for all $w \in
A^*$, we have $(\alpha(w),\bigl\{\alpha(w)\bigr\}) \in
\fC_{2,2}[\alpha]$. The algorithm starts from the set containing only these trivial \djuns, and
then saturates this set with two operations, which both preserve membership in
$\fC_{2,2}[\alpha]$. Let us describe these two operations.
\begin{itemize}
\item The first one is multiplication: by Fact~\ref{fct:chaincomp2},
  $\fC_{2,2}[\alpha]$ is a submonoid of $M \times 2^M$.
\item The second operation exploits the following specific property of~\siwd,
  which is a consequence of of the \sici-property Lemma (Lemma~\ref{lem:siprop}): for all words $u,v,w,w'$,
  we have
  \begin{equation}
    \label{eq:djun}
    \forall k\ \exists \ell\qquad \text{\ \ \ } \bigl[w \ksieq{2} u, \ w \ksieq{2} v \text{ and
    } \content{w'} = \content{w}\bigr] \ \ \ \Longrightarrow\ \ \  w^{2\ell} \ksieq{2}
    u^\ell w' v^{\ell}.
  \end{equation}
  This is why \djuns are needed: in order to use this property, we need to
  have a \emph{single word~$w$} such that $w \ksieq{2} u$ and $w \ksieq{2} v$, and
  this information is not provided by \dchains alone. Once abstracted at the
  monoid level, Equation~\eqref{eq:djun} yields an operation that states that
  whenever $(s,\Ss)$ belongs to $\fC_{2,2}[\alpha]$, then so does
  $(s,\Ss)^\omega \cdot (1_M,\Ts) \cdot (s,\Ss)^\omega$, where $\Ts$ is the
  set $\{t \mid \content{t} = \content{s}\}$.  Note that this is also where we
  need $\alpha$ to be alphabet compatible.
\end{itemize}
Let us now formalize this procedure and generalize it to arbitrary length.

\medskip
\noindent {\bf Algorithm.} As we explained, our algorithm is a \emph{least
  fixpoint}. We start from a set of trivial \djuns and saturate this set with
two operations until stabilization. Denote by $n \geqslant 1$ the
common length of all \chains in \juns we want to compute. We initialize our
fixpoint algorithm with $\fI_n \subseteq M \times 2^{M^{n-1}}$ defined by
$$\fI_n = \bigl\{(\alpha(w),\{(\alpha(w),\dots,\alpha(w))\}) \mid w \in A^*\bigr\}.$$

\smallskip

We now describe our fixpoint operation. To any set of \juns $\fR
\subseteq M \times 2^{M^{n-1}}$, we associate another subset $\Sat_n(\fR)$ of $M \times
2^{M^{n-1}}$ such that
\begin{equation*}
  \Sat_n(\fR) \supseteq \fR,
\end{equation*}
defined as a lowest fixpoint (with respect to inclusion).  We will then prove
that for all $n$, one can extract from $\Sat_n(\fI_n)$ the set
$\fCtwon[\alpha]$ (and therefore also $\Cstwolen{n}[\alpha]$ by
Fact~\ref{fct:recoverchains}).

\smallskip
For length $n=1$, we simply set $\Sat_1$ as the identity, \emph{i.e.}, $\Sat_1(\fR) =
\fR$. This is because, by definition, $\fCtwolen1[\alpha] = \fI_1$. We
now define $\Sat_n$ for a length $n \geqslant 2$.  
For $\fR \subseteq M \times 2^{M^{n-1}}$, we define $\Sat_n(\fR)$ as the smallest
subset of $M\times 2^{M^{n-1}}$ containing $\fR$ and satisfying the three
following closure properties:
\begin{enumerate}[label=$(\mathit{{Op}}_{\arabic*})$,ref=$\mathit{{Op}}_{\arabic*}$]
\item \label{eq:dwn} $\downclos \Sat_n(\fR) \subseteq \Sat_n(\fR)$.
\item \label{eq:mul} $\Sat_n(\fR) \cdot \Sat_n(\fR) \subseteq \Sat_n(\fR)$.
\item \label{eq:oper} For all $(s,\Ss) \in \Sat_n(\fR)$, if $\Ts =
  \big\{(t_1,\dots,t_{n-1}) \in \Cs_{2,n-1}[\alpha] \mid \content{t_1} =
  \content{s}\big\}$, then
  \[
    (s,\Ss)^\omega \cdot (1_M,\Ts) \cdot (s,\Ss)^{\omega} \in \Sat_n(\fR).
  \]
\end{enumerate}
It is straightforward that $\Sat_n(\fR)$
can be effectively computed from~$\fR$
and $\Cs_{2,n-1}[\alpha]$
using a smallest fixpoint algorithm. Note however that the definition of
$\Sat_n$
is parametrized by the set $\Cstwolen{n-1}[\alpha]$,
\emph{i.e.}, the set of \dchains of length $n-1$.
This means that in order to compute $\Sat_n$,
we need to have previously computed the \dchains of length $n-1$.
This set can be computed by the same algorithm at stage $n-1$:
indeed, from its output $\Sat_{n-1}(\fI_{n-1})$,
one can compute the set of \djuns of length $n-1$,
and then by Fact~\ref{fct:recoverchains}, the set of all \dchains of length
$n-1$.

\smallskip
This finishes the definition of the algorithm. Its soundness and
completeness are stated in the following proposition.

\begin{proposition} \label{prop:compu}
  Given $n \gmo 1$ and $\ell_{2,n} = 9n|M|^2 \cdot 2^{|M|^{n-1}}$, we
  have
  \begin{equation}
    \label{eq:corr-compl}
    \fCtwon[\alpha] = \fCgen{2}{\ell_{2,n}}{n}[\alpha] = \Sat_n(\fI_n).
  \end{equation}
\end{proposition}
\noindent
Proposition~\ref{prop:compu} establishes both soundness and completeness of
the algorithm:
\begin{itemize}
\item the inclusion $\Sat_n(\fI_n)\subseteq
  \fCtwon[\alpha]$ gives its soundness: $\Sat_n(\fI_n)$
  \emph{only} consists of \djuns of length $n$,
\item the containment $\Sat_n(\fI_n)\supseteq \fCtwon[\alpha]$ gives its
  completeness: the set $\Sat_n(\fI_n)$ contains \emph{all} \djuns of
  length~$n$.
\end{itemize}
It also establishes a bound $\ell_{2,n}$ on a sufficient quantifier rank, whose
existence was already known from Lemma~\ref{fct:chainref2}. This bound is a
byproduct of the proof of the algorithm. It is of particular interest for
separation and Theorem~\ref{thm:sep}. Indeed, one can prove that for any two
languages that are \siwd-separable and recognized by $\alpha$, the separator
can be chosen with quantifier rank $\ell_{2,2}$ (refer to
Remark~\ref{rq:rank-sep}). From Theorem~\ref{thm:sep}, we also get
decidability of the separation problem for \siwd, as stated in the following
corollary.

\begin{corollary} \label{cor:decidsep}
  Given as input two regular languages $L_1,L_2$ it is decidable to test
  whether $L_1$ can be $\siwd$-separated (resp. $\piwd$-separated) from
  $L_2$.
\end{corollary}

Similarly, we get decidability of the membership problem for \siwt,
\piwt and \dewt from Theorem~\ref{thm:caracsig}.

\begin{corollary} \label{cor:decidthree}
  Given as input a regular language $L$, the following problems are
  decidable:
  \begin{itemize}
  \item whether $L$ is definable in \siwt.
  \item whether $L$ is definable in \piwt.
  \item whether $L$ is definable in \dewt.
  \end{itemize}
\end{corollary}
Moreover, we will see in Section~\ref{sec:caracbc} that an algorithm
for the membership problem for \bswd can also be obtained by relying
on Proposition~\ref{prop:compu}.

\medskip It now remains to prove Proposition~\ref{prop:compu}, that it is the
soundness and completeness of the algorithm. We devote the rest of
Section~\ref{sec:comput} to this proof.

We proceed by induction on $n$. Observe that when $n = 1$, all
three sets $\fCtwon[\alpha]$, $\fCgen{2}{\ell_{2,n}}{n}[\alpha]$ and
$\Sat_n(\fI_n)$ are, by definition, all equal to $\fI_n$. Therefore, the
result is immediate for $n=1$.

Assume now that $n \geqslant 2$ and set $\ell_{2,n}$ and $\ell_{2,n-1}$ as defined
in Proposition~\ref{prop:compu}. Our induction hypothesis implies the
following fact.

\begin{fact} \label{fct:inducomp}
  We have $\fCtwolen{n-1}[\alpha] =
  \fCgen{2}{\ell_{2,n-1}}{n-1}[\alpha]$. In particular,
  $\Cstwolen{n-1}[\alpha] = \Csgen{2}{\ell_{2,n-1}}{n-1}[\alpha]$.
\end{fact}
We shall prove the following inclusions, which clearly
entail~\eqref{eq:corr-compl} and Proposition~\ref{prop:compu}:
\[
  \fCtwon[\alpha] \subseteq \fCgen2{\ell_{2,n}}{n}[\alpha] \subseteq
  \Sat_n(\fI_n) \subseteq \fCtwon[\alpha]
\]
That $\fCtwon[\alpha] \subseteq \fCgen2{\ell_{2,n}}{n}[\alpha]$ is immediate
from Fact~\ref{fct:chainref2}. Hence, two inclusions are left to prove:
\begin{itemize}
\item $\Sat_n(\fI_n) \subseteq \fCtwon[\alpha]$ (corresponding to
  soundness).
\item $\fCgen2{\ell_{2,n}}{n}[\alpha] \subseteq \Sat_n(\fI_n)$
  (corresponding to completeness).
\end{itemize}
\noindent
We give each proof its own subsection: soundness is shown in
Section~\ref{sec:corr-algor} and completeness in
Section~\ref{sec:compl-algor}. Note that Fact~\ref{fct:inducomp}
(\emph{i.e.}, induction on $n$) is only used for proving completeness.

\subsection{Soundness of the Algorithm}
\label{sec:corr-algor}
In this subsection, we prove that $\Sat_n(\fI_n) \subseteq \fCtwon[\alpha]$. This is a
consequence of the following proposition.

\begin{proposition}
  \label{prop:correc}
  For all $k \in \nat$, $\Sat_n(\fI_n) \subseteq \fCgen2{k}{n}[\alpha]$.
\end{proposition}

Since by definition, $\fCtwon[\alpha] = \bigcap_{k \in \nat}
\fCgen2{k}{n}[\alpha]$, it is immediate from
Proposition~\ref{prop:correc} that $\Sat_n(\fI_n) \subseteq
\fCtwon[\alpha]$ which terminates the soundness proof.

\begin{proof}[of Proposition~\ref{prop:correc}]
  Let $k \in \nat$. It is immediate from the definitions that
  $\fI_n \subseteq \fCgen2{k}{n}[\alpha]$.  Hence, by definition of $\Sat_n$, it
  suffices to prove that $\fCgen2{k}{n}[\alpha]$ is closed under
  Operations~\eqref{eq:dwn}, \eqref{eq:mul} and~\eqref{eq:oper}, \emph{i.e.}, that
  \begin{enumerate}
  \item\label{item:5}
   $\downclos \fCgen2{k}{n}[\alpha] \subseteq \fCgen2{k}{n}[\alpha]$.
  \item\label{item:6}
    $\fCgen2{k}{n}[\alpha] \cdot \fCgen2{k}{n}[\alpha] \subseteq
    \fCgen2{k}{n}[\alpha]$.
  \item\label{item:7} for all $(s,\Ss) \in \fCgen2{k}{n}[\alpha]$, if
    $\Ts = \{(t_1,\dots,t_{n-1}) \in \Cs_{2,n-1}[\alpha] \mid \content{t_1} =
    \content{s}\}$, then
    \[
      (s,\Ss)^\omega \cdot (1_M,\Ts) \cdot (s,\Ss)^{\omega} \in \fCgen2{k}{n}[\alpha].
    \]
  \end{enumerate}

  \smallskip Item~\ref{item:5} is exactly Fact~\ref{fct:subsets}:
  $\fCgen2{k}{n}[\alpha]$ is closed under subsets. That~\ref{item:6}
  holds follows from Fact~\ref{fct:chaincomp2}:
  $\fCgen2{k}{n}[\alpha]$ is a submonoid of~$M \times 2^{M^{n-1}}$.
  It remains to prove
  Item~\ref{item:7}. 
  For this, set $(s,\Ss)$ and $\Ts$ as in Item~\ref{item:7}. Let
  $B = \content{s}$. Let
  \begin{equation*}
    (r,\Rs) = (s,\Ss)^\omega \cdot (1_M,\Ts) \cdot (s,\Ss)^{\omega}.
  \end{equation*}
  We have to prove that $\Rs$ belongs to $\fCgen2{k}{n}[\alpha]$.

  Let $h = \omega \times 2^{2k}$, where $\omega=\omega(M \times 2^{M^{n-1}})$, so
  that by definition $(s,\Ss)^\omega =
  (s,\Ss)^h=  (s,\Ss)^{2h}$. Therefore:
  \begin{equation*}
    (r,\Rs) = (s,\Ss)^h \cdot (1_M,\Ts) \cdot (s,\Ss)^{h}.
  \end{equation*}
  Since $(s,\Ss)\in \fCgen2{k}{n}[\alpha]$, there exists a $k$-witness $u \in A^*$ for $(s,\Ss)$, \emph{i.e.}, such that
  \begin{equation*}
    \alpha(u) = s,
  \end{equation*}
  and for every \dchain $(s_2,\dots,s_n) \in \Ss$ there exist
  $u_2,\dots,u_n \in A^*$ satisfying
  \begin{equation}
    \label{eq:1}
    \begin{cases}
      u \ksieq{2} u_2 \ksieq{2} \cdots \ksieq{2} u_n,\\
      \forall j,\ \alpha(u_j)=s_j.
    \end{cases}
  \end{equation}
  Observe that $\alpha(u) = s$ implies that $\content{u} = \content{s} = B$.
  Let $$w = u^{2h},$$ so that $\content{w} = \content{u} = B$ and
  $\alpha(w) = \alpha(u)^{2h} = s^{2h} = r$. We prove that
  $(r,\Rs) \in \fCgen2{k}{n}[\alpha]$ with $w$ as $k$-witness. It suffices to show that for any \chain
  $(r_2,\dots,r_n) \in \Rs$, there exist $w_2,\dots,w_n \in A^*$ satisfying
  $w \ksieq{2} w_2 \ksieq{2} \cdots \ksieq{2} w_n$ and such that $\alpha(w_j)=r_j$ for all $j$.

  \smallskip Let $(r_2,\dots,r_n) \in \Rs$. By definition of $\Rs$, we have
  $(r_2,\dots,r_n) = (s'_2t_2s''_2,\dots,s'_nt_ns''_n)$ with $(s'_2,\dots,s'_n),$
  $(s''_2,\dots,s''_n) \in \Ss^{h}$ and $(t_2,\dots,t_n) \in \Ts$. Since
  $(s'_2,\dots,s'_n) \in
  \Ss^{h}$, 
  using $h$ times~\eqref{eq:1} and the fact that $\ksieq{2}$ is a
  pre-congruence (Lemma~\ref{lem:efconcat}), we obtain words
  $u'_2,\dots,u'_n \in A^*$ such that
  \begin{equation}
    \label{eq:2}
    \begin{cases}
      u^{h} \ksieq{2} u'_2 \ksieq{2} \cdots \ksieq{2} u'_n\\
      \forall j,\ \alpha(u'_j)=s'_j.
    \end{cases}
  \end{equation}
  Similarly, since $(s''_2,\dots,s''_n) \in
  \Ss^{h}$ we get $u''_2,\dots,u''_n\in A^*$ such that
  \begin{equation}
    \label{eq:3}
    \begin{cases}
      u^{h} \ksieq{2} u''_2 \ksieq{2} \cdots \ksieq{2} u''_n\\
      \forall j,\ \alpha(u''_j)=s''_j.
    \end{cases}
  \end{equation}
  On the other hand, since $(t_2,\dots,t_n) \in \Ts$, we obtain that
  $\content{t_2} = B$ and $(t_2,\dots,t_n) \in \Cstwolen{n-1}[\alpha]$. Hence, we
  get words $v_2,\dots,v_n \in A^*$, such that
  \begin{equation}
    \label{eq:4}
    \begin{cases}
      v_2 \ksieq{2} \cdots \ksieq{2} v_n\\
      \forall j\geqslant 2,\ \alpha(v_j)=t_j.
    \end{cases}
  \end{equation}
  Observe that this implies in particular that $\content{v_2} = B$. For all
  $j \geqslant 2$, set
  \begin{equation*}
    w_j = u'_jv^{}_ju''_j.
  \end{equation*}
  Note that for all $j \geqslant 2$, $\alpha(w_j) = s'_jt_js''_j = r_j$. It remains to prove
  that $w \ksieq{2} w_2 \ksieq{2} \cdots \ksieq{2} w_n$ to terminate the
  proof. That $w_2 \ksieq{2} \cdots \ksieq{2} w_n$ is immediate by~\eqref{eq:2},
  \eqref{eq:3} and \eqref{eq:4}, since $\ksieq2$
  is a pre-congruence (by Lemma~\ref{lem:efconcat} again). Since $w = u^{2h}$, the
  remaining inequality to prove~is
  \begin{equation}
    \label{eq:5}
    u^{h}u^{h} \ksieq{2} u'_2v^{}_2u''_2.
  \end{equation}
  Since $\ksieq2$ is a pre-congruence by Lemma~\ref{lem:efconcat}, we know by \eqref{eq:2},
  \eqref{eq:3} and \eqref{eq:4} that
  $u^{h} v_2 u^{h} \ksieq{2} u'_2v^{}_2u''_2$. Therefore to
  establish~\eqref{eq:5}, it suffices to prove that
  \begin{equation}
    \label{eq:6}
    u^{h}u^{h} \ksieq{2} u^{h}v_2u^{h}.
  \end{equation}
  Recall that by definition $\content{v_2} = \content{u} = B$. Therefore, it is
  straightforward that
  \begin{equation}
    \label{eq:inegcorrec}
    v_2 \sieq{k}{1} u^{2^k}.
  \end{equation}
  Now,~\eqref{eq:6} follows from Lemma~\ref{lem:siprop}, in view of the choice of
  $h = \omega \times 2^{2k}$ and of~\eqref{eq:inegcorrec}.
\end{proof}
 
\subsection{Completeness of the Algorithm}
\label{sec:compl-algor}
We prove that for any $\ell \geqslant \ell_{2,n} = 9n|M|^2 \cdot
2^{|M|^{n-1}}$, we have $\fCgen2{\ell}{n}[\alpha]
\subseteq \Sat_n(\fI_n)$.  We
denote by $k_n$ the size of the set of \juns of length $n$, \emph{i.e.}, $$k_n
= |M \times 2^{M^{n-1}}|.$$ In particular, this means that $\ell_{2,n}
= 9n|M|k_n$. The proof is by induction and relies on Simon's
Factorization Forests Theorem. To state the induction, we need more
terminology.

\medskip
\noindent
{\bf Generated \Juns.} Set $k \in \nat$, $w \in A^*$. We set $\fg_n^k(w) \in M \times 2^{M^{n-1}}$ as the  maximal
\jun  of $\fCgen{2}{k}{n}[\alpha]$ that has $w$ as a $k$-witness. Formally,
setting $\fg_n^k(w)=(\alpha(w),\Gs)$, we have $(t_2,\dots,t_{n}) \in \Gs$ if and only if
there exist $w_2,\dots,w_{n} \in A^*$ satisfying
\begin{itemize}
\item for all $j$, $\alpha(w_j)=t_j$.
\item $w \ksieq{2} w_2 \ksieq{2} \cdots \ksieq{2} w_{n}$.
\end{itemize}
By definition, any $\fg_n^k(w)$ is a \dkjun of length $n$: $\fg_n^k(w)
\in \fCgen{2}{k}{n}[\alpha]$. Moreover, by definition we have
\[
  \fCgen{2}{k}{n}[\alpha] = \downclos \big\{\fg_n^k(w) \mid w \in A^*\big\}.
\]
We illustrate this definition with two lemmas that will be useful in the
proof. The first one states that $\fg_n^k(w)$ gets smaller as $k$
gets larger.
\begin{lemma} \label{lem:genref}
  Let $w \in A^*$, $n \in \nat$ and $k < \ell$. We have,
  \[
    \fg_n^{\ell}(w) \subseteq \fg_n^{k}(w).
  \]
\end{lemma}
\begin{proof}
  Immediate from the fact that if $k < \ell$, then for all $u,v$, $u \sieq{\ell}{2} v
  \Rightarrow u \ksieq{2} v$.
\end{proof}

\noindent
Our second lemma is a decomposition result that we will use several times.

\begin{lemma}[Decomposition Lemma] \label{lem:gendecomp}
  Let $w,w' \in A^*$ and $k \geqslant 1$. Then
  \[
    \fg_n^k(ww') \subseteq \fg_n^{k-1}(w) \cdot
    \fg_n^{k-1}(w').
  \]
\end{lemma}

\begin{proof}
  Set $(r,\Rs) = \fg_n^k(ww')$, $(s,\Ss) = \fg_n^{k-1}(w)$ and
  $(t,\Ts) = \fg_n^{k-1}(w')$. By definition, $r=\alpha(ww')$,
  $s=\alpha(w)$ and $t =\alpha(w')$, hence $r=st$. It remains to
  prove that $\Rs \subseteq \Ss \cdot \Ts$. Let $(r_2,\ldots,r_n) \in \Rs$. By
  definition, there exist $u_2,\dots,u_n$ such that for all $j$,
  $\alpha(u_j) = r_j$, and
  \begin{equation}
    w\cdot w' \sieq{k}{2} u_2 \sieq{k}{2} \cdots \sieq{k}{2} u_n.\label{eq:cut}
  \end{equation}
  Using $(n-1)$ times a simple \efgame argument, one for each $\sieq{k}{2}$
  relation in \eqref{eq:cut}, we obtain that all words $u_j$ can be
  decomposed as $u_j = v^{}_{j} \cdot v'_{j}$ such that
  \[
    \begin{array}{lllll}
      w & \sieq{k-1}{2} & v_{2}  & \sieq{k-1}{2} \cdots \sieq{k-1}{2} & v_{n} \\[1.2ex]
      w' & \sieq{k-1}{2} & v'_{2} & \sieq{k-1}{2} \cdots \sieq{k-1}{2} & v'_{n}
    \end{array}
  \]
  For instance, $v^{}_2$ and $v'_2$ are obtained by playing the $k$-round \efgame
  game over $w\cdot w'$ and $u_2$, where the first move of Spoiler is to play in
  $ww'$ on the first letter of $w'$. The answer of Duplicator in $u_2$
  splits this word into two factors, $v^{}_2$ and $v'_2$.

  Now for all $j$, set $s_{j} = \alpha(v_{j})$ and  $t_{j} =
  \alpha(v'_{j})$. By definition, we have $(s_2,\dots,s_n) \in \Ss$ and
  $(t_2,\dots,t_n) \in  \Ts$. Moreover, by definition
  \[
    (r_2,\dots,r_n) = (s_2t_2,\dots,s_nt_n).
  \]
  It follows that $(r_2,\dots,r_n) \in \Ss \cdot \Ts$ which terminates
  the proof.
\end{proof}

We can now prove that
$\fCgen2{\ell}{n}[\alpha] \subseteq \Sat_n(\fI_n)$ when
$\ell \geqslant \ell_{2,n}$ (recall that $\ell_{2,n} = 9n|M|k_n$, where
$k_n=|M\times 2^{M^{n-1}}|$). This is a consequence of the next proposition.

\begin{proposition} \label{prop:comp}
  Set $\ell \geqslant  \ell_{2,n}$ and $w \in A^*$, then $\fg_n^\ell(w) \in
  \Sat_n(\fI_n)$.
\end{proposition}

It is immediate from Proposition~\ref{prop:comp} that for any
$\ell \geqslant \ell_{2,n}$,
$\big\{\fg_n^\ell(w) \mid w \in A^*\big\} \subseteq \Sat_n(\fI_n)$. Since we
know that
$\fCgen{2}{\ell}{n}[\alpha] = \downclos \big\{\fg_n^\ell(w) \mid w \in A^*\big\}$, we
obtain that $\fCgen2{\ell}{n}[\alpha] \subseteq \downclos
\Sat_n(\fI_n) = \Sat_n(\fI_n)$ which
yields completeness.

\smallskip
It remains therefore to prove Proposition~\ref{prop:comp}. The proof is once
again by induction on the height of the $\alpha$-factorization forest of $w$.
We state the induction in the following proposition. Recall again that
$k_n = |M \times 2^{M^{n-1}}|$ is the size of the set of \juns of length $n$.

\begin{proposition} \label{prop:compi}
  Let $h \geqslant 1$ and let $k \geqslant h \cdot 3k_n + \ell_{2,n-1}$. Then
  for any $w \in A^*$ that admits an $\alpha$-factorization forest of
  height at most $h$, we have $\fg_n^k(w) \in \Sat_n(\fI_n)$.
\end{proposition}

\noindent
Proposition~\ref{prop:comp} is indeed a consequence of
Proposition~\ref{prop:compi}. This is because
\begin{itemize}
\item any $w \in A^*$ admits an $\alpha$-factorization forest of height at most
  $3|M|-1$, by Theorem~\ref{thm:facto},
\item one can verify that $\ell_{2,n} \geqslant (3|M|-1) \cdot 3k_n +
  \ell_{2,n-1}$.
\end{itemize}
We now prove
Proposition~\ref{prop:compi}. Note that
this is where we use Fact~\ref{fct:inducomp}, \emph{i.e.}, induction on~$n$.

\smallskip

As in the statement of Proposition~\ref{prop:compi}, take $h \geqslant 1$,
$k \geqslant h \cdot 3k_n + \ell_{2,n-1}$ and let $w \in A^*$ admitting an
$\alpha$-factorization forest of height at most $h$. We need to prove that
$\fg_n^k(w) \in  \Sat_n(\fI_n)$. The proof is by induction on $h$.

\smallskip
If $h = 1$, then $w$ admits an $\alpha$-factorization forest which is a
leaf. In that case, $w$ is a single letter word $a \in A$ or the empty word
$\varepsilon$. Observe that $k \geqslant 2$. Therefore, one can check that the
language $\{w\}$ is definable in $\siw2$, hence
$\fg_n^k(w) = \big(\alpha(w),\big\{(\alpha(w),\dots,\alpha(w))\big\}\big)$. It follows that
$\fg_n^k(w) \in \fI_n \subseteq \Sat_n(\fI_n)$, which finishes the proof for
this case.

\smallskip
Assume now that $h > 1$. If the $\alpha$-factorization forest of $w$
is again a leaf, we conclude as above. Otherwise, we apply induction
to the factors given by this factorization forest. In particular, we
will use Lemma~\ref{lem:gendecomp} (the Decomposition Lemma) to decompose $\fg_n^k(w)$ according
to this factorization forest. Then, once the factors have been treated
by induction, we will use the operations in the definition of $\Sat_n$
to lift the result to the whole word $w$. We distinguish two cases
depending on the nature of the topmost node in the
$\alpha$-factorization forest of $w$.

\medskip
\noindent
{\bf Case 1: the topmost node is a binary node.} We use induction on
$h$ and Operation~\eqref{eq:mul} in the definition of $\Sat_n$. By
hypothesis $w = w_1 \cdot w_2$ with $w_1,w_2$ words admitting
$\alpha$-factorization forests of respective heights $h_1,h_2 \leqslant h-1$. Observe
that
\[
  k - 1 \geqslant (h-1) \cdot 3k_n + \ell_{2,n-1}
\]
Therefore, we can apply our induction hypothesis to $w_1,w_2$ and we
obtain that $\fg_n^{k-1}(w_1) \in \Sat_n(\fI_n)$ and $\fg_n^{k-1}(w_2) \in
\Sat_n(\fI_n)$. By Operation~\eqref{eq:mul} in the definition of
$\Sat_n$, it is immediate 
that $\fg_n^{k-1}(w_1) \cdot \fg_n^{k-1}(w_2) \in \Sat_n(\fI_n)$. Moreover, by
Lemma~\ref{lem:gendecomp} (the Decomposition Lemma), $\fg_n^k(w) \subseteq \fg_n^{k-1}(w_1)
\cdot \fg_n^{k-1}(w_2)$. It follows from Operation~\eqref{eq:dwn} that
$\fg_n^k(w) \in \downclos \Sat_n(\fI_n)$, which concludes this case.

\medskip\noindent
{\bf Case 2: the topmost node is an idempotent node.} This is the
most involved case. We use induction on $h$ and the two operations
in the definition of $\Sat_n$. Note that this is also where
Fact~\ref{fct:inducomp} (\emph{i.e.}, induction on $n$ in the general proof of
Proposition~\ref{prop:compu}) is used. We set $$e = \alpha(w).$$ Observe that by
hypothesis of this case, $e$ is an idempotent. Let $$B =
\content{w} = \content{e}.$$ Also set
\begin{equation}
  \label{eq:T}
  \Ts = \big\{(t_1,\dots,t_{n-1}) \in \Cs_{2,n-1}[\alpha] \mid \content{t_1} = B\big\}.
\end{equation}
Since one can test the alphabet of a word in $\siw2$, all
elements of any \dchain of $\Ts$ actually have alphabet~$B$.

\smallskip
We begin by summarizing our hypothesis: $w$ admits what we call an
$(e,p)$-decomposition.

\medskip
\noindent
{\bf $(e,p)$-Decompositions.} For the rest of the section, we set $p = (h-1) \cdot 3k_n +
\ell_{2,n-1}$. Let $u \in A^*$. We say that $u$ admits an
\emph{$(e,p)$-decomposition} $u_1,\dots,u_m$ if

\begin{enumerate}[label=$\alph*)$,ref=\alph*]
\item\label{item:8} $u = u_1 \cdots u_m$,
\item\label{item:9} for all $j$, $\alpha(u_j) = e$ and
\item\label{item:10} for all $j$, $\fg_n^{p}(u_j) \in
  \Sat_n(\fI_n)$.
\end{enumerate}
Note that $\ref{item:9})$ means that $\alpha(u_j)$ is a constant idempotent. In
particular, since $\alpha$ is alphabet compatible, this also implies that all
factors $u_i$ have the same alphabet as $e$, namely $B$. Using
Fact~\ref{fct:inducomp} (\emph{i.e.}, induction on $n$ in the general proof of
Proposition~\ref{prop:compu}) we obtain the following fact.

\begin{fact} \label{fct:inducnb}
  For any $(e,p)$-decomposition $u_1,\dots,u_m$ of a
  word and for all $i \leqslant j$, we have
  $\fg_n^{p}(u_i \cdots u_j) \subseteq (e,\Ts)$, where $\Ts$ is defined
  by~\eqref{eq:T}.
\end{fact}

\begin{proof}
  By definition, the ``root'' of the \jun $\fg_n^{p}(u_i \cdots u_j)$ is labeled by $\alpha(u_i \cdots u_j)=e$.
  Therefore, we may set $(e,\Ts') = \fg_n^{p}(u_i \cdots u_j)$. Since $p \geqslant
  \ell_{2,n-1}$, that $\Ts' \subseteq \Cs_{2,n-1}[\alpha]$ follows from
  Lemma~\ref{lem:genref} and Fact~\ref{fct:inducomp}. Set
  $(t_1,\dots,t_{n-1}) \in \Ts'$, we have to prove that $\content{t_1} =
  B$. This is because $t_1 = \alpha(v)$ for some word $v$ satisfying
  $u_i \cdots u_j \sieq{p}{2} v$. Since $p \geqslant 2$, it follows that
  $\content{v} = \content{u_i \cdots u_j} = B$, which terminates the
  proof.
\end{proof}

We now use the hypothesis of Case~2 to conclude that $w$ admits an
$(e,p)$-decomposition.

\begin{fact} \label{fct:ebdecomp}
  The word $w$ admits an $(e,p)$-decomposition.
\end{fact}

\begin{proof}
  By hypothesis of Case~2, there exists a decomposition $w_1,\dots,w_m$
  of $w$ that satisfies points $a)$ and $b)$. Moreover, for all $j$,
  $w_j$ admits an $\alpha$-factorization forest of height $h_j \leqslant
  h-1$. Therefore point $c)$ is obtained by induction hypothesis on
  $h$.
\end{proof}

Recall that we want to prove that $\fg_n^k(w) \in \Sat_n(\fI_n)$. In general, the number of factors $m$ in the
$(e,p)$-decomposition of $w$ can be arbitrarily large. In
particular, it is possible that $k - (m-1) < p$. This
means that we cannot simply use Lemma~\ref{lem:gendecomp} as we did in
the previous case to conclude that $\fg_n^k(w) \subseteq
\fg_n^{p}(w_{1}) \cdots
\fg_n^{p}(w_{m})$. However, we will partition $w_1,\dots,w_m$
as a bounded number of subdecompositions that we can treat using the
second operation in the definition of $\Sat_n$. The partition is given
by induction on a parameter of the $(e,p)$-decomposition
$w_1,\dots,w_m$, which we define now.

\medskip
\noindent
{\bf Index of an $(e,p)$-decomposition.} Recall that $k_n = |M \times
2^{M^{n-1}}|$ and let $u \in A^*$ that admits an
$(e,p)$-decomposition $u_1,\dots,u_{m}$. Let $(f,\Fs) \in M \times
2^{M^{n-1}}$ be an idempotent and $j \leqslant m$, we say that $(f,\Fs)$
can be \emph{inserted} at position $j$ is there exists $i \leqslant
(k_n-1)$ such that
\[
  \fg_n^{p}(u_{j-i}) \cdots \fg_n^{p}(u_{j}) \cdot
  (f,\Fs) = \fg_n^{p}(u_{j-i}) \cdots \fg_n^{p}(u_{j}).
\]
The \emph{index} of the $(e,p)$-decomposition $u_1,\dots,u_{m}$ is the number
of distinct idempotents $(f,\Fs) \in M \times 2^{|M|^{n-1}}$ that can
be inserted at some position $j \leqslant m$. Observe that by definition,
the index of any $(e,p)$-decomposition is bounded by $k_n$.

\begin{lemma} \label{lem:inducase3}
  Let $u \in A^*$ admitting an $(e,p)$-decomposition $u_1,\dots,u_m$
  of index $g$ and set $\widehat{k} \geqslant g + 2k_n + p$. Then
  $\fg_n^{\widehat{k}}(u) \in \Sat_n(\fI_n)$.
\end{lemma}

Before proving this lemma, we use it to conclude Case~2. We know that
$w$ admits an $(e,p)$-decomposition of index $g \leqslant k_n$. By definition,
$k \geqslant 3k_n + p$, hence, it is immediate from
Lemma~\ref{lem:inducase3} that $\fg_n^k(w) \in 
\Sat_n(\fI_n)$. It now remains to prove Lemma~\ref{lem:inducase3}.

\begin{proof}[of Lemma~\ref{lem:inducase3}]
  The proof goes by induction on the index $g$. When $m \leqslant k_n$,
  the result can be obtained from Lemma~\ref{lem:gendecomp} by using $(m-1)$ times
  the argument we used in Case~1. Assume now that $m >
  k_n$, we rely on the following fact:

  \begin{fact} \label{fct:insert}
    There exists a position $j\leqslant k_n$ and an idempotent $(e,\Es) \in \Sat_n(\fI_n)$ that
    can be inserted at position $j$.
  \end{fact}

  \begin{proof}
    Since all $\fg_n^p(w_i)$ belong to the monoid $M \times 2^{M^{n-1}}$ whose size
    is $k_n$, it follows from the pigeon-hole principle that there exist $j <
    j' \leqslant k_n+1$ such that:
    \[
      \fg_n^{p}(w_{1}) \cdots \fg_n^{p}(w_{j}) = \fg_n^{p}(w_{1}) \cdots \fg_n^{p}(w_{j'}).
    \]
    Hence it suffices to take $(e,\Es) = (\fg_n^{p}(w_{j+1})
    \cdots \fg_n^{p}(w_{j'}))^\omega$. Note that $(e,\Es) \in
    \Sat_n(\fI_n)$ because of Item~$\ref{item:10})$ in the
    definition of $(e,p)$-decompositions and Operation~\eqref{eq:mul} in
    the definition of $\Sat_n$.
  \end{proof}
  Denote by $j\leqslant k_n$ a position given by Fact~\ref{fct:insert}, and set
  $\ell \leqslant m$ as the largest integer such that $(e,\Es)$ can be
  inserted at position $\ell$. In particular, $j\leqslant\ell$. Using
  Lemma~\ref{lem:gendecomp}, we get that
  \[
    \fg_n^{\widehat{k}}(u) \subseteq \fg_n^{\widehat{k}-1}(u_1\cdots u_{\ell}) \cdot
    \fg_n^{\widehat{k}-1}(u_{\ell+1} \cdots u_m).
  \]
  By definition, $u_{\ell+1},\dots,u_m$ is an $(e,p)$-decomposition and it has index
  strictly smaller than that of $u_1, \dots,u_m$ (by definition of $\ell$,
  there is no position between $\ell+1$ and $m$ at which $(e,\Es)$ can be
  inserted). Hence, it is immediate by induction hypothesis that
  \[
    \fg_n^{\widehat{k}-1}(u_{\ell+1} \cdots u_m) \in \Sat_n(\fI_n).
  \]
  It now remains to prove that $\fg_n^{\widehat{k}-1}(u_1\cdots u_{\ell})
  \in  \Sat_n(\fI_n)$. The result will then follow from
  Operations~\eqref{eq:dwn} and~\eqref{eq:mul} in the definition of $\Sat_n$. We distinguish
  two cases depending on the distance between $j$ and $\ell$.

  \medskip
  \noindent
  {\it Case a)} Assume first that $\ell \leqslant j + k_n$. In that case, since
  $j \leqslant k_n$, we have $\ell \leqslant 2k_n$. The result can then be obtained
  from Lemma~\ref{lem:gendecomp} by using $\ell-1$ times the same argument as the one
  we used in Case~1.

  \medskip
  \noindent
  {\it Case b)} It remains to treat the case when $\ell > j + k_n$. This is where
  Operation~\eqref{eq:oper} in the definition of $\Sat_n$ is used.
  Consider the following,
  \[
    \begin{array}{lcl}
      (e,\Rs) & = & \fg_n^{p}(u_1) \cdots
                    \fg_n^{{p}}(u_{j}) \\
      (e,\Ts') & = & \fg_n^{p}(u_{j+1} \cdots u_{j-k_n}) \cdot \fg_n^{p}(u_{\ell-(k_n-1)}) \cdots
                     \fg_n^{{p}}(u_{\ell})
    \end{array}
  \]
  Note that we know from Item~$\ref{item:10})$ in the definition of
  $(e,p)$-decompositions and Operation~\eqref{eq:mul} in the definition of
  $\Sat_n$ that $(e,\Rs) \in \Sat_n(\fI_n)$. Moreover, using
  Fact~\ref{fct:inducnb} we obtain that $(e,\Ts') \subseteq (e,\Ts)$.

  Observe that $(\widehat{k}-1)-(j+k_n-1) \geqslant p$. Hence using $j+k_n-1$ times the
  Decomposition Lemma (Lemma~\ref{lem:gendecomp}) and Lemma~\ref{lem:genref},
  we obtain:
  \[
    \fg_n^{\widehat{k}-1}(u_1\cdots u_{\ell}) \subseteq (e,\Rs) \cdot (e,\Ts')
  \]
  By definition, $(e,\Es) \in \Sat_n(\fI_n)$ can be inserted
  at both positions $j$ and $\ell$, hence we have:
  \[
    \fg_n^{\widehat{k}-1}(u_1\cdots u_{\ell}) \subseteq (e,\Rs) \cdot
    (e,\Es) \cdot (e,\Ts') \cdot (e,\Es)
  \]
  We now prove that $(e,\Es) \cdot (e,\Ts') \cdot (e,\Es) \in
  \Sat_n(\fI_n)$. Since we already know that $(e,\Rs) \in
  \Sat_n(\fI_n)$, it will then follow from
  Operations~\eqref{eq:dwn} and~\eqref{eq:mul} that
  $\fg_n^{\widehat{k}-1}(u_1\cdots u_{\ell}) \in \Sat_n(\fI_n)$.

  Since $(e,\Es) \in \Sat_n(\fI_n)$ and $\content{e} = B$, it
  follows from Operation~\eqref{eq:oper} in the fixpoint procedure that:
  \begin{equation*}
    (e,\Es) \cdot (e,\Ts') \cdot (e,\Es) \subseteq (e,\Es) \cdot (e,\Ts) \cdot (e,\Es) =  (e,\Es) \cdot
    (1_M,\Ts) \cdot (e,\Es) \in \Sat_n(\fI_n).
    \qed
  \end{equation*}
We conclude from Operation~\eqref{eq:dwn} that $(e,\Es) \cdot (e,\Ts')
\cdot (e,\Es) \in \Sat_n(\fI_n)$ which terminates the proof.
\end{proof}

\section{\texorpdfstring{Decidable Characterization of \bswd}{Decidable
    Characterization of BΣ\texttwoinferior(<)}}
\label{sec:caracbc}
In this section we present our decidable characterization for
\bswd. We already proved a (non-effective) characterization of \bswd
in Section~\ref{sec:generic} using the notion of \emph{alternation}.

\smallskip
Recall that a \chain $(s_1,\dots,s_n) \in M^*$ has \emph{alternation}
$\ell$ if there are exactly $\ell$ indices $i$ such that $s_i \neq
s_{i+1}$. Recall also that a set of \chains $\Ss$ has \emph{bounded
  alternation} if there exists a bound $\ell \in \nat$ such that all
\chains in $\Ss$ have alternation at most $\ell$. We know by
Corollary~\ref{cor:membc} that a regular language $L$ is definable in
\bswi if and only if $\Csi[\alpha]$ has bounded alternation with
$\alpha$ as the syntactic morphism of $L$.

\smallskip
In this section, we prove that a third equivalent (effective)
criterion can be given in the special case $i=2$. This criterion is
presented as an equation that needs to be satisfied by the alphabet completion
of the syntactic morphism of the language. This equation is parametrized by \juns of
length $2$ through a relation that we now define.

\medskip
\noindent
{\bf Alternation Schema.} Let $\alpha: A^* \rightarrow M$ be an
alphabet compatible monoid morphism. An \emph{alternation schema}
for $\alpha$ is a triple $(s,s_1,s_2) \in M^3$ such that there exist
$(r_1,\Rs_1), (r_2,\Rs_2), (e,\Es) \in \fC_{2,2}[\alpha]$ with
$(e,\Es)$ \emph{idempotent}, and such that
\begin{itemize}
\item $\content{e} = \content{s}$.
\item $s = r_1er_2$.
\item $s_1 \in \Rs_1 \cdot \Es$.
\item $s_2 \in \Es \cdot \Rs_2$.
\end{itemize}
Observe that, the set of all alternation schemas for $\alpha$ can be
computed from $\fC_{2,2}[\alpha]$.

The purpose of alternation schemas is to abstract over~$M$ a
property of words relatively to \siwd: if $(s,s_1,s_2)$ is an
alternation schema, then for all $k \in \nat$, there exist $w,w_1,w_2
\in A^*$, mapped respectively to $s,s_1,s_2$ under $\alpha$, and such
that for all $u \in \content{s}^*$, $w\ksieq{2} w_1uw_2$ (see
Lemma~\ref{lem:schemprop} below).

We now have all the terminology we need to state our decidable
characterization of \bswd.

\begin{theorem} \label{thm:caracbc}
  Let $L$ be a regular language and let $\alpha: A^* \rightarrow M$ be
  the alphabet completion of its syntactic morphism. The three following
  properties are equivalent:

  \begin{enumerate}
  \item $L$ is definable in \bswd.
  \item $\Cstwo[\alpha]$ has bounded alternation.
  \item $\alpha$ satisfies the following equation:
    \begin{equation}
      \begin{array}{c}
        (s_1t_1)^{\omega}s(t_2s_2)^{\omega} = (s_1t_1)^{\omega}s_1ts_2(t_2s_2)^{\omega} \\[.75ex]
        \text{for $(s,s_1,s_2)$ and $(t,t_1,t_2)$ alternation schemas
        such that $\content{s} = \content{t}$}.
      \end{array}\label{eq:bcs2}
    \end{equation}
  \end{enumerate}
\end{theorem}

We know from Proposition~\ref{prop:compu} that all \dchains of length
$3$ and all alternation schemas associated to $\alpha$ can be
computed. Hence, the third item of Theorem~\ref{thm:caracbc} can be
decided and we get the desired corollary.

\begin{corollary} \label{cor:decid2}
  Given as input a regular language $L$, it is decidable to test whether
  $L$ is definable in \bswd.
\end{corollary}

Note that the characterization of \bswd that we present in
Theorem~\ref{thm:caracbc} is different from the one presented
in the conference version of this paper~\cite{pzqalt}. Indeed,
the characterization of~\cite{pzqalt} was relying on three
equations, \eqref{eq:bcs2} and the following two equations, which are
parametrized by \dchains of length $3$.
\begin{equation}
  \begin{array}{rcl}
    s_1^{\omega}s_3^{\omega} & = & s_1^{\omega}s_2s_3^{\omega} \\[.75ex]
    s_3^{\omega}s_1^{\omega} & = & s_3^{\omega}s_2s_1^{\omega}
  \end{array} \quad \text{for all $(s_1,s_2,s_3) \in \Cstwo[\alpha]$}. \label{eq:bcs1}
\end{equation}

It turns out that~\eqref{eq:bcs1} is actually a consequence
of~\eqref{eq:bcs2}. We state this property in the next
lemma.

\begin{lemma} \label{lem:equiveq}
  Let $\alpha: A^* \rightarrow M$ be an alphabet compatible morphism
  into a finite monoid~$M$. If $\alpha$ satisfies~\eqref{eq:bcs2}, then
  $\alpha$ satisfies~\eqref{eq:bcs1} as well.
\end{lemma}

\begin{proof}
  Assume that $\alpha$ satisfies~\eqref{eq:bcs2} and let $(s_1,s_2,s_3)
  \in \Cstwo[\alpha]$. We have to prove that $s_1^{\omega}s_3^{\omega}
  = s_1^{\omega}s_2s_3^{\omega}$ and $s_3^{\omega}s_1^{\omega} =
  s_3^{\omega}s_2s_1^{\omega}$. We only prove the first equality, since the
  other one is obtained symmetrically. 

  We claim that $(s_1^\omega,s_1^\omega,s_3^\omega)$ and
  $(s_1^\omega s_2,s_1^\omega,s_3^\omega)$ are alternation schemas.
  Assuming this claim, note that since $(s_1,s_2,s_3)$ is a \dchain, we have
  in particular $\content{s_1} = \content{s_2} = \content{s_3}$, whence
  $\content{s_1^\omega} = \content{s_1^\omega s_2}$. It follows
  from~\eqref{eq:bcs2} that,
  \[
    (s_1^\omega s_1^\omega)^{\omega}s_1^\omega(s_3^\omega
    s_3^\omega)^{\omega} = (s_1^\omega
    s_1^\omega)^{\omega}s_1^\omega \cdot s_1^\omega
    s_2\cdot s_3^\omega (s_3^\omega
    s_3^\omega)^{\omega} 
  \]
  This exactly says that $s_1^{\omega}s_3^{\omega}=
  s_1^{\omega}s_2s_3^{\omega}$, as desired.

  Let us now prove the claim. Observe that since $(s_1,s_2,s_3) \in \Cstwo[\alpha]$, it is
  immediate from the definition of \dchains and \djuns that 
  $(s_1,\{s_1,s_3\}) \in \fC_{2,2}[\alpha]$ and $(s_2,\{s_3\}) \in 
  \fC_{2,2}[\alpha]$. We begin by proving that
  $(s_1^\omega,s_1^\omega,s_3^\omega)$  is an alternation schema.
  Let $(e,\Es) = (s_1,\{s_1,s_3\})^\omega \in \fC_{2,2}[\alpha]$ and
  $(r_1,\Rs_1) = (r_2,\Rs_2) = (1_M,\{1_M\}) \in \fC_{2,2}[\alpha]$.
  By definition, $s_1^\omega = e =r_1er_2$, $s_1^\omega
  \in \Es = \Rs_1\Es$ and $s_3^\omega \in \Es = \Es\Rs_2$. It follows
  that $(s_1^\omega,s_1^\omega,s_3^\omega)$  is an alternation
  schema.

  Finally, we prove that $(s_1^\omega
  s_2,s_1^\omega,s_3^\omega)$ is also an alternation schema. Let again
  $(e,\Es) = (s_1,\{s_1,s_3\})^\omega \in \fC_{2,2}[\alpha]$. Recall
  our algorithm for computing $\fC_{2,2}[\alpha]$ (see
  Section~\ref{sec:comput}). Since $\content{s_3} = \content{s_1} =
  \content{e}$,
  we know from Operation~\eqref{eq:oper} in the algorithm that,
  \[
    (e,\Es) \cdot (1_M,(s_3)^{\omega-1}) \cdot (e,\Es) \in \fC_{2,2}[\alpha]
  \]
  By closure under downset, it follows that
  $((s_1)^\omega,\{(s_3)^{2\omega-1}\}) \in \fC_{2,2}[\alpha]$. We
  define $(r_1,\Rs_1) = (1_M,\{1_M\}) \in \fC_{2,2}[\alpha]$ and 
  $(r_2,\Rs_2) =  ((s_1)^\omega,\{(s_3)^{2\omega-1}\}) \cdot
  (s_2,\{s_3\}) = ((s_1)^\omega s_2,\{(s_3)^{\omega}\}) \in
  \fC_{2,2}[\alpha]$. By definition, $(s_1)^\omega s_2 = r_1er_2$,
  $(s_1)^\omega \in \Rs_1\Es$ and $(s_3)^\omega \in \Es\Rs_2$. Finally,
  $\content{e}=\content{s_1^\omega s_2}$. It follows
  that $(s_1^\omega s_2,s_1^\omega,s_3^\omega)$  is an alternation
  schema.
\end{proof}

Note that we still use Equation~\eqref{eq:bcs1} in the proof of
Theorem~\ref{thm:caracbc} as it will be more convenient to use
it instead of~\eqref{eq:bcs2} in some places.

Another important remark is that there are similarities between
Theorem~\ref{thm:caracbc} and a theorem of~\cite{bpopen} that states
the decidable characterization of an entirely different formalism:
boolean combination of open sets 
of infinite trees. In~\cite{bpopen} as well, the authors present a
notion of ``\chains'' tailored to their formalism (although they do
not make the link with separation). This is not surprising as the
notion of \chain is quite generic for formalisms defined by boolean
combinations and what is specific is the algorithms computing them.

A more surprising fact is that our equations are very similar to
the ones stated in~\cite{bpopen}. Despite this fact, since the
formalisms are of different nature, the way the \chains
of~\cite{bpopen} and the way our \dchains are constructed are
completely independent. This also means that the proofs are also
mostly independent. However, we do reuse several combinatorial
arguments of~\cite{bpopen} at the end of the proof. One could say that
the proofs are both (very different) setups to apply similar
combinatorial arguments in the end.

\medskip It now remains to prove Theorem~\ref{thm:caracbc}. Observe that we
already know from Corollary~\ref{cor:membc} and Lemma~\ref{lem:extchains} that
1~$\Leftrightarrow$~2. To conclude the proof, we shall show that
1~$\Rightarrow$~3 and 3~$\Rightarrow$~2. The direction 3~$\Rightarrow$~2 is
the most involved proof of this paper. We devote three sections to this proof.
In Section~\ref{app:ctrees}, we define a key object for this
proof: \emph{\chains trees}. We then use this object to reduce the proof to
two independent propositions that are then proved in Sections~\ref{app:depth}
and~\ref{app:width}.

\medskip

We finish this section with the much easier 1~$\Rightarrow$~3
direction. Assume that $L$ is a \hbox{\bswd}-definable language and
let $\alpha: A^* \rightarrow M$ be the alphabet completion of its
syntactic morphism. We prove that $\alpha$
satisfies~\eqref{eq:bcs2}. This is an \efgame argument. We begin with
a lemma on alternation schemas, which formalizes the property we
sketched above.

\begin{lemma} \label{lem:schemprop}
  Assume that $(s,s_1,s_2)$ is an alternation schema. Then for all $k
  \in \nat$, there exist $w,w_1,w_2 \in A^*$ such that:
  \begin{itemize}
  \item $\alpha(w) = s, \alpha(w_1) = s_1$ and $\alpha(w_2) = s_2$.
  \item for all $u \in \content{s}^*$, $w \ksieq{2} w_1uw_2$.
  \end{itemize}
\end{lemma}

\begin{proof}
  This is proved using Lemma~\ref{lem:siprop}. Fix an alternation schema
  $(s,s_1,s_2)$ and $k \in \nat$. Let $(r_1,\Rs_1), (r_2,\Rs_2),
  (e,\Es) \in \fC_{2,2}[\alpha]$ be as in the definition of alternation schemas.

  Set $h = 2^{2k}$.
  Since $(e,\Es)$ is idempotent, we have $(e,\Es)^h = (e,\Es)$. By
  definition of \djuns, we obtain words $v_1,v'_1, v_2,v'_2,x,x'_1,x'_2
  \in A^*$ satisfying the following properties:
  \begin{enumerate}[label=$\alph*)$]
  \item $\alpha(v_1) = r_1$,
    $\alpha(v_2) = r_2$,
    $\alpha(x) = e$,
    $\alpha(v'_1x'_1) = s_1$, $\alpha(x'_2v'_2) =s_2$.
  \item $v_1x^{h} \ksieq{2} v'_1x'_1$ and $x^{h}v_2 \ksieq{2} x'_2v'_2$.
  \end{enumerate}
  Set $w = v_1x^{2h}v_2$, $w_1 = v'_1x'_1$ and $w_2 = x'_2v'_2$. It follows from
  Item~$a)$ that $\alpha(w) = r_1er_2= s$, $\alpha(w_1) = s_1$ and
  $\alpha(w_2) = s_2$. Finally, since $\alpha$ is alphabet compatible, we have
  $\content{x}=\content{e}$, and by definition of alternation schemas,
  $\content{e}=\content{s}$. Therefore, it is immediate using \efgame games
  that for any word $u \in \content{s}^*$, $u \ksieq{1} x^{h}$. It then follows
  from Lemma~\ref{lem:siprop} that $x^ {2h}\ksieq{2} x^{h}ux^{h}$, whence by
  Lemma~\ref{lem:efconcat}, that $w\ksieq{2} v_1x^{h}ux^{h}v_2$. Finally, using
  Item $b)$, we conclude that $w \ksieq{2} w_1uw_2$.
\end{proof}

We can now use Lemma~\ref{lem:schemprop} to prove that $\alpha$
satisfies Equation~\eqref{eq:bcs2}. Let $(s,s_1,s_2)$ and
$(t,t_1,t_2)$ be alternation schemas such that $\content{s} =
\content{t}$. Let $w,w_1,w_2 \in A^*$ of images $s,s_2,s_2$ and
$z,z_1,z_2 \in A^*$ of images $t,t_1,t_2$ satisfying the conditions of
Lemma~\ref{lem:schemprop}. We prove that for any $u,v \in A^*$:
\begin{equation}
  u[(z_1w_1)^{N}z(w_2z_2)^{N}]v \ \ \ \kbceq{2}\ \ \
  u[(z_1w_1)^{N}z_1wz_2(w_2z_2)^{N}]v  \label{eq:bcnec2}
\end{equation}
where again $N=2^k\omega$. By definition of the alphabet completion of the
syntactic monoid, of the alphabetic conditions and
since $L$ is defined by a  \bswd formula of rank $k$,
Equation~\eqref{eq:bcs2} will follow. Since  $\content{s} =
\content{t}$, the words $w,w_1,w_2$  and $z,z_1,z_2$  given by
Lemma~\ref{lem:schemprop} satisfy
\begin{align}
  \label{eq:9}
  z &\ksieq{2} z_1wz_2,\\
  \label{eq:10}
  w &\ksieq{2} w_1zw_2.
\end{align}
Using Lemma~\ref{lem:efconcat}, we may multiply \eqref{eq:9} by
$u(z_1w_1)^{N}$ on the left and by $(w_2z_2)^{N}v$ on the right:
\[
  u(z_1w_1)^{N}z(w_2z_2)^{N}v \ \ \ \ksieq{2}\ \ \
  u(z_1w_1)^{N}z_1wz_2(w_2z_2)^{N}v.
\]
For the converse direction, from Lemma~\ref{lem:aperiodic}, we have
$(z_1w_1)^{N} \ksieq{2} (z_1w_1)^{N-1}$
and $(w_2z_2)^{N} \ksieq{2}
(w_2z_2)^{N-1}$. Using   \eqref{eq:10} and Lemma~\ref{lem:efconcat} again, we
conclude that:
\[u(z_1w_1)^{N}z_1wz_2(w_2z_2)^{N}v \ \ \ \ksieq{2}\ \ \
  u(z_1w_1)^{N-1}z_1(w_1zw_2)z_2(w_2z_2)^{N-1}v\]
\emph{i.e.},
\[u(z_1w_1)^{N}z_1wz_2(w_2z_2)^{N}v\ \ \ \ksieq{2}\ \ \ u(z_1w_1)^{N}z(w_2z_2)^{N}v.\]

\section{\texorpdfstring{Proof of Theorem~\lowercase{\ref{thm:caracbc}}:
    \Chain Trees}{Proof of Theorem \ref{thm:caracbc}: Chain Trees}}
\label{app:ctrees}
In this section, we begin the proof of the difficult direction of
Theorem~\ref{thm:caracbc}. Given an alphabet compatible morphism
$\alpha: A^* \to M$, we prove that if Equation~\eqref{eq:bcs2}
is satisfied, then $\Cstwo[\alpha]$ has bounded
alternation. More precisely, we prove the contrapositive: if
$\Cstwo[\alpha]$ has unbounded alternation, then the equation is not
satisfied.

To prove this, we rely on a new notion: ``\emph{\chain trees}''. Chain
trees are a mean to analyze how \dchains with high alternation are
built. In particular, we will use them at the end of the section to
decide which equation is contradicted. Intuitively, a \chain tree is
associated to a single \dchain and represents a computation of our
least fixpoint algorithm of Section~\ref{sec:comput} that generates
this \dchain.

As we explained in the previous section, one can find connections
between our proof and that of the decidable characterization of
boolean combination of open sets of
trees~\cite{bpopen}. In~\cite{bpopen} as well, the authors consider a
notion of  ``\chains'', which corresponds to open sets of trees and
analyze how they are built. This is achieved with an object called
``Strategy Tree''. Though strategy trees and \chain trees share the
same purpose, \emph{i.e.}, analyzing how \chains are built, there is no
connection between the notions themselves since they deal with
completely different objects.

We organize the section in two subsections. First, we define the
general notion of \chain trees. Then, we use \chain trees to reduce
the proof of Theorem~\ref{thm:caracbc} to two independent propositions
(we will then prove these two propositions in Sections~\ref{app:depth}
and~\ref{app:width}).

\subsection{Definition}

Set $\alpha: A^* \rightarrow M$ as an alphabet compatible morphism
into a finite monoid $M$. We associate to $\alpha$ a set
$\ctc[\alpha]$ of \emph{\chain trees}. As we explained, a \chain tree
is associated to a single \dchain for $\alpha$ and represents a way to
compute this \dchain using our least fixpoint algorithm. However,
recall that this algorithm does not work directly with \chains but
with the more general notion of \juns. For this reason, we actually
define two notions:
\begin{enumerate}
\item The set $\cts[\alpha]$ of \emph{\jun trees} associated to
  $\alpha$. Each tree in $\cts[\alpha]$ represents an
  actual computation of the least fixpoint algorithm. Hence, we can
  associate the result of this computation to the tree: this
  \djun is called the \emph{\jun value} of the tree.
\item The set $\ctc[\alpha]$ of \emph{\chain trees}. Each tree in
  $\ctc[\alpha]$ instantiates a \jun tree of $\cts[\alpha]$ and is
  associated to a specific \dchain that belongs to its \jun
  value. This \dchain is called the \emph{\chain value} of the \chain
  tree.
\end{enumerate}

\medskip
\noindent
{\bf \Jun Trees.} For any $n \geqslant 1$, a \emph{\jun tree} $T$ \emph{of
  level} $n$ for $\alpha$ is an ordered unranked tree that may have
four types of nodes: product nodes, operation nodes, initial leaves
and ports. To each node that is not a port, we associate a \emph{\jun
  value}, $\vals{x} \in M \times 2^{M^{n-1}}$, by induction on the
structure of the tree. Intuitively, each type of node corresponds to a
stage of the least fixpoint algorithm while constructing the \jun value of the
tree. We now give a precise definition
of each type of node.

\begin{itemize}
\item {\it Initial Leaves.} An initial leaf $x$ is labeled with a constant
  \dchain $(s,\ldots,s) \in \Cstwon[\alpha]$. We set
  $\vals{x} = (s,\{(s,\ldots,s)\})\in\fCtwon[\alpha]$. Initial leaves
  correspond to the set $\fI_n$ of trivial \djuns, which serves to initialize the
  least fixpoint algorithm when it starts.

\item {\it Ports.} A port is an unlabeled leaf whose parent has to be an
  operation node. In particular, a port may never be the root of the tree.
  Ports have no \jun value and are simply placeholders that get replaced by
  true leaves when the \jun tree is instantiated into a \chain tree (see
  below).

\item {\it Product Nodes.} A product node $x$ is unlabeled. It has
  exactly two children $x_1$ and $x_2$, which may be of any node type except
  `port'. We set $\vals{x} = \vals{x_1} \cdot \vals{x_2}$. Product nodes
  correspond to Operation~\eqref{eq:mul} in the fixpoint algorithm.

\item {\it Operation Nodes.} An operation node $x$ is unlabeled and has
  exactly $3$ children $x_1,x_2$ and $x_3$ from left to right. The middle
  child $x_2$ has to be a port. The left and right children, $x_1$ and $x_3$
  may be of any node type except `port'. However, the trees rooted in $x_1$ and
  $x_3$ must be \emph{identical}. Moreover, we require $\vals{x_1} = \vals{x_3}$
  to be an \emph{idempotent} $(e,\Es)$ of $M \times 2^{M^{n-1}}$.
  Finally, we set the \jun value of the operation node $x$ as
  $\vals{x} = (e,\Es) \cdot (1_M,\Ts) \cdot (e,\Es)$ with
  $\Ts = \{(t_1,\dots,t_{n-1}) \in \Cs_{2,n-1}[\alpha] \mid \content{t_1} =
  \content{e}\}$.
  Operation nodes and ports correspond to Operation~\eqref{eq:oper} in the
  fixpoint algorithm.
\end{itemize}

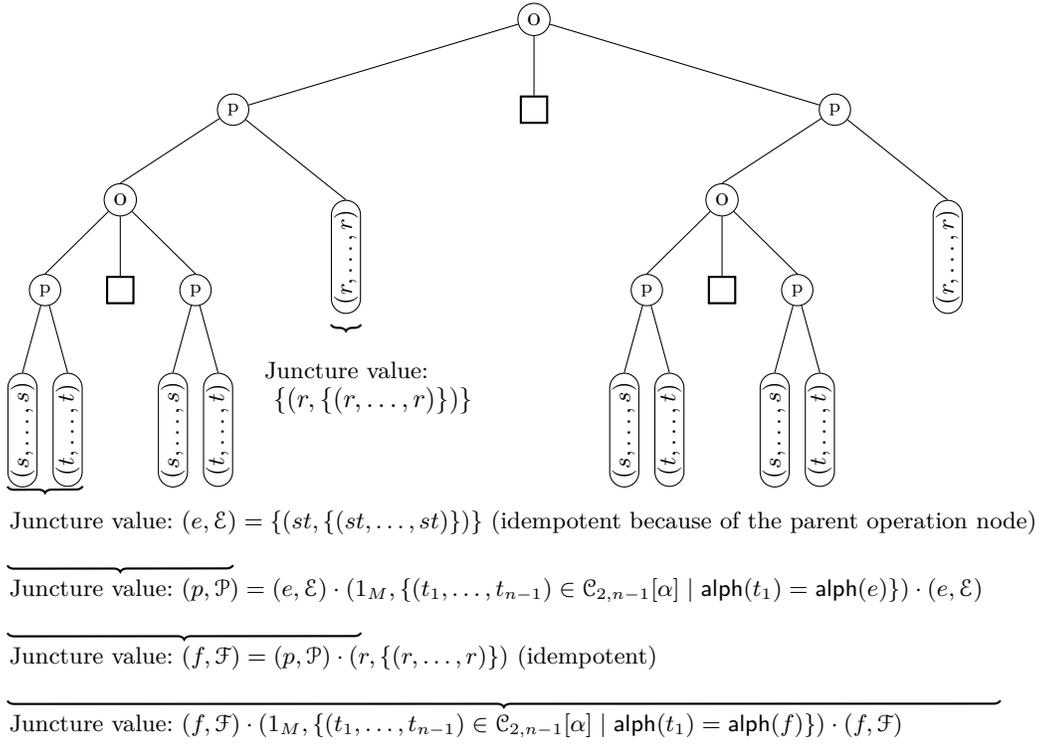
\begin{figure}[ht]
  \begin{center}
    \begin{tikzpicture}
      \node[nod] (x1) at (-4.0,-1.2) {o};
      \node[port] (p1) at (-4.0,-2.4) {};

      \draw (x1) to (p1);

      \node[nod] (y1) at (-8.0,-2.4) {\scriptsize p};
      \node[nod] (y2) at (0.0,-2.4) {\scriptsize p};

      \node[nod] (x11) at (-9.5,-3.6) {o};
      \node[nod] (x12) at (-1.5,-3.6) {o};

      \node[port] (p2) at (-9.5,-4.8) {};
      \node[port] (p3) at (-1.5,-4.8) {};

      \draw (x11) to (p2);
      \draw (x12) to (p3);

      \node[nol,anchor=east] (x21) at (-6.5,-3.6) {\small $(r,\dots,r)$};
      \node[nol,anchor=east] (x22) at (1.5,-3.6) {\small $(r,\dots,r)$};

      \node[nod] (x31) at (-10.5,-4.8) {\scriptsize p};
      \node[nod] (x32) at (-8.5,-4.8) {\scriptsize p};
      \node[nod] (x33) at (-2.5,-4.8) {\scriptsize p};
      \node[nod] (x34) at (-0.5,-4.8) {\scriptsize p};

      \begin{scope}[xshift=0.5cm]
        \def\esp{\hspace*{.3mm}}
        \node[nol,anchor=east] (z1) at (-11.3,-5.9) {\small $(s,\dots,s)$};
        \node[nol,anchor=east] (z2) at (-10.7,-5.9) {\small $\esp(t,\dots,t)\esp$};
        \node[nol,anchor=east] (z3) at (-9.3,-5.9) {\small $(s,\dots,s)$};
        \node[nol,anchor=east] (z4) at (-8.7,-5.9) {\small $\esp(t,\dots,t)\esp$};
        \node[nol,anchor=east] (z5) at (-3.3,-5.9) {\small $(s,\dots,s)$};
        \node[nol,anchor=east] (z6) at (-2.7,-5.9) {\small $\esp(t,\dots,t)\esp$};
        \node[nol,anchor=east] (z7) at (-1.3,-5.9) {\small $(s,\dots,s)$};
        \node[nol,anchor=east] (z8) at (-0.7,-5.9) {\small $\esp(t,\dots,t)\esp$};
      \end{scope}

      \draw (x1) to (y1);
      \draw (x1) to (y2);

      \draw (y1) to (x21.east);
      \draw (y1) to (x11.north);
      \draw (y2) to (x22.east);
      \draw (y2) to (x12.north);

      \draw (x11) to (x31.north);
      \draw (x11) to (x32.north);
      \draw (x12) to (x33.north);
      \draw (x12) to (x34.north);

      \draw (x31) to (z1.east);
      \draw (x31) to (z2.east);
      \draw (x32) to (z3.east);
      \draw (x32) to (z4.east);
      \draw (x33) to (z5.east);
      \draw (x33) to (z6.east);
      \draw (x34) to (z7.east);
      \draw (x34) to (z8.east);

      \draw[decorate,decoration={brace},thick] ($(z2)+(+0.2,-.75)$) to coordinate[below] (c1) ($(z1)+(-0.2,-.75)$);
      \node[align=left,anchor=north west] at ($(c1) -(0.6,0.2)$) {\small\small \Jun value: $(e,\Es) = \{(st,\{(st,\dots,st)\})\}$
        (idempotent because of the parent operation node)};

      \draw[decorate,decoration={brace},thick] ($(x21)+(+0.2,-0.9)$) to ($(x21)+(-0.2,-0.9)$);
      \node[align=left,anchor=north] at ($(x21)-(-0.3,1.1)$) {\small \\\small \Jun value: \\$\ \{(r,\{(r,\dots,r)\})\}$};

      \draw[decorate,decoration={brace},thick] ($(z4)+(+0.2,-1.8)$) to coordinate[below] (d1) ($(z1)+(-0.2,-1.8)$);
      \node[align=left,anchor=north west] at ($(c1) -(0.6,1.1)$) {\small \Jun
        value: $(p,\Ps) = (e,\Es) \cdot (1_M,\{(t_1,\dots,t_{n-1}) \in
        \Cs_{2,n-1}[\alpha] \mid \content{t_1} = \content{e}\}) \cdot (e,\Es)$};

      \draw[decorate,decoration={brace},thick] ($(z4)+(+1.9,-2.7)$) to
      coordinate[below] (e1) ($(z1)+(-0.2,-2.7)$);
      \node[align=left,anchor=north west] at ($(c1) -(0.6,2.0)$) {\small \Jun
        value: $(f,\Fs) = (p,\Ps) \cdot (r,\{(r,\dots,r)\})$ (idempotent)};

      \draw[decorate,decoration={brace},thick] ($(z4)+(+10.4,-3.6)$) to
      coordinate[below] (f1) ($(z1)+(-0.2,-3.6)$);
      \node[align=left,anchor=north west] at ($(c1) -(0.6,2.9)$) {\small \small \Jun value: $(f,\Fs) \cdot (1_M,\{(t_1,\dots,t_{n-1}) \in
        \Cs_{2,n-1}[\alpha] \mid \content{t_1} = \content{f}\}) \cdot (f,\Fs)$};

      \node[anchor=west] at (-9.75,-11.7) {$=$ Operation Node (no label)};
      \node[anchor=west] at (-9.75,-12.2) {$=$ Product Node (no label)};
      \node[anchor=west] at (-3.75,-11.7) {$=$ Initial Leaf (label written inside)};
      \node[anchor=west] at (-3.75,-12.2) {$=$ Port (no label)};

      \node[nod] at (-10.0,-11.7) {o};
      \node[nod] at (-10.0,-12.2) {\scriptsize p};
      \node[nol] at (-4.0,-11.7) {};
      \node[port] at (-4.0,-12.2) {};
    \end{tikzpicture}
  \end{center}
  \caption{An example of \jun tree of level $n$}
  \label{fig:ctree}
\end{figure}

We denote by $\cts[\alpha]$ the set of \jun trees that can be
associated to $\alpha$. If $T \in \cts[\alpha]$, we denote by
$\vals{T}$ the \jun value of the root of $T$. An example of \jun tree
is given in Figure~\ref{fig:ctree}. The following proposition is
essentially an alternate statement of Proposition~\ref{prop:compu}.

\begin{proposition} \label{prop:ctt}
  Let $n \geqslant 1$. Then,
  \[
    \begin{array}{lcl}
      \fCtwon[\alpha] & = & \downclos \big\{\vals{T} \mid T \in \cts[\alpha] \text{ with
                            level $n$}\big\},\\[1ex]
      \fCtwo[\alpha] & = & \downclos \big\{\vals{T} \mid T \in \cts[\alpha]\big\}.

    \end{array}
  \]
\end{proposition}

\begin{proof}
  Immediate from Proposition~\ref{prop:compu}.
\end{proof}

\medskip
\noindent
{\bf \Chain Trees.} \Chain trees are obtained by instantiating \jun
trees. Let $T$ be a \jun tree and let $n$ be its level. An
\emph{instantiation} of $T$ is a new tree $T'$ which is obtained
from $T$ by replacing all ports with new \emph{operation leaves}.

An operation leaf $x$ is labeled with a \chain of length
$n$. Moreover, this \chain has to satisfy an additional condition
with respect to its parent. Observe first that since ports carry no
information in \jun trees, \vals{t} remains well defined for any node
$t$ of $T'$ that is not a new operation leaf. Since $x$ replaces a
port, its parent $z$ has to be an operation node. We ask the label of
$x$ to be chosen in $\vals{z}$, \emph{i.e.}, in the \jun $(e,\Es) \cdot
(1_M,\Ts) \cdot (e,\Es)$ where $(e,\Es)$ is the (idempotent) \jun
value shared by the left and right children of $z$ and $\Ts =
\big\{(t_1,\dots,t_{n-1}) \in \Cs_{2,n-1}[\alpha] \mid \content{t_1} =
\content{e}\big\}$. Note that by Proposition~\ref{prop:ctt}, this means
that the label of $x$ belongs to $\Cs_{2,n}[\alpha]$.

\begin{figure}[ht]
  \begin{center}
    \begin{tikzpicture}
      \node[nod] (x1) at (-4.0,-1.2) {o};
      \node[nop,anchor=east] (p1) at (-4.0,-2.4) {\small
        $(f,r_1,\dots,r_{n-1})$};

      \node[align=center] (i1) at (-5.5,-7.4) {\small Label in\\\small $(f,\Fs)
        \cdot (1_M,\Rs) \cdot (f,\Fs)$\\$\Rs = \{(r_1,\dots,r_{n-1}) \in
        \Cs_{2,n-1}[\alpha] \mid \content{r_1} = \content{f}\}$};

      \draw[thick,->] (i1) to [out=90,in=-90] (p1.west);

      \draw (x1) to (p1);

      \node[nod] (y1) at (-8.0,-2.4) {\scriptsize p};
      \node[nod] (y2) at (0.0,-2.4) {\scriptsize p};

      \node[nod] (x11) at (-9.5,-3.6) {o};
      \node[nod] (x12) at (-1.5,-3.6) {o};

      \node[nop,anchor=east] (p2) at (-9.5,-4.8) {\small $(e,t_1,\dots,t_{n-1})$};
      \node[nop,anchor=east] (p3) at (-1.5,-4.8) {\small $(e,t'_1,\dots,t'_{n-1})$};

      \node[align=center] (i2) at (-5.5,-8.9) {\small Labels in\\\small $(e,\Es)
        \cdot (1_M,\Ts) \cdot (e,\Es)$\\$\Ts = \{(t_1,\dots,t_{n-1}) \in
        \Cs_{2,n-1}[\alpha] \mid \content{t_1} = \content{e}\}$};

      \draw[thick,->] (i2) to [out=180,in=-90] (p2.west);
      \draw[thick,->] (i2) to [out=0,in=-90] (p3.west);

      \draw (x11) to (p2);
      \draw (x12) to (p3);

      \node[nol,anchor=east] (x21) at (-6.5,-3.6) {\small $(r,\dots,r)$};
      \node[nol,anchor=east] (x22) at (1.5,-3.6) {\small $(r,\dots,r)$};

      \node[nod] (x31) at (-10.5,-4.8) {\scriptsize p};
      \node[nod] (x32) at (-8.5,-4.8) {\scriptsize p};
      \node[nod] (x33) at (-2.5,-4.8) {\scriptsize p};
      \node[nod] (x34) at (-0.5,-4.8) {\scriptsize p};

      \begin{scope}[xshift=0.5cm]
        \def\esp{\hspace*{.3mm}}
        \node[nol,anchor=east] (z1) at (-11.3,-5.9) {\small $(s,\dots,s)$};
        \node[nol,anchor=east] (z2) at (-10.7,-5.9) {\small $\esp(t,\dots,t)\esp$};
        \node[nol,anchor=east] (z3) at (-9.3,-5.9) {\small $(s,\dots,s)$};
        \node[nol,anchor=east] (z4) at (-8.7,-5.9) {\small $\esp(t,\dots,t)\esp$};
        \node[nol,anchor=east] (z5) at (-3.3,-5.9) {\small $(s,\dots,s)$};
        \node[nol,anchor=east] (z6) at (-2.7,-5.9) {\small $\esp(t,\dots,t)\esp$};
        \node[nol,anchor=east] (z7) at (-1.3,-5.9) {\small $(s,\dots,s)$};
        \node[nol,anchor=east] (z8) at (-0.7,-5.9) {\small $\esp(t,\dots,t)\esp$};

      \end{scope}

      \draw (x1) to (y1);
      \draw (x1) to (y2);

      \draw (y1) to (x21.east);
      \draw (y1) to (x11.north);
      \draw (y2) to (x22.east);
      \draw (y2) to (x12.north);

      \draw (x11) to (x31.north);
      \draw (x11) to (x32.north);
      \draw (x12) to (x33.north);
      \draw (x12) to (x34.north);

      \draw (x31) to (z1.east);
      \draw (x31) to (z2.east);
      \draw (x32) to (z3.east);
      \draw (x32) to (z4.east);
      \draw (x33) to (z5.east);
      \draw (x33) to (z6.east);
      \draw (x34) to (z7.east);
      \draw (x34) to (z8.east);

    \end{tikzpicture}
  \end{center}
  \caption{An instantiation of the \jun tree in Figure~\ref{fig:ctree}}
  \label{fig:ctree2}
\end{figure}
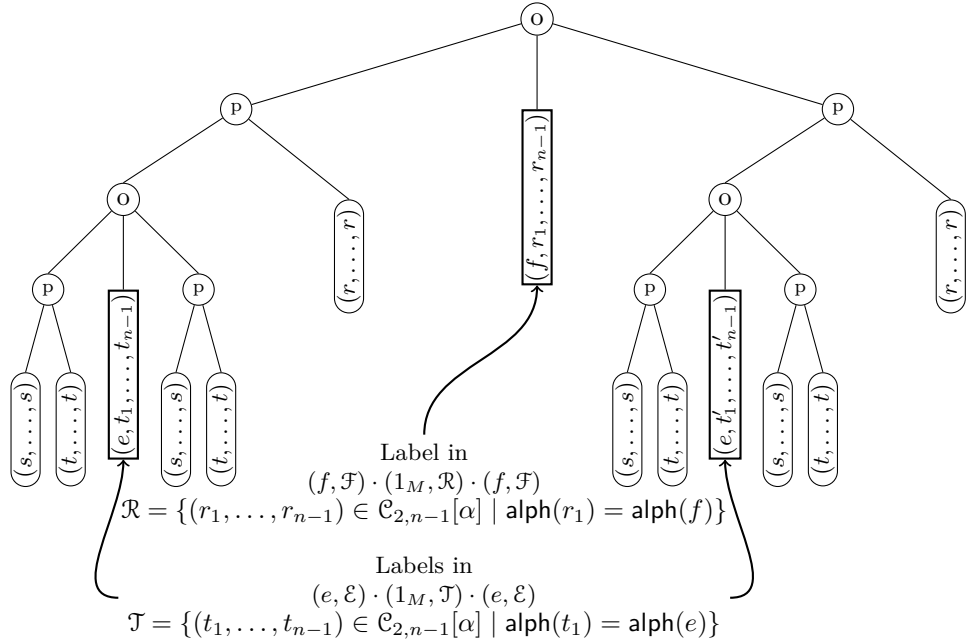

For every \jun tree $T$, we denote by \inst{T} the set of
instantiations of $T$ (see Figure~\ref{fig:ctree2} for an
example). The set $\ctc[\alpha]$ of \emph{\chain trees} associated to
$\alpha$ is the set $\bigcup_{T \in \cts[\alpha]} \inst{T}$. Finally, if $R \in
\ctc[\alpha]$ is of level $n$, to every node $x$ of $R$, we associate a
second value $\valc{x} \in M^n$, called the \emph{\chain value} of $x$.
If $x$ is an initial (resp. operation) leaf, $\valc{x}$ is simply the
label of $x$. If $x$ is a product node with children $x_1$ and $x_2$,
then $\valc{x} = \valc{x_1} \cdot \valc{x_2}$. Finally, if $x$ is an
operation node with children $x_1,x_2$ and $x_3$, then $\valc{x} =
\valc{x_1} \cdot \valc{x_2} \cdot \valc{x_3}$. We set $\valc{R}$ as
the \chain value of the root of $R$. The following facts are immediate
from the definitions.

\begin{fact} \label{fct:value}
  Let $T \in \ctc[\alpha]$ and let $x_1,\dots,x_m$ be its leaves listed
  from left to right. Then $\valc{T} = \valc{x_1} \cdots \valc{x_m}$.
\end{fact}

\begin{fact} \label{fct:value2}
  Let $T \in \ctc[\alpha]$ of level $n$ and let $x$ be a node of $T$. Then $\valc{x} \in
  \Cstwon[\alpha]$.
\end{fact}

We now prove that the definition of \chain trees matches our purpose, \emph{i.e.}, that the set of
\dchains is exactly the set of values of trees in $\ctc[\alpha]$. This
is a corollary of the following proposition.

\begin{proposition} \label{prop:ctree}
  Let $T \in \cts[\alpha]$. Then
  \begin{equation}
    \label{eq:11}
    \vals{T} = \bigl\{\valc{T'} \mid T' \in \inst{T}\bigr\}.
  \end{equation}
\end{proposition}

\begin{proof}
  Before proving the statement, note that $\vals{T}$ is a \jun, while the
  right member of~\eqref{eq:11} is a set of \chains. To simplify the notation,
  we identify in this proof the \jun $\vals{T}$ in~\eqref{eq:11} with the set
  of \chains it contains, \emph{i.e.}, $\{(s,\bar s)\mid\{(s,\{\bar s\})\}\subseteq \vals{T}\}$.

  We proceed by induction on the structure of $T$ (which is shared with any
  \chain tree $T' \in \inst{T}$). If $T$ is a single initial leaf, then
  $\inst{T} = \{T\}$ since there is no port in $T$ to replace, and the result
  is by definition. Otherwise let $x$ be the root of $T$.

  If $x$ is a product node, then let $T_1$ and $T_2$ be the
  subtrees rooted at its children. By induction hypothesis, we have
  $\vals{T_1} = \bigl\{\valc{T'_1} \mid T'_1 \in \inst{T_1}\bigr\}$ and
  $\vals{T_2} = \bigl\{\valc{T'_2} \mid T'_2 \in \inst{T_2}\bigr\}$. By
  definition, $\vals{T} = \vals{T_1} \cdot \vals{T_2}$ and
  \[
    \bigl\{\valc{T'} \mid T' \in \inst{T}\bigr\} = \bigl\{\valc{T_1'} \cdot \valc{T_2'}
    \mid T'_1 \in \inst{T_1} \text{ and } T'_2 \in \inst{T_2}\bigr\},
  \]
  \noindent
  which terminates this case.

  If $x$ is an operation node, let $R$ be the single \jun tree that is
  rooted in both its left and right children and let $(e,\Es) =
  \vals{R}$. By definition, $\vals{T} = (e,\Es) \cdot (1_M,\Ts) \cdot
  (e,\Es)$ with $\Ts = \{(t_1,\dots,t_{n-1}) \in \Cs_{2,n-1}[\alpha]
  \mid \content{t_1} = \content{e}\}$. Moreover, since $(e,\Es)$ is
  idempotent by definition, we have:
  \[
    \vals{T} = (e,\Es) \cdot (e,\Es) \cdot (1_M,\Ts) \cdot (e,\Es) \cdot (e,\Es)
  \]
  This terminates the proof since the set of values that can be given to
  an operation leaf replacing the port child of $x$ is exactly $(e,\Es)
  \cdot (1_M,T) \cdot (e,\Es)$ and by induction hypothesis, $(e,\Es) =
  \vals{R} = \{\valc{R'} \mid R' \in \inst{R}\}$.
\end{proof}

The following corollary states that the set of \dchains is exactly the
set of \chain values of \chain trees and is immediate from
Proposition~\ref{prop:ctt} and Proposition~\ref{prop:ctree}.

\begin{corollary} \label{cor:ctree}
  Let $B \subseteq A$, $n \in \nat$. Then,
  \[
    \begin{array}{lcl}
      \Cstwon[\alpha] & = & \big\{\valc{T} \mid T \in \ctc[\alpha] \text{ with level $n$}\big\},\\[1ex]
      \Cstwo[\alpha] & = & \big\{\valc{T} \mid T \in \ctc[\alpha]\big\}.
    \end{array}
  \]
\end{corollary}

\medskip
\noindent
{\bf Alternation and Recursive Alternation of a \Chain Tree.} The
\emph{alternation} of a \chain tree is the alternation of its
\chain value. We say that a set of \chain trees \cs has
\emph{unbounded alternation} if the set $\{\valc{T} \mid T \in \cs\}$
has unbounded alternation. Note that by Proposition~\ref{prop:ctree},
$\Cstwo[\alpha]$ has unbounded alternation if and only if
$\ctc[\alpha]$ has unbounded alternation.

In the proof, we will be interested in another property of \chain
trees: \emph{recursive alternation}. Recursive alternation corresponds
to the maximal alternation of labels at operation leaves in the tree.
More precisely, if $T$ is a \chain tree, its \emph{recursive
  alternation} is the largest integer $j$ such that there exists an
\emph{operation leaf} in $T$ whose label has alternation $j$. An
important idea in the proof is to separate the case when we can
find a set of \chain trees with unbounded alternation but bounded
recursive alternation from the converse one. This is what we do in the
following subsection.

\subsection{Applying \Chain Trees to Theorem~\ref{thm:caracbc}}

We prove that 3~$\Rightarrow$~2 in Theorem~\ref{thm:caracbc}. Let
$\alpha: A^* \rightarrow M$ be an alphabet compatible morphism into a
finite monoid $M$. We have to prove that if $\alpha$ satisfies
Equation~\eqref{eq:bcs2}, then
$\Cstwo[\alpha]$ has bounded alternation.

The proof is by contrapositive. We assume that $\Cstwo[\alpha]$ has
unbounded alternation and prove that $\alpha$ does not satisfy the
equation. Using \chain trees, we separate the argument into two
independent cases. These two cases are 
stated in the following propositions.

\begin{proposition} \label{prop:width}
  Assume that there exists a set of \chain trees
  $\cs \subseteq \ctc[\alpha]$ with unbounded alternation but
  bounded recursive alternation. Then $\alpha$ does not satisfy
  both equations in~\eqref{eq:bcs1}.
\end{proposition}

\begin{proposition} \label{prop:depth}
  Assume that there exists a set of \chain trees $\cs \subseteq
  \ctc[\alpha]$ with unbounded alternation and that all such sets have
  unbounded recursive alternation. Then $\alpha$ does not satisfy
  Equation~\eqref{eq:bcs2}.
\end{proposition}

Proposition~\ref{prop:width} and Proposition~\ref{prop:depth} are
both involved and proved in Section~\ref{app:width} and
Section~\ref{app:depth} respectively. We finish this section by
using them to conclude the proof of Theorem~\ref{thm:caracbc}.

By hypothesis, $\Cs_2[\alpha]$ has unbounded alternation. Hence, it
follows from Corollary~\ref{cor:ctree} that $\ctc[\alpha]$ also has
unbounded alternation. Therefore, there exists at least one set of \chain
trees $\cs$ with unbounded alternation. If $\cs$ can be chosen with
bounded recursive alternation, it follows from
Proposition~\ref{prop:width} that there is a contradiction to
one of the equations in~\eqref{eq:bcs1} and therefore
to~\eqref{eq:bcs2} by Lemma~\ref{lem:equiveq}. Otherwise,
there is a contradiction to Equation~\eqref{eq:bcs2} by
Proposition~\ref{prop:depth}, which terminates the proof.

\section{\texorpdfstring{Proof of Proposition~\lowercase{\ref{prop:depth}}}{Proof of Proposition \ref{prop:depth}}}
\label{app:depth}
In this section, we prove Proposition~\ref{prop:depth}. Recall that we
have fixed an alphabet compatible morphism $\alpha: A^* \rightarrow M$ into a finite monoid
$M$. Assume that there exists a set of \chain trees $\cs \subseteq
\ctc[\alpha]$ with unbounded alternation and that all such sets have
unbounded recursive alternation. We need to prove that $\alpha$ does
not satisfy Equation~\eqref{eq:bcs2}.

\smallskip

We rely on a new object that is specific to this case, the
\emph{\chain graph}. A \chain graph describes a construction process
for a subset of the set of \dchains for $\alpha$. While this subset
may not be the whole set of \dchains for $\alpha$, we will prove that under the hypothesis of
Proposition~\ref{prop:depth}, it is sufficient to derive a
contradiction to Equation~\eqref{eq:bcs2}.

\medskip
\noindent {\bf The \Chain Graph}. We define a directed graph $G[\alpha]=(V,E)$
whose edges are unlabeled ($E \subseteq V \times V$). We call $G[\alpha]$ the
\emph{\chain graph} of $\alpha$. The set $V$ of nodes of $G[\alpha]$ is the
set $V = M^2 \times M$. We now define the set $E$ of edges of $G[\alpha]$. Let
$((p_1,p_2),s)$ and $((q_1,q_2),t)$ be two nodes of $G[\alpha]$, then $E$
contains an edge from $((p_1,p_2),s)$ to $((q_1,q_2),t)$ if there exist
$s_1,s_2 \in M$ such that $(s,s_1,s_2) \in M^3$ is an alternation schema,
$p_1 \cdot s_1 = q_1$, and $s_2 \cdot p_2 = q_2$. Observe that this definition
does not depend on $t$.

\smallskip
Define the \emph{value} of a node $((p_1,p_2),s)$ as $p_1sp_2$, and its
\emph{alphabet} as $\content{s}$ (recall that $\alpha$ is alphabet
compatible).

\smallskip We say that $G[\alpha]$ is \emph{recursive} if it contains a cycle such that
\begin{enumerate}[label=$\alph*)$]
\item all nodes in the cycle have the same alphabet,
\item the cycle contains two nodes with different values.
\end{enumerate}
Such a cycle is called \emph{productive}. We now prove
Proposition~\ref{prop:depth} as a consequence of the two following
propositions.

\begin{proposition} \label{prop:graphcont1}
  Assume that $G[\alpha]$ is recursive. Then $\alpha$ does not
  satisfy~\eqref{eq:bcs2}.
\end{proposition}

\begin{proposition} \label{prop:graphcont2}
  Assume that there exists a set of \chain trees $\cs \subseteq
  \ctc[\alpha]$ with unbounded alternation and that all such sets have
  unbounded recursive alternation. Then $G[\alpha]$ is recursive.
\end{proposition}

Observe that Proposition~\ref{prop:depth} is an immediate consequence of
Propositions~\ref{prop:graphcont1} and~\ref{prop:graphcont2}. Before proving
them, note that the notion of \chain graph is inspired from the notion of
strategy graph in~\cite{bpopen}. This is because both notions are designed to
derive contradictions to similar equations. However, our proof remains fairly
different from the one of~\cite{bpopen}. The reason for this is that the main
difficulty here is proving Proposition~\ref{prop:graphcont2}, \emph{i.e.}, going from
\chain trees (which are unique to our setting) to a recursive \chain graph. On
the contrary, the much simpler proof of Proposition~\ref{prop:graphcont1} is
similar to the corresponding one in~\cite{bpopen}.

\subsection{Proof of Proposition~\ref{prop:graphcont1}}

Assume that $G[\alpha]$ is recursive. By definition, we get a productive cycle in
the graph $G[\alpha]$. We first prove that we may assume this cycle to
consist exactly of two nodes.

\begin{lemma} \label{lem:cycle}
  If $G[\alpha]$ is recursive, it has a productive cycle with exactly two nodes.
\end{lemma}

\begin{proof}
  Since $G[\alpha]$ is recursive, by definition it contains a productive
  cycle, \emph{i.e.}, a cycle whose nodes all share the same alphabet, and
  containing two nodes with different values. In particular, the number $n$ of
  nodes in the cycle is at least 2. If $n=2$, the lemma is immediate. Assume
  that $n \geqslant 3$, we prove that $G[\alpha]$ must contain a productive
  cycle of length $n-1$. The lemma will then follow by induction.

  To construct such a productive cycle of length $n-1$, it suffices to show
  that one can replace any two consecutive nodes
  \[
    ((u_1,u_2),r) \rightarrow ((p_1,p_2),s)
  \]
  in the cycle by a single one having the same value as $((p_1,p_2),s)$. Indeed,
  since the cycle is of length at least 3, there exists such an edge, where
  $((p_1,p_2),s)$ is one of the two nodes having distinct values and the other
  one is not $((u_1,u_2),r)$, meaning that the resulting shortened cycle will still
  exhibit two nodes with distinct values.

  Pick such an edge in the cycle, by definition there exists an alternation
  schema $(r,r_1,r_2)$ such that $u_1r_1=p_1$ and $r_2u_2 = p_2$. Consider the node
  $((u_1,u_2),r_1sr_2)$.
  \begin{itemize}
  \item By definition of an alternation schema and of a productive cycle,
    $\content{r_1sr_2}=\content{rs}=\content{s}$, hence the node
    $((u_1,u_2),r_1sr_2)$ has the same alphabet as all nodes in the cycle.
  \item Its value is $u_1(r_1sr_2)u_2=p_1sp_2$, hence $((u_1,u_2),r_1sr_2)$ has the same
    value as $((p_1,p_2),s)$.
  \item By definition of the graph, any node having an outgoing edge to
    $((u_1,u_2),r)$ also has an outgoing edge to $((u_1,u_2),r_1sr_2)$.
  \item It remains to show that if there is an edge $((p_1,p_2),s) \rightarrow
    ((q_1,q_2),t)$, then there is also an edge $((u_1,u_2),r_1sr_2)\rightarrow
    ((q_1,q_2),t)$.
  \end{itemize}
  Let $(s,s_1,s_2)$ be an alternation schema such that $p_1s_1=
  q_1$ and $s_2p_2 = q_2$ (such an alternation schema exists by definition of the edges).
  One can verify that $(r_1sr_2,r_1s_1,s_2r_2)$ is an alternation schema
  as well. Moreover, $u_1r_1s_1 =p_1s_1 = q_1$ and $s_2r_2u_2 = s_2p_2 =
  q_2$, which proves that there is an edge from $((u_1,u_2),r_1sr_2)$ to
  $((q_1,q_2),t)$.
\end{proof}

We now conclude the proof of Proposition~\ref{prop:graphcont1}: we have to
show that $\alpha$ fails Equation~\eqref{eq:bcs2}. Let $((p_1,p_2),s)$ and $((q_1,q_2),t)$ be two nodes forming a productive cycle of
length 2, as defined in Lemma~\ref{lem:cycle}. We get alternation schemas
$(s,s_1,s_2)$ and $(t,t_1,t_2)$ such that
\begin{enumerate}
\item\label{item:15} $p_1sp_2 \neq q_1tq_2$.
\item\label{item:16} $\content{s} = \content{t}$.
\item\label{item:17} $p_1s_1 = q_1$ and $q_1t_1 = p_1$, hence $p_1 = p_1(s_1t_1)^\omega$.
\item\label{item:18} $s_2p_2 = q_2$ and $t_2q_2 = p_2$, hence $p_2 = (t_2s_2)^\omega p_2$.
\end{enumerate}
By combining Items~\ref{item:17} and~\ref{item:18}, we obtain that
\begin{eqnarray*}
  p_1sp_2 & = & p_1(s_1t_1)^\omega s (t_2s_2)^\omega p_2 \\
  q_1tq_2 & = & p_1(s_1t_1)^\omega s_1 t s_2 (t_2s_2)^\omega p_2
\end{eqnarray*}
Hence, since $\content{s} = \content{t}$, Equation~\eqref{eq:bcs2}
would require that $p_1sp_2 = q_1tq_2$ which contradicts Item~\ref{item:15}
above. We conclude that Equation~\eqref{eq:bcs2} is not satisfied by
$\alpha$.

\subsection{Proof of Proposition~\ref{prop:graphcont2}}

In the remainder of the section, we assume that $\alpha$ satisfies the
hypothesis of Proposition~\ref{prop:graphcont2}. We prove that $G[\alpha]$ is
recursive by constructing a productive cycle.

We say that a node $((p_1,p_2),s)$ of $G[\alpha]$ is \emph{alternating} if for
all $n$, there exists $(s_1,\dots,s_n) \in \Cstwon[\alpha]$ such that
$s_1 = s$ and the \chain $(p_1s_1p_2,\dots,p_1s_np_2)$ has alternation $n-1$.

\begin{lemma} \label{lem:alt1}
  $G[\alpha]$ contains at least one alternating node.
\end{lemma}

\begin{proof}
  By hypothesis, $\Cstwo[\alpha]$ has unbounded alternation. It follows that
  there exists a least one $s \in M$ such that there are \dchains with
  arbitrary high alternation and $s$ as first element. By definition, the node
  $((1_M,1_M),s)$ is then alternating.
\end{proof}

For the remainder of the proof we define $B$ as a minimal alphabet such that
there exists an alternating node $((p_1,p_2),s)$ in $G[\alpha]$ with
$\content{s} = B$. By this we mean that the only $C \subseteq B$ such that
there exists an alternating node $((q_1,q_2),t)$ in $G[\alpha]$ with
$\content{t} = C$ is~$B$~itself.

\begin{lemma} \label{lem:alt2} Let $((p_1,p_2),s)$ be an alternating node of
  $G[\alpha]$ with $\content{s} = B$. Then there exists an alternating node
  $((q_1,q_2),t)$ such that
  \begin{enumerate}
  \item $\content{t} = B$.
  \item there exists an edge from $((p_1,p_2),s)$ to $((q_1,q_2),t)$.
  \item $p_1sp_2 \neq q_1tq_2$.
  \end{enumerate}
\end{lemma}

By definition $G[\alpha]$ has finitely many nodes. Therefore,
since by Lemma~\ref{lem:alt1}, there exists at least one alternating
node, it is immediate from Lemma~\ref{lem:alt2} that $G[\alpha]$ must
contain a cycle witnessing that $G[\alpha]$ is recursive. This
terminates the proof of Proposition~\ref{prop:graphcont2}. It remains
to prove Lemma~\ref{lem:alt2}. We present the proof in the next
subsection.

\subsection{Proof of Lemma~\ref{lem:alt2}}

Let $((p_1,p_2),s)$ be an alternating node of $G[\alpha]$ with
$\content{s} = B$. We need to construct a node $((q_1,q_2),t)$
satisfying the conditions of the lemma (namely, a successor of $((p_1,p_2),s)$
with a different value and the same minimal alphabet $B$). Since $((p_1,p_2),s)$ is
alternating, there exists a set of \dchains $\Ss$ such that for every
\chain $(s_1,\dots,s_n)$ of $\Ss$, we have $s= s_1$ and
$(p_1s_1p_2,\dots,p_1s_np_2)$ has alternation $n-1$. By
Corollary~\ref{cor:ctree}, this yields a set of \chain trees \cs such
that $\Ss = \{\valc{T} \mid T \in \cs\}$. By construction, \cs has
unbounded alternation and hence unbounded recursive alternation by
hypothesis in Proposition~\ref{prop:graphcont2}.

We now proceed in two steps. First we use \cs to construct a new set of \chain
trees \crr and that satisfies an additional property that we call \emph{local
  optimality}. We then choose a tree $T \in \crr$ with large enough recursive
alternation and use it to construct the desired node~$((q_1,q_2),t)$.

\medskip
\noindent
{\bf Construction of \crr: Local Optimality.} Let us first define
local optimality. Note that the definition depends on the pair
$(p_1,p_2)$. Let $T$ be a \chain tree, $x$ be any operation node in
$T$ and $(t_1,\ldots,t_{n}) = \valc{x}$. We say that $x$ is
\emph{locally optimal} if for all $i < n$, either $t_i = t_{i+1}$ or
the \chain tree $T_i$ obtained from $T$ by replacing the label of $x$
by $(t_1,\dots,t_{i-1},t_{i},t_{i},t_{i+2},\dots,t_n)$ satisfies
\[
  (p_1,\dots,p_1) \cdot \valc{T} \cdot (p_2,\dots,p_2) \neq
  (p_1,\dots,p_1) \cdot \valc{T_i} \cdot (p_2,\dots,p_2).
\]
Intuitively this means that for all $i$, alternating from $t_i$ to $t_{i+1}$
in the label of $x$ is necessary to maintain the value of the tree (in the
context determined by $(p_1,\ldots,p_1)$ and $(p_2,\ldots,p_2)$).  We say that a \chain tree $T$ is
\emph{locally optimal} if {\bf all its operation leaves} are locally
optimal.

\begin{lemma} \label{lem:optimal}
  There exists a set of locally optimal \chain trees \crr such that for
  any $(u_1,\dots,u_n) \in \{\valc{T} \mid T \in \crr\}$, we have
  $s = u_1$ and $(p_1u_1p_2,\dots,p_1u_np_2)$ has alternation~$n-1$.
\end{lemma}

\begin{proof}
  From any \chain tree $T$, we construct a new \chain
  tree $T'$ such that
  \begin{enumerate}
  \item\label{item:13} $(p_1,\dots,p_1) \cdot \valc{T} \cdot (p_2,\dots,p_2) =
    (p_1,\dots,p_1) \cdot \valc{T'} \cdot (p_2,\dots,p_2)$.
  \item\label{item:14} $\valc{T}$ and $\valc{T'}$ have the same first element.
  \item $T'$ is locally optimal.
  \end{enumerate}
  It then suffices to let \crr be the set of all trees
  $T'$ constructed in this way from trees $T$ of $\cs$.

  Let $T$ be any \chain tree of level $n$. For all $i < n$, define the
  \emph{$i$-alternation} of $T$ as the number of operation leaves $x$ in $T$ such
  that $\valc{x} = (t_1,\cdots,t_{n})$ with $t_i \neq t_{i+1}$. Finally, define
  the \emph{index of $T$} as the sequence of size $n-1$ of its $i$-alternations, ordered
  increasingly with respect to values of $i$. Note that the lexicographic ordering on
  this set of sequences of fixed length is well-founded. 

  Assume that $T$ is not locally optimal. We explain how to construct a new
  \chain tree $T'$ satisfying \ref{item:13}, \ref{item:14} and
  \begin{enumerate}[label=(3')]
  \item $T'$ has strictly smaller index than $T$.
  \end{enumerate}
  It then suffices to iteratively apply this construction starting from $T$
  until we get the desired locally optimal tree (which must eventually happen
  since the ordering on indices of \chain trees of level $n$ is
  well-founded). The construction of $T'$ is as follows. Since $T$ is not
  locally optimal, there exists an operation leaf $x$ of $T$ that is not
  locally optimal. Let $(t_1,\dots,t_{n}) = \valc{x}$. By hypothesis, there
  exists $i < n$ such that $t_i \neq t_{i+1}$ and the \chain tree $T'$ obtained
  by replacing the label of $x$ by
  $(t_1,\dots,t_{i-1},t_i,t_i,t_{i+2}, \dots,t_n)$ satisfies \ref{item:13}. Since
  this replacement does not modify the first component of any node,
  Property~\ref{item:14} is satisfied as well. Finally, by definition, for any
  $j < i$, $T,T'$ have the same $j$-alternation and $T'$ has strictly smaller
  $i$-alternation. It follows that $T'$ has strictly smaller index than $T$,
  which terminates the proof.
\end{proof}

For the remainder of the section, we assume that $\crr$ is fixed as the set of
locally optimal \chain trees of Lemma~\ref{lem:optimal}. Observe that by definition, $\crr$
has unbounded alternation. Hence, by hypothesis in
Proposition~\ref{prop:graphcont2} it has unbounded recursive alternation as
well.

\medskip\noindent
{\bf Construction of the node $((q_1,q_2),t)$.} We choose a tree $T \in
\crr$. The choice is based on the following lemma.

\begin{lemma} \label{lem:chooseK}
  There exists an integer $k$ such that for all $t_1,t_2 \in M$
  \[
    (t_1,t_2)^k \in \Cslev 2[\alpha] \Rightarrow (t_1,t_2)^* \subseteq \Cslev
    2[\alpha].
  \]
\end{lemma}

\begin{proof}
  If for all $t_1,t_2\in M$, we have $(t_1,t_2)^* \subseteq \Cslev2[\alpha]$, we
  choose $k=1$. Otherwise, since $\Cslev 2[\alpha]$ is closed under subwords
  (Fact~\ref{fct:high}), if $(t_1,t_2)^k\notin \Cslev2[\alpha]$, then
  for all $j\geqslant k$, we have $(t_1,t_2)^j\notin \Cslev2[\alpha]$ as well.
  Therefore, one can define $k$ as the largest integer such
  that there exist $t_1,t_2 \in M$ with $(t_1,t_2)^{k-1} \in \Cslev 2[\alpha]$
  but $(t_1,t_2)^{k} \not\in \Cslev 2[\alpha]$ (with the convention that $(t_1,t_2)^{0} \in \Cslev 2[\alpha]$).
\end{proof}

Set $m = |M|^2 \cdot k$ with $k$ defined as in Lemma~\ref{lem:chooseK}. Since
\crr has unbounded recursive alternation, there exists $T \in \crr$ with
recursive alternation $m$. Let $n$ be the level of~$T$.

We now use $T$ to construct the desired node $((q_1,q_2),t)$ in $G[\alpha]$
fulfilling all properties of Lemma~\ref{lem:alt2}. We begin by summarizing all
hypotheses we have on $T$ (these hypotheses are also represented in
Figure~\ref{fig:depth}). Set $\bar{u} = (u_1,\dots,u_n) = \valc{T}$ and recall
that by choice of $T$ in \crr, we have $u_1 = s$. Let $x_1,\dots,x_h$ be the leaves of
$T$ (from left to right). Recall that by Fact~\ref{fct:value},
$\valc{T} = \valc{x_1} \cdots \valc{x_h}$.

\begin{figure}[ht]
  \begin{center}
    \begin{tikzpicture}

      \coordinate (a1) at (+0.0,+0.0);
      \coordinate (a2) at (-4.5,-4.5);
      \coordinate (a3) at (+4.5,-4.5);
      \coordinate (a4) at (-2.7,-4.5);
      \coordinate (a5) at (+2.7,-4.5);
      \coordinate (a6) at (+0.0,-1.5);

      \draw (a1) to (a2) to (a4) to (a6) to (a5) to (a3) to (a1);

      \node[nok] (z) at ($(a6) -(0.0,0.7)$) {\small $z$};
      \node[nok] (y1) at ($(a6) -(0.8,1.4)$) {\small $y$};
      \node[nok] (y2) at ($(a6) -(-0.8,1.4)$) {\small $y'$};

      \coordinate (b1) at ($(y1) - (0.0,0.5)$);
      \coordinate (b2) at (-2.4,-4.5);
      \coordinate (b3) at (-0.5,-4.5);

      \draw (b1) to (b2) to (b3) to (b1);

      \coordinate (c1) at ($(y2) - (0.0,0.5)$);
      \coordinate (c2) at (2.4,-4.5);
      \coordinate (c3) at (0.5,-4.5);

      \draw (c1) to (c2) to (c3) to (c1);

      \draw (a6) to (z);
      \draw (z) to (y1);
      \draw (z) to (y2);
      \draw (y1) to (b1);
      \draw (y2) to (c1);

      \node[nok] (x) at (0.0,-5.0) {\small $x$};
      \draw (z) to (x);

      \node[nok] at (-4.2,-5.0) {\small $x_1$};
      \node[nok] at (-3.0,-5.0) {\small $x_j$};
      \node[nok] at (+3.0,-5.0) {\small $x_{j'}$};
      \node[nok] at (+4.2,-5.0) {\small $x_h$};

      \node at (-3.6,-5.0) {$\cdots$};
      \node at (+3.6,-5.0) {$\cdots$};

      \draw[thick,decorate,decoration={brace}] (-2.7,-5.3) to node[below]
      {$\bar{r}$ } (-4.5,-5.3);

      \draw[thick,decorate,decoration={brace}] (-0.5,-5.3) to node[below]
      {$\bar{v} = \valc{y}$ } (-2.4,-5.3);

      \draw[thick,decorate,decoration={brace}] (2.4,-5.3) to node[below]
      {$\bar{v}' = \valc{y'}$ } (0.5,-5.3);

      \draw[thick,decorate,decoration={brace}] (4.5,-5.3) to node[below]
      {$\bar{r}'$ } (2.7,-5.3);

      \draw[thick,decorate,decoration={brace}] (4.5,-5.9) to node[below]
      {$\bar{u} = \valc{T} = \bar{r} \cdot \bar{v} \cdot \valc{x}
        \cdot \bar{v}' \cdot \bar{r}'$ } (-4.5,-5.9);

    \end{tikzpicture}
  \end{center}
  \caption{The \chain tree $T$}
  \label{fig:depth}
\end{figure}

By definition of recursive alternation, $T$ must contain an operation
leaf $x \in \{x_1,\dots,x_h\}$ whose label \valc{x} has alternation
$m$. By definition of \chain trees, $x$ is the middle child of an
operation node $z$. We set $y,y'$ as the left and right children of
this node. Finally, we set $j,j' \leqslant h$ such that $x_{j+1}$ is the
leftmost leaf that is a descendant of $y$ and $x'_{j'-1}$ the
rightmost leaf that is a descendant of $y'$ (see
Figure~\ref{fig:depth}). We now define the following \chains:
\[
  \begin{array}{lclcl}
    \bar{t} & = & (t_1,\dots,t_n) & = & \valc{x} \\
    \bar{v} & = & (v_1,\dots,v_n) & = & \valc{y} \\
    \bar{v}' & = & (v'_1,\dots,v'_n) & = & \valc{y'} \\
    \bar{r} & = & (r_1,\dots,r_n) & = & \valc{x_1} \cdots \valc{x_j}  \\
    \bar{r}' & = & (r'_1,\dots,r'_n) & = & \valc{x_{j'}} \cdots \valc{x_h}
  \end{array}
\]
By definition, we have $\valc{T} = \bar{r} \cdot \bar{v} \cdot \bar{t}
\cdot \bar{v}' \cdot \bar{r}'$. Since $x$ is an operation node, there
exists an idempotent $(e,\Es) \in \fC_{2,n}[\alpha]$ such that:
\begin{itemize}
\item $\vals{y} = \vals{y'} = (e,\Es)$.
\item $\bar{v},\bar{v}' \in (e,\Es)$.
\item $\bar{t} \in (e,\Es) \cdot (1_M,\Ts) \cdot (e,\Es)$ with $\Ts =
  \big\{(t_1,\dots,t_{n-1}) \in \Cs_{2,n-1}[\alpha] \mid \content{t_1} = \content{e}\big\}$.
\end{itemize}
By choice of $x$, $\bar{t} = (t_1,\dots,t_n) =\valc{x}$ has alternation $m = |M|^2 \cdot k$. It
follows from the pigeonhole principle that there exist $i$ such that
$t_i\neq t_{i+1}$ and a set $I \subseteq \{1,\dots,n-1\}$ of size at least~$k$
such that for all $j \in I$, $t_j = t_i$ and $t_{j+1} = t_{i+1}$. Note that this implies
that the \chain $(t_i,t_{i+1})^{k}$ is a subword of $(t_1,\dots,t_{n})$, and
therefore a \dchain. By choice of $k$ (see Lemma~\ref{lem:chooseK}) it follows
that $(t_i,t_{i+1})^{*} \subseteq \Cstwo[\alpha]$.

Recall that $T$ is locally optimal since it belongs to \crr. By
definition of local optimality, changing $t_{i+1}$ to $t_i$ in the label
$\valc{x}$ of the operation node $x$ changes the value $\valc{T}$, hence its
$(i+1)$-th component. We therefore obtain the following fact.

\begin{fact} \label{fct:ccont}
  $p_1r_{i+1}v_{i+1}t_iv'_{i+1}r'_{i+1}p_2 \neq p_1r_{i+1}v_{i+1}t_{i+1}v'_{i+1}r'_{i+1}p_2$.
\end{fact}

We now define the node $((q_1,q_2),t)$. We let
\[
  q_1 = p_1r_{i+1}v_{i+1} \text{ and } q_2 = v'_{i+1}r'_{i+1}p_2.
\]
It is immediate
from Fact~\ref{fct:ccont} that either $q_1t_iq_2 \neq p_1sp_2$ or
$q_1t_{i+1}q_2 \neq p_1sp_2$. In the first case, we set $t = t_i$, in
the second, we set $t = t_{i+1}$. Note that since
$(t_i,t_{i+1})^{*} \subseteq \Cstwo[\alpha]$ and $q_1t_iq_2 \neq
q_1t_{i+1}q_2$, the node $((q_1,q_2),t)$ is alternating by definition.

\smallskip
It remains to prove that
\begin{itemize}
\item $\content{t} = \content{s}$, and that
\item there is an edge $((p_1,p_2),s)\to ((q_1,q_2),t)$ in $G[\alpha]$.
\end{itemize}
For the proof, we assume that $t= t_i$ (the case $t= t_{i+1}$ is similar).

\smallskip
Observe that in the \dchain $\valc{T} = (u_1,\dots,u_n)$, $u_1 = s$
and $u_i = p_1r_{i}v_{i}t_iv'_{i}r'_{i}p_2$. Since $(u_1,\dots,u_n)$
is a \dchain, one can verify that all its elements have the
same alphabet, hence $\content{u_i} = \content{s} = B$ and $\content{t}
\subseteq B$. Now, recall that $B$ was chosen as a minimal alphabet
such that there is an alternating node $((q_1,q_2),t)$ with
$\content{t} = B$. Hence, since $((q_1,q_2),t)$ is alternating and
$\content{t}  \subseteq B$, we have $\content{t} = B = \content{s}$.

Finally, we need to prove that there is an edge from $((p_1,p_2),s)$
to $((q_1,q_2),t)$, \emph{i.e.}, to find $s_1,s_2 \in M$ such that
$(s,s_1,s_2)$ is an alternation schema and $p_1s_1 = q_1$ and $s_2p_2
= q_2$. We define $s_1 = r_{i+1}v_{i+1}$ and $s_2 =
v'_{i+1}r'_{i+1}$. That $p_1s_1 = q_1$ and $s_2p_2 = q_2$ is immediate
by definition of $q_1$ and $q_2$. It remains to prove that $(s,s_1,s_2)$ is an
alternation schema.

Recall that $\bar{v},\bar{v}' \in (e,\Es)$ with
$(e,\Es) \in \fC_{2,n}[\alpha]$. Define $\Fs \subseteq M$ as the set
containing all elements that are at component $i$ of some \chain in $\Es$. In
particular $v_{i+1},v'_{i+1} \in \Fs$. It is immediate from
Fact~\ref{fct:high2} (closure of \juns under subwords) that
$(e,\Fs) \in \fC_{2,2}[\alpha]$. Moreover, the idempotency of $(e,\Es)$ entails that $(e,\Fs)$
is also idempotent. By Fact~\ref{fct:high} (closure of \chains under subwords), we have
$(r_1,r_{i+1}) \in \Cstwo[\alpha]$ and $(r'_1,r'_{i+1}) \in \Cstwo[\alpha]$.
Hence we have $(r_1,\{r_{i+1}\}) \in \fC_{2,2}[\alpha]$ and
$(r'_1,\{r'_{i+1}\}) \in \fC_{2,2}[\alpha]$. We conclude that
$s= u_1 = r_1er'_1$, $s_1 \in \{r_{i+1}\} \cdot \Fs$ and
$s_2 \in \Fs \cdot \{r'_{i+1}\}$. Moreover, by definition
$\content{e} = \content{t} = B = \content{s}$: we conclude that $(s,s_1,s_2)$
is an alternation schema, which terminates the proof.\qed

\section{\texorpdfstring{Proof of Proposition \lowercase{\ref{prop:width}}}{Proof of Proposition \ref{prop:width}}}
\label{app:width}
\newcommand\alt[2]{\ensuremath{\textsf{alt}(#1,#2)}\xspace}
In this section, we prove Proposition~\ref{prop:width}. Recall that we
have fixed a morphism $\alpha: A^* \rightarrow M$ into a finite monoid
$M$. Assume that there exists a set of \chain trees
$\cs \subseteq \ctc[\alpha]$ with unbounded alternation but bounded
recursive alternation. We need to prove that $\alpha$ does not satisfy
one of the equations in~\eqref{eq:bcs1}. As for the previous section, we will use a
new object that is specific to this case: \emph{\chain matrices}.

\medskip
\noindent {\bf \Chain Matrices.} Let $n \in \nat$. A \emph{\chain matrix of
  length $n$} is a rectangular matrix with~$n$ columns and whose rows belong
to $\Cstwolen{n}[\alpha]$. If \mat is a \chain matrix, we will denote by
$\mat_{i,j}$ the entry at row $i$ (starting from the top) and column $j$
(starting from the left) in \mat. If \mat is a \chain matrix of length $n$ and
with $m$ rows, we call the \chain
$\bigl((\mat_{1,1} \cdots \mat_{m,1}), \dots, (\mat_{1,n} \cdots
\mat_{m,n})\bigr)$,
the \emph{value} of \mat. Since $\Cstwolen{n}[\alpha]$ is a monoid by
Fact~\ref{fct:chaincomp}, the value of a \chain matrix is a \dchain. We give
an example with $3$ rows in Figure~\ref{fig:valmat}.

\begin{figure}[h]
  \begin{center}
    \begin{tikzpicture}
      \node[anchor=mid] (s1) at (0.0,0) {$s_1$};
      \node[anchor=mid] (s2) at (1.1,0) {$s_2$};
      \node[anchor=mid] (s3) at (2.2,0) {$s_3$};
      \node[anchor=mid] (s4) at (3.3,0) {$s_4$};
      \node[anchor=mid] (s5) at (4.4,0) {$\cdots$};
      \node[anchor=mid] (sn) at (5.5,0) {$s_{n}$};

      \node[anchor=mid] (t1) at (0.0,-0.5) {$t_1$};
      \node[anchor=mid] (t2) at (1.1,-0.5) {$t_2$};
      \node[anchor=mid] (t3) at (2.2,-0.5) {$t_3$};
      \node[anchor=mid] (t4) at (3.3,-0.5) {$t_4$};
      \node[anchor=mid] (t5) at (4.4,-0.5) {$\cdots$};
      \node[anchor=mid] (tn) at (5.5,-0.5) {$t_{n}$};

      \node[anchor=mid] (r1) at (0.0,-1.0) {$r_1$};
      \node[anchor=mid] (r2) at (1.1,-1.0) {$r_2$};
      \node[anchor=mid] (r3) at (2.2,-1.0) {$r_3$};
      \node[anchor=mid] (r4) at (3.3,-1.0) {$r_4$};
      \node[anchor=mid] (r5) at (4.4,-1.0) {$\cdots$};
      \node[anchor=mid] (rn) at (5.5,-1.0) {$r_{n}$};

      \draw (-0.5,0.25) to (-0.5,-1.25);
      \draw (0.55,0.25) to (0.55,-1.25);
      \draw (1.65,0.25) to (1.65,-1.25);
      \draw (2.75,0.25) to (2.75,-1.25);
      \draw (3.85,0.25) to (3.85,-1.25);
      \draw (4.95,0.25) to (4.95,-1.25);
      \draw (6.00,0.25) to (6.00,-1.25);

      \draw (-0.5,0.25) to (6.0,0.25);
      \draw (-0.5,-0.25) to (6.0,-0.25);
      \draw (-0.5,-0.75) to (6.0,-0.75);
      \draw (-0.5,-1.25) to (6.0,-1.25);

      \node[anchor=mid] (p1) at (-0.55,-2.2) {$($};
      \node[anchor=mid] (v1) at (0.0,-2.2) {$s_1t_1r_1,$};
      \node[anchor=mid] (v2) at (1.1,-2.2) {$s_2t_2r_2,$};
      \node[anchor=mid] (v3) at (2.2,-2.2) {$s_3t_3r_3,$};
      \node[anchor=mid] (v4) at (3.3,-2.2) {$s_4t_4r_4,$};
      \node[anchor=mid] (v5) at (4.4,-2.2) {$\dots$};
      \node[anchor=mid] (vn) at (5.5,-2.2) {$,s_{n}t_{n}r_{n}$};
      \node[anchor=mid] (p2) at (6.1,-2.2) {$)$};
      \node[anchor=mid] (text) at (-2.0,-2.2) {Value};
      \draw[ar] (text) to (p1);

      \draw[->] ($(r1)-(0.0,0.4)$) to (v1);
      \draw[->] ($(r2)-(0.0,0.4)$) to (v2);
      \draw[->] ($(r3)-(0.0,0.4)$) to (v3);
      \draw[->] ($(r4)-(0.0,0.4)$) to (v4);
      \draw[->] ($(rn)-(0.0,0.4)$) to (vn);
    \end{tikzpicture}
  \end{center}
  \caption{Value of a \chain matrix with $3$ rows}
  \label{fig:valmat}
\end{figure}
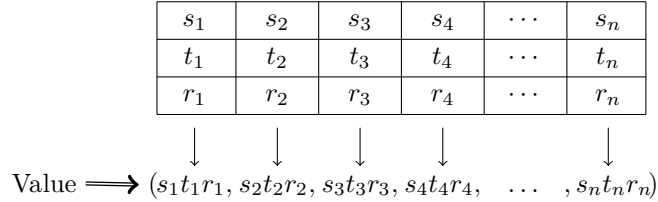

Given a \chain matrix \mat, the \emph{alternation} of \mat is the
alternation of its value. Finally, the \emph{local alternation} of a
\chain matrix, \mat, is the largest integer $m$ such that \mat has a
row with alternation $m$. We now prove the two following propositions.

\begin{proposition} \label{prop:matinit}
  Assume that there exists a set of \chain trees $\cs
  \subseteq \ctc[\alpha]$ with unbounded alternation and recursive
  alternation bounded by $K \in \nat$. Then there exist \chain
  matrices with arbitrarily large alternation and local alternation
  bounded by $K$.
\end{proposition}

\begin{proposition} \label{prop:contradend}
  Assume that there exist \chain matrices with arbitrarily large alternation
  and local alternation bounded by $K\in \nat$. Then $\alpha$ does not
  satisfy~\eqref{eq:bcs1}.
\end{proposition}

Proposition~\ref{prop:width} is an immediate consequence of
Proposition~\ref{prop:matinit} and~\ref{prop:contradend}. Note that
\chain matrices are reused from~\cite{bpopen} (where they are called
``strategy matrices''). Moreover, in this case, going from \chain trees
to \chains matrices (\emph{i.e.}, proving Proposition~\ref{prop:matinit}) is
simple and the main difficulty is proving
Proposition~\ref{prop:contradend}. This means that while our presentation
is different from that of~\cite{bpopen}, the fundamental arguments
themselves are essentially the same. We give a full proof for the sake
of completeness. We begin by proving Proposition~\ref{prop:matinit}.

\begin{proof}[of Proposition~\ref{prop:matinit}]
  We prove that for all $n \in \nat$, there exists a \chain matrix~\mat
  of alternation $n$ and local alternation bounded by $K$. By definition
  of $\cs$, there exists a tree $T \in \cs$ whose value has alternation
  $n$ and has recursive alternation bounded by $K$. Set $x_1,\dots,x_m$
  as leaves of $T$ listed from left to right.
  By Fact~\ref{fct:value}, $\valc{T} = \valc{x_1} \cdots
  \valc{x_m}$. Observe that by definition, for all $i$, $\valc{x_i}$ has
  alternation bounded by $K$. Therefore it suffices to set \mat as the
  $m\times n$ matrix where row $i$ is filled with $\valc{x_i}$.
\end{proof}

It now remains to prove Proposition~\ref{prop:contradend}. We proceed as
follows: assuming there exists a \chain matrix \mat with local alternation
bounded by $K$ and very large alternation, we refine~\mat in several steps to
ultimately obtain a \chain matrix of a special kind that we call a
\emph{contradiction matrix}. There are two types of contradiction matrices,
\emph{increasing} and \emph{decreasing}. Both are \chain matrices of length
$6$ with the following entries:

\begin{center}
  \begin{tikzpicture}

    \node (m1) at (0,0) {
      $\begin{array}{|c|c|c|c|c|c|}
         \hline
         u_1 & v_1 & f & f & f & f\\
         \hline
         e & e & u_2 & v_2 & f & f\\
         \hline
         e & e & e & e & u_3 & v_3\\
         \hline
       \end{array}$};
     \node (l1) at(0,-1.2) {Increasing Contradiction Matrix};

     \node (m2) at (6,0) {$
       \begin{array}{|c|c|c|c|c|c|}
         \hline
         f & f & f & f & u_3 & v_3\\
         \hline
         f & f & u_2 & v_2 & e & e\\
         \hline
         u_1 & v_1 & e & e & e & e\\
         \hline
       \end{array}$};
     \node (l2) at(6,-1.2) {Decreasing Contradiction Matrix};
   \end{tikzpicture}
 \end{center}

 \noindent
 where $e,f$ are idempotents and $fu_2e \neq fv_2e$. As the name suggests, the
 existence of a contradiction matrix contradicts~\eqref{eq:bcs1}.
 This is what we state in the following lemma.

 \begin{lemma} \label{lem:contradmat}
   If there exists a contradiction matrix, then $\alpha$ does not
   satisfy~\eqref{eq:bcs1}.
 \end{lemma}

 \begin{proof}
   Assume that we have an increasing contradiction matrix (the other case
   is treated in a symmetrical way). Since $fu_2e \neq fv_2e$, either
   $fu_2e \neq fe$ or $fv_2e \neq fe$. By symmetry assume it is the
   former. Since $e,f$ are idempotents, this means that $f^\omega u_2
   e^\omega \neq f^\omega e^\omega$. However by definition of \chain
   matrices $(e,u_2,v_2,f) \in \Cstwo[\alpha]$ and therefore  $(e,u_2,f)
   \in \Cstwo[\alpha]$ which contradicts the second equation in~\eqref{eq:bcs1}. Note
   that we only used one half of~\eqref{eq:bcs1}, the other half
   is used in the decreasing case.
 \end{proof}

 By Lemma~\ref{lem:contradmat}, it suffices to prove the existence of a
 contradiction matrix to conclude the proof of
 Proposition~\ref{prop:contradend}. This is what we do in the remainder of
 this section. By hypothesis, we know that there exist \chain matrices with
 arbitrarily large alternation and local alternation bounded by $K \in \nat$.
 For the remainder of the section, we assume that this hypothesis holds. We
 use several steps to prove that we can choose our \chain matrices with
 increasingly strong properties until we get a contradiction matrix. We use
 two intermediaries types of matrices, that we call \emph{Tame \Chain
   Matrices} and \emph{Monotonous \Chain Matrices}. We divide the proof in
 three subsections, one for each step.

 \subsection{Tame \Chain Matrices}

 Let \mat be a \chain matrix of \emph{even length} $2\ell$ and let
 $j \leqslant \ell$. The \emph{set of alternating rows for $j$}, denoted by
 \alt{\mat}{j}, is the set $\{i \mid \mat_{i,2j-1} \neq \mat_{i,2j}\}$. Let
 $(s_1,\dots,s_{2\ell})$ be the value of \mat. We say that \mat is \emph{tame}
 if
 \begin{enumerate}[label=$\alph*)$]
 \item\label{item:19} for all $j \leqslant \ell$, $s_{2j-1} \neq s_{2j}$,
 \item\label{item:20} for all $j \leqslant \ell$, \alt{\mat}{j} is a singleton and
 \item\label{item:21} if $j \neq j'$ then $\alt{\mat}{j} \neq \alt{\mat}{j'}$.
 \end{enumerate}
 We represent a tame \chain matrix of length $6$ in Figure~\ref{fig:tamemat}.
 Observe that the definition only considers the relationship between odd
 columns and the next even column. Moreover, observe that a tame \chain matrix
 of length $2\ell$ has by definition alternation at least $\ell$.

 \begin{figure}[h]
   \begin{center}
     \begin{tikzpicture}
       \draw[pattern=north west lines,draw] (0.0,1.5) to (0.5,1.5) to
       (0.5,2.0) to (0.0,2.0) to (0.0,1.5);
       \draw[pattern=crosshatch dots,draw] (0.0,1.0) to (0.5,1.0) to
       (0.5,1.5) to (0.0,1.5) to (0.0,1.0);
       \draw[pattern=north west lines,draw] (0.0,0.5) to (0.5,0.5) to
       (0.5,1.0) to (0.0,1.0) to (0.0,0.5);
       \draw[pattern=crosshatch dots,draw] (0.0,0.0) to (0.5,0.0) to
       (0.5,0.5) to (0.0,0.5) to (0.0,0.0);

       \draw[pattern=north west lines,draw] (0.5,1.5) to (1.0,1.5) to
       (1.0,2.0) to (0.5,2.0) to (0.5,1.5);
       \draw[pattern=crosshatch dots,draw] (0.5,1.0) to (1.0,1.0) to
       (1.0,1.5) to (0.5,1.5) to (0.5,1.0);
       \draw[pattern=north west lines,draw] (0.5,0.5) to (1.0,0.5) to
       (1.0,1.0) to (0.5,1.0) to (0.5,0.5);
       \draw[pattern=vertical lines,draw] (0.5,0.0) to (1.0,0.0) to (1.0,0.5)
       to (0.5,0.5) to (0.5,0.0);

       \draw[pattern=horizontal lines,draw] (1.0,1.5) to (1.5,1.5) to
       (1.5,2.0) to (1.0,2.0) to (1.0,1.5);
       \draw[pattern=bricks,draw] (1.0,1.0) to (1.5,1.0) to
       (1.5,1.5) to (1.0,1.5) to (1.0,1.0);
       \draw[pattern=north east lines,draw] (1.0,0.5) to (1.5,0.5) to
       (1.5,1.0) to (1.0,1.0) to (1.0,0.5);
       \draw[pattern=north east lines,draw] (1.0,0.0) to (1.5,0.0) to
       (1.5,0.5) to (1.0,0.5) to (1.0,0.0);

       \draw[pattern=crosshatch dots,draw] (1.5,1.5) to (2.0,1.5) to
       (2.0,2.0) to (1.5,2.0) to (1.5,1.5);
       \draw[pattern=bricks,draw] (1.5,1.0) to (2.0,1.0) to
       (2.0,1.5) to (1.5,1.5) to (1.5,1.0);
       \draw[pattern=north east lines,draw] (1.5,0.5) to (2.0,0.5) to
       (2.0,1.0) to (1.5,1.0) to (1.5,0.5);
       \draw[pattern=north east lines,draw] (1.5,0.0) to (2.0,0.0) to
       (2.0,0.5) to (1.5,0.5) to (1.5,0.0);

       \draw[pattern=bricks,draw] (2.0,1.5) to (2.5,1.5) to
       (2.5,2.0) to (2.0,2.0) to (2.0,1.5);
       \draw[pattern=bricks,draw] (2.0,1.0) to (2.5,1.0) to
       (2.5,1.5) to (2.0,1.5) to (2.0,1.0);
       \draw[pattern=crosshatch dots,draw] (2.0,0.5) to (2.5,0.5) to
       (2.5,1.0) to (2.0,1.0) to (2.0,0.5);
       \draw[pattern=crosshatch dots,draw] (2.0,0.0) to (2.5,0.0) to
       (2.5,0.5) to (2.0,0.5) to (2.0,0.0);

       \draw[pattern=bricks,draw] (2.5,1.5) to (3.0,1.5) to
       (3.0,2.0) to (2.5,2.0) to (2.5,1.5);
       \draw[pattern=north east lines,draw] (2.5,1.0) to (3.0,1.0) to
       (3.0,1.5) to (2.5,1.5) to (2.5,1.0);
       \draw[pattern=crosshatch dots,draw] (2.5,0.5) to (3.0,0.5) to
       (3.0,1.0) to (2.5,1.0) to (2.5,0.5);
       \draw[pattern=crosshatch dots,draw] (2.5,0.0) to (3.0,0.0) to
       (3.0,0.5) to (2.5,0.5) to (2.5,0.0);

       \node[anchor=mid,inner sep=1pt] (s1) at (0.25,-0.7) {$s_1$};
       \node[anchor=mid,inner sep=1pt] (s2) at (0.75,-0.7) {$s_2$};
       \node[anchor=mid,inner sep=1pt] (s3) at (1.25,-0.7) {$s_3$};
       \node[anchor=mid,inner sep=1pt] (s4) at (1.75,-0.7) {$s_4$};
       \node[anchor=mid,inner sep=1pt] (s5) at (2.25,-0.7) {$s_5$};
       \node[anchor=mid,inner sep=1pt] (s6) at (2.75,-0.7) {$s_6$};

       \node[anchor=mid east] (text) at (-1.0,-0.7) {Value};
       \draw[ar] (text) to (s1);

       \draw (s1.south) to [in=-90,out=-90]
       node[sloped,draw,fill=white,circle,inner sep=0.5pt] {\tiny $\neq$}
       (s2.south);
       \draw (s3.south) to [in=-90,out=-90]
       node[sloped,draw,fill=white,circle,inner sep=0.5pt] {\tiny $\neq$}
       (s4.south);
       \draw (s5.south) to [in=-90,out=-90]
       node[sloped,draw,fill=white,circle,inner sep=0.5pt] {\tiny $\neq$}
       (s6.south);

       \draw[very thick,red] (3.0,1.5) to (2.0,1.5) to (2.0,1.0) to (3.0,1.0) to (3.0,1.5);
       \draw[very thick,red] (1.0,2.0) to (2.0,2.0) to (2.0,1.5) to (1.0,1.5) to (1.0,2.0);
       \draw[very thick,red] (1.0,0.5) to (0.0,0.5) to (0.0,0.0) to (1.0,0.0) to (1.0,0.5);
     \end{tikzpicture}
   \end{center}
   \caption{A tame \chain matrix of length $6$}
   \label{fig:tamemat}
 \end{figure}
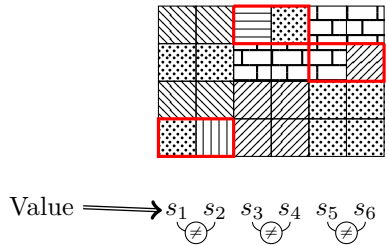

 \begin{lemma} \label{lem:tame}
   There exist tame \chain matrices of arbitrarily large length.
 \end{lemma}

 \begin{proof}
   Set $n \in \nat$, we explain how to construct a tame \chain matrix of
   length $2n$. By hypothesis, there exists a \chain matrix \mat with local
   alternation at most~$K$ and alternation greater than $2nK$. Set $m$ as the
   number of rows of \mat. We explain how to modify \mat to obtain a matrix
   satisfying \ref{item:19}, \ref{item:20} and \ref{item:21}. Recall that
   \dchains are closed under subwords, therefore removing columns from \mat
   yields a \chain matrix. Since \mat has alternation greater than $2nK$, it
   is simple to see that by removing columns one can obtain a \chain matrix of
   length $2nK$ that satisfies $a)$. We denote by \mnat this matrix. We now
   proceed in two steps: first, we modify the entries in \mnat to get a matrix
   \pat of length $2nK$ satisfying both $a)$ and $b)$. Then we use our bound
   on local alternation to remove columns and enforce $c)$ in the resulting
   matrix.

   \medskip
   \noindent
   {\bf Construction of \pat.} Let $j \leqslant nK$ such that \alt{\mnat}{j}
   is of size at least $2$. We modify the matrix to reduce
   the size of \alt{\mnat}{j} while preserving $a)$. One can then repeat
   the operation to get the desired matrix. Let $i \in \alt{\mnat}{j}$.
   Set $s_{1} = \mnat_{1,2j-1} \cdots \mnat_{i-1,2j-1}$ and $s_2 =
   \mnat_{i+1,2j-1} \cdots \mnat_{m,2j-1}$. We distinguish two~cases.

   First, if $s_1\mnat_{i,2j-1}s_2 \neq s_1\mnat_{i,2j}s_2$, then for all $i' \neq i$,
   we replace entry $\mnat_{i',2j}$ with entry $\mnat_{i',2j-1}$. One can verify
   that this yields a \chain matrix of length $2nK$, local alternation bounded by
   $K$. Moreover, it still satisfies $a)$, since $s_1\mnat_{i,2j-1}s_2 \neq
   s_1\mnat_{i,2j}s_2$. Finally, \alt{\mnat}{j} is now a singleton, namely $\{i\}$.

   In the second case, we have $s_1\mnat_{i,2j-1}s_2 = s_1\mnat_{i,2j}s_2$. In that
   case, we replace $\mnat_{i,2j-1}$ with $\mnat_{i,2j}$. One can
   verify that this yields a \chain matrix of length $2nK$, local
   alternation bounded by $K$. Moreover, it still satisfies $a)$ since we did not
   change the value of the matrix. Finally,
   the size of \alt{\mnat}{j} has decreased by~$1$.

   \medskip
   \noindent
   {\bf Construction of the tame matrix.} We now have a \chain matrix
   \pat of length $2nK$, with local alternation bounded by $K$ and
   satisfying both $a)$ and $b)$. Since $a)$ and $b)$ are satisfied, for
   all $j \leqslant nK$ there exists exactly one row $i$ such that
   $\mnat_{i,2j-1} \neq \mnat_{i,2j}$. Moreover, since each row has
   alternation at most $K$, a single row $i$ has this property for at
   most $K$ indices $j$. Therefore, it suffices to remove at most
   $n(K-1)$ pairs of odd-even columns to get a matrix that satisfies
   $c)$. Since the original matrix had length $2nK$, this leaves a matrix
   of length at least $2n$, as desired.
 \end{proof}

 \subsection{Monotonous \Chain Matrices}

 Let \mat be a tame \chain matrix of length $2n$ and let $x_1,\dots,x_n$
 be integers such that for all $j$, $\alt{\mat}{j} = \{x_j\}$. We say
 that \mat is a \emph{monotonous \chain matrix} if it has exactly $n$
 rows and $1 = x_{1} < x_2 < \cdots < x_{n} = n$ (in which case the
 matrix is said \emph{increasing}) or $n = x_{1} > x_2 > \cdots > x_{n} = 1$
 (in which case we say the matrix is \emph{decreasing}). We give a
 representation of the increasing case in Figure~\ref{fig:inctame}.

 \begin{figure}[h]
   \begin{center}
     \begin{tikzpicture}
       \draw[pattern=north west lines,draw] (0.0,1.5) to (0.5,1.5) to
       (0.5,2.0) to (0.0,2.0) to (0.0,1.5);
       \draw[pattern=bricks,draw] (0.0,1.0) to (0.5,1.0) to
       (0.5,1.5) to (0.0,1.5) to (0.0,1.0);
       \draw[pattern=north west lines,draw] (0.0,0.5) to (0.5,0.5) to
       (0.5,1.0) to (0.0,1.0) to (0.0,0.5);
       \draw[pattern=crosshatch dots,draw] (0.0,0.0) to (0.5,0.0) to
       (0.5,0.5) to (0.0,0.5) to (0.0,0.0);

       \draw[pattern=bricks,draw] (0.5,1.5) to (1.0,1.5) to
       (1.0,2.0) to (0.5,2.0) to (0.5,1.5);
       \draw[pattern=bricks,draw] (0.5,1.0) to (1.0,1.0) to
       (1.0,1.5) to (0.5,1.5) to (0.5,1.0);
       \draw[pattern=north west lines,draw] (0.5,0.5) to (1.0,0.5) to
       (1.0,1.0) to (0.5,1.0) to (0.5,0.5);
       \draw[pattern=crosshatch dots,draw] (0.5,0.0) to (1.0,0.0) to (1.0,0.5)
       to (0.5,0.5) to (0.5,0.0);

       \draw[pattern=north east lines,draw] (1.0,1.5) to (1.5,1.5) to
       (1.5,2.0) to (1.0,2.0) to (1.0,1.5);
       \draw[pattern=crosshatch dots,draw] (1.0,1.0) to (1.5,1.0) to
       (1.5,1.5) to (1.0,1.5) to (1.0,1.0);
       \draw[pattern=north east lines,draw] (1.0,0.5) to (1.5,0.5) to
       (1.5,1.0) to (1.0,1.0) to (1.0,0.5);
       \draw[pattern=north east lines,draw] (1.0,0.0) to (1.5,0.0) to
       (1.5,0.5) to (1.0,0.5) to (1.0,0.0);

       \draw[pattern=north east lines,draw] (1.5,1.5) to (2.0,1.5) to
       (2.0,2.0) to (1.5,2.0) to (1.5,1.5);
       \draw[pattern=north west lines,draw] (1.5,1.0) to (2.0,1.0) to
       (2.0,1.5) to (1.5,1.5) to (1.5,1.0);
       \draw[pattern=north east lines,draw] (1.5,0.5) to (2.0,0.5) to
       (2.0,1.0) to (1.5,1.0) to (1.5,0.5);
       \draw[pattern=north east lines,draw] (1.5,0.0) to (2.0,0.0) to
       (2.0,0.5) to (1.5,0.5) to (1.5,0.0);

       \draw[pattern=bricks,draw] (2.0,1.5) to (2.5,1.5) to
       (2.5,2.0) to (2.0,2.0) to (2.0,1.5);
       \draw[pattern=horizontal lines,draw] (2.0,1.0) to (2.5,1.0) to
       (2.5,1.5) to (2.0,1.5) to (2.0,1.0);
       \draw[pattern=crosshatch dots,draw] (2.0,0.5) to (2.5,0.5) to
       (2.5,1.0) to (2.0,1.0) to (2.0,0.5);
       \draw[pattern=crosshatch dots,draw] (2.0,0.0) to (2.5,0.0) to
       (2.5,0.5) to (2.0,0.5) to (2.0,0.0);

       \draw[pattern=bricks,draw] (2.5,1.5) to (3.0,1.5) to
       (3.0,2.0) to (2.5,2.0) to (2.5,1.5);
       \draw[pattern=horizontal lines,draw] (2.5,1.0) to (3.0,1.0) to
       (3.0,1.5) to (2.5,1.5) to (2.5,1.0);
       \draw[pattern=bricks,draw] (2.5,0.5) to (3.0,0.5) to
       (3.0,1.0) to (2.5,1.0) to (2.5,0.5);
       \draw[pattern=crosshatch dots,draw] (2.5,0.0) to (3.0,0.0) to
       (3.0,0.5) to (2.5,0.5) to (2.5,0.0);

       \draw[pattern=crosshatch dots,draw] (3.0,1.5) to (3.5,1.5) to
       (3.5,2.0) to (3.0,2.0) to (3.0,1.5);
       \draw[pattern=north east lines,draw] (3.0,1.0) to (3.5,1.0) to
       (3.5,1.5) to (3.0,1.5) to (3.0,1.0);
       \draw[pattern=grid,draw] (3.0,0.5) to (3.5,0.5) to
       (3.5,1.0) to (3.0,1.0) to (3.0,0.5);
       \draw[pattern=north east lines,draw] (3.0,0.0) to (3.5,0.0) to
       (3.5,0.5) to (3.0,0.5) to (3.0,0.0);

       \draw[pattern=crosshatch dots,draw] (3.5,1.5) to (4.0,1.5) to
       (4.0,2.0) to (3.5,2.0) to (3.5,1.5);
       \draw[pattern=north east lines,draw] (3.5,1.0) to (4.0,1.0) to
       (4.0,1.5) to (3.5,1.5) to (3.5,1.0);
       \draw[pattern=grid,draw] (3.5,0.5) to (4.0,0.5) to
       (4.0,1.0) to (3.5,1.0) to (3.5,0.5);
       \draw[pattern=grid,draw] (3.5,0.0) to (4.0,0.0) to
       (4.0,0.5) to (3.5,0.5) to (3.5,0.0);

       \draw[very thick,red] (0.0,2.0) to (1.0,2.0) to
       (1.0,1.5) to (0.0,1.5) to (0.0,2.0);

       \draw[very thick,red] (1.0,1.5) to (2.0,1.5) to
       (2.0,1.0) to (1.0,1.0) to (1.0,1.5);

       \draw[very thick,red] (2.0,1.0) to (3.0,1.0) to
       (3.0,0.5) to (2.0,0.5) to (2.0,1.0);

       \draw[very thick,red] (3.0,0.5) to (4.0,0.5) to
       (4.0,0.0) to (3.0,0.0) to (3.0,0.5);

       \node[anchor=mid,inner sep=1pt] (s1) at (0.25,-0.7) {$s_1$};
       \node[anchor=mid,inner sep=1pt] (s2) at (0.75,-0.7) {$s_2$};
       \node[anchor=mid,inner sep=1pt] (s3) at (1.25,-0.7) {$s_3$};
       \node[anchor=mid,inner sep=1pt] (s4) at (1.75,-0.7) {$s_4$};
       \node[anchor=mid,inner sep=1pt] (s5) at (2.25,-0.7) {$s_5$};
       \node[anchor=mid,inner sep=1pt] (s6) at (2.75,-0.7) {$s_6$};
       \node[anchor=mid,inner sep=1pt] (s7) at (3.25,-0.7) {$s_7$};
       \node[anchor=mid,inner sep=1pt] (s8) at (3.75,-0.7) {$s_8$};

       \node[anchor=mid east] (text) at (-1.0,-0.7) {Value};
       \draw[ar] (text) to (s1);

       \draw (s1.south) to [in=-90,out=-90]
       node[sloped,draw,fill=white,circle,inner sep=0.5pt] {\tiny $\neq$}
       (s2.south);
       \draw (s3.south) to [in=-90,out=-90]
       node[sloped,draw,fill=white,circle,inner sep=0.5pt] {\tiny $\neq$}
       (s4.south);
       \draw (s5.south) to [in=-90,out=-90]
       node[sloped,draw,fill=white,circle,inner sep=0.5pt] {\tiny $\neq$}
       (s6.south);
       \draw (s7.south) to [in=-90,out=-90]
       node[sloped,draw,fill=white,circle,inner sep=0.5pt] {\tiny $\neq$}
       (s8.south);
     \end{tikzpicture}
   \end{center}
   \caption{A monotonous \chain matrix (increasing)}
   \label{fig:inctame}
 \end{figure}
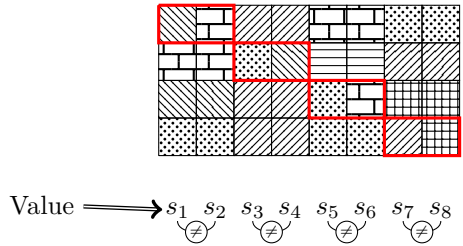

 \begin{lemma} \label{lem:mono}
   There exist monotonous \chain matrices of arbitrarily large length.
 \end{lemma}

 \begin{proof}
   Set $n \in \nat$, we explain how to construct a monotonous \chain matrix of
   length~$2n$. By Lemma~\ref{lem:tame}, there exists a tame \chain matrix
   \mat of length $2n^2$. Set $x_1,\dots,x_{n^2}$ the indices such that
   for all $j$, $\alt{\mat}{j} = \{x_j\}$. Note that by tameness, $x_j
   \neq x_{j'}$ for $j \neq j'$. Since the sequence $x_1,\dots,x_{n^2}$
   is of length $n^2$, we can extract, using Erdös-Szekeres theorem, a
   monotonous sequence of length $n$, $x_{j_1} < \cdots < x_{j_n}$ or
   $x_{j_1} > \cdots > x_{j_n}$ with $j_1 < \cdots < j_n$.  By symmetry
   we assume it is the former and construct an increasing \chain matrix
   of length $n$.

   Let \pat be the matrix of length $2n$ obtained from \mat, by keeping
   only the pairs of columns $2j-1,2j$ for $j \in \{j_1,\dots,j_n\}$. Set
   $x'_1,\dots,x'_{n}$ the indices such that for all $j$, $\alt{\pat}{j}
   = \{x'_j\}$. By definition, $x'_{1} < \cdots < x'_{n}$. We now want
   $\pat$ to have exactly $n$ rows. Note that the rows  whose indices do not belong to
   $\{x'_{1} , \cdots , x'_{n}\}$ are constant \chains. We simply merge these rows
   with others. For example, if row $i$ is labeled with the
   constant \chain $(s,\dots,s)$, let $(s_1,\dots,s_{2n})$ be the label of
   row $i+1$. We remove row $i$ and replace row $i+1$ by the \dchain
   $(ss_1,\dots,ss_{2n})$. Repeating the operation yields the desired increasing
   monotonous \chain~matrix.
 \end{proof}

 \subsection{Construction of the Contradiction Matrix}

 We can now use Lemma~\ref{lem:mono} to construct a contradiction
 matrix and end the proof of Proposition~\ref{prop:width}. We state
 this in the following proposition.

 \begin{proposition} \label{prop:contmat}
   There exists a contradiction matrix.
 \end{proposition}

 The remainder of this section is devoted to the proof of
 Proposition~\ref{prop:contmat}. The result follows from a Ramsey
 argument. We use Lemma~\ref{lem:mono} to choose a monotonous matrix of
 sufficiently large length. Then, we use Ramsey's Theorem (for
 hypergraphs with edges of size $3$) to extract the desired
 contradiction matrix.

 We first define the length of the monotonous \chain matrix that we need to
 pick. By Ramsey's Theorem, for every $m \in \nat$ there exists a number
 $\varphi(m)$ such that for any complete 3-hypergraph with hyperedges colored
 over the monoid~$M$, there exists a complete sub-hypergraph of size $m$ in
 which all edges share the same color. We choose $n = \varphi(\varphi(4)+1)$. By Lemma~\ref{lem:mono}, there exists a monotonous
 \chain matrix \mat of length $2n$. Since it is monotonous, \mat has $n$ rows.

 By symmetry, we assume that \mat is increasing and use it to construct an
 increasing contradiction matrix.  We use our choice of $n$ to extract a
 contradiction matrix from \mat. We proceed in two steps using Ramsey's Theorem
 each time. In the first step we treat all entries above the diagonal in \mat
 and in the second step all entries below the diagonal. We state the first step
 in the next lemma.

 \begin{lemma} \label{lem:matlemma}
   There exists an increasing monotonous matrix \mnat of length $2 \cdot
   \varphi(4)$ such that all cells above the diagonal contain the same
   idempotent $f \in M$.
 \end{lemma}

 \begin{proof}
   This is proved by applying Ramsey's Theorem to \mat. Consider the complete
   3-hypergraph whose nodes are $\{0,\ldots,n\}$. We label the hyperedge
   $\{i_1,i_2,i_3\}$ where $i_1 < i_2 < i_3$ by the value obtained by
   multiplying in the monoid $M$, the cells that appear in rows
   $i_1+1,\ldots,i_2$ in column $2i_3-1$. Observe that since
   $i_1 < i_2 < i_3$, by monotonicity, these entries are the same as in column
   $2i_3$. More formally, the label of the hyperedge $\{i_1,i_2,i_3\}$ with
   $i_1<i_2<i_3$ is therefore
   \[
     \mat_{i_1+1,2i_3-1} \cdots \mat_{i_2,2i_3-1} = \mat_{i_1+1,2i_3} \cdots \mat_{i_2,2i_3}.
   \]
   By choice of $n$, we can apply Ramsey's Theorem to this coloring. We get a
   subset of $\varphi(4)+1$ vertices, say
   $K = \{k_1,\ldots,k_{\varphi(4)+1}\} \subseteq \{0,\ldots,n\}$, such that
   all hyperedges connecting nodes in $K$ have the same color, say $f \in M$.
   For $i_1<i_2<i_3<i_4$ in~$K$, note that the color of the hyperedge
   $\{i_1,i_3,i_4\}$ is by definition the product of the colors of the
   hyperedges $\{i_1,i_2,i_4\}$ and $\{i_2,i_3,i_4\}$. Therefore, the common
   color $f$ needs to be an idempotent: $f = ff$. We now extract the desired
   matrix \mnat from \mat according to the subset $K$. The main idea is that
   the new row $i$ in \mnat will be the merging of rows $k_{i}+1$ to $k_{i+1}$
   in \mat and the new pair of columns $2j-1,2j$ will correspond to the pair
   $2k_{j+1}-1,2k_{j+1}$ in \mat.

   \smallskip We first merge rows. For all $i \geqslant 1$, we ``merge'' all
   rows from $k_{i}+1$ to $k_{i+1}$ into a single row. More precisely, this
   means that we replace the rows $k_{i}+1$ to $k_{i+1}$ by a single row
   containing the \dchain
   \[
     (\mat_{k_{i}+1,1} \cdots \mat_{k_{i+1},1},\ldots,\mat_{k_{i}+1,2n} \cdots \mat_{k_{i+1},2n})
   \]

   Moreover, we remove the top and bottom rows, \emph{i.e.}, rows $1$ to $k_1$ and
   rows $k_{\varphi(4)+1}+1$ to $\varphi(4)+1$. Then we remove all columns from $1$ to
   $2k_2 -2$, all columns from $2k_{\varphi(4)+1} + 1$ to $2n$, and for all $i
   \geqslant 2$, all columns from $2k_i+1$ to $2k_{i+1}-2$. One can verify
   that these two operations applied together preserve
   monotonicity. Observe that the resulting matrix \mnat has exactly
   $2\cdot \varphi(4)$ columns. Moreover, the cell $i,2j$ in the new
   matrix contains entry $\mat_{k_{i}+1,2k_{j+1}} \cdots
   \mat_{k_{i+1},2k_{j+1}}$. In particular if $j > i$, by definition of
   the set $K$, this entry is $f$, which means \mnat satisfies the
   conditions of the lemma.
 \end{proof}

 It remains to apply Ramsey's Theorem a second time to the matrix \mnat
 obtained from Lemma~\ref{lem:matlemma} to treat the cells below the
 diagonal and get the contradiction matrix. We state this in the
 following last lemma.

 \begin{lemma}
   There exists an increasing monotonous matrix \pat of length $6$ such
   that all cells above the diagonal contain the same idempotent $f \in
   M$ and all cells below the diagonal contain the same idempotent $e \in
   M$ (i.e. \pat is an increasing contradiction matrix).
 \end{lemma}

 \begin{proof}
   The argument is identical to the one of Lemma~\ref{lem:matlemma}. This
   time we apply it to the matrix \mnat of length $2 \cdot \varphi(4)$ for the
   cells below the diagonal. The monochromatic set given by Ramsey's theorem
   is this time of size $4$, which, with the above construction, will leave a
   matrix with 3 rows and 6 columns.
 \end{proof}

\section{Adding Successor: The Enriched Hierarchy}
\label{sec:transfer}
All decidability results we have proved so far are for fragments of
the order hierarchy. In this section we transfer these results to
the enriched hierarchy. More precisely, we present algorithms for the
following problems:
\begin{itemize}
\item the separation problem for \siwsd and \piwsd.
\item the membership problem for \siwst.
\item the membership problem for \bswsd.
\end{itemize}
For each problem, we actually present a reduction to same problem for the
corresponding fragment in the order hierarchy, and decidability then follows
from the results of the previous sections. The transfer results are not new
and were initially presented in~\cite{StrauVD,pinweilVD} for the membership
problem and in~\cite{MR1709911,Steinberg:delay-pointlikes:2001} and
\cite{pzsucc} for the separation problem (unlike the former, this latter
work also cope with classes not closed under complement and can therefore
be applied to \siwi). In this section, we only state the reductions and refer the
reader to these papers for proofs.

Note that the reductions we use are all taken from~\cite{pzsucc}. In
particular, for membership, while the underlying ideas remain similar
to that of~\cite{StrauVD,pinweilVD}, the reduction itself is fairly
different from the original one.

We divide the section in two parts. In the first part, we define the
main tool used in the reductions: \emph{the morphism of well-formed
  words}. In the second part, we present the reductions themselves.

\subsection{Morphism of Well-Formed Words}

Fix a morphism $\alpha: A^* \rightarrow M$ into a finite monoid
$M$. We define $E \subseteq M$ as the set of idempotents of
$\alpha(A^+)$, \emph{i.e.}, $E$ is the set of idempotents of $M$ that are
images of a nonempty word. We define a new alphabet \wfA, called
\emph{alphabet of well-formed words of $\alpha$}, as follows:
\[
  \wfA = \begin{array}{cl} &M\\ \cup &M \times E\\ \cup& E \times M\\ \cup &E \times M \times E\end{array}
\]
We will not be interested in all words in $\wfA^*$, but only in those
that are well-formed. A word $\wbb \in \wfA^*$ is said to be
\emph{well-formed} if one of the two following properties
hold:
\begin{itemize}
\item $\wbb = \varepsilon$ or is a single-letter word $s \in M$.
\item $\wbb = (s_1,f_1)(e_2,s_2,f_2)(e_3,s_3,f_3) \cdots (e_n,s_n) \in
  (S \times E) \cdot (E \times S \times E)^* \cdot (E \times S)$ and for all $1 \leqslant i \leqslant n-1$, we have
  $f_i=e_{i+1}$.
\end{itemize}

The following fact is immediate.

\begin{fact} \label{fct:reg}
  The set of well-formed words of $\wfA^*$ is a regular language.
\end{fact}

Observe that one can define a monoid morphism $\beta : \wfA^*
\rightarrow M$ by setting $\beta(s) = s$ for all $s \in M$,
$\beta((e,s)) = es$ for all $(e,s) \in E \times M$, $\beta((s,e)) =
se$ for all $(s,e) \in M \times E$ and $\beta((e,s,f)) = esf$ for all
$(e,s,f) \in E \times M \times E$. We call $\beta$ the \emph{morpshim
  of well-formed words associated to $\alpha$}.

\medskip
\noindent {\bf Associated language of well-formed words.} To any
language $L \subseteq A^*$ that is recognized by $\alpha$, one can
associate a language of well-formed words $\Lbb \subseteq \wfA^*$
(depending on $\alpha$):
\[
  \Lbb =\bigl \{\wbb \in \wfA^* \mid \wbb \text{ is well-formed and } \beta(\wbb)
  \in \alpha(L)\bigr\}.
\]
By definition, the language $\Lbb \subseteq \wfA^*$ is the intersection
of the language of well-formed words with $\beta^{-1}(\alpha(L))$. Therefore,
it is immediate by Fact~\ref{fct:reg} that it is regular, more precisely:

\begin{fact} \label{fct:reg2}
  Let $L \subseteq A^*$ that is recognized by $\alpha$. Then the
  associated language of well-formed words $\Lbb \subseteq \wfA^*$ is
  a regular language, and one can compute it from~$\alpha$.
\end{fact}

\subsection{Reductions}

We can now state the reductions, we begin with the separation result.

\begin{theorem}[\citeN{pzsucc}] \label{thm:reducs}
  Let $L_0,L_1$ be regular languages and let $\alpha: A^* \rightarrow M$
  be a morphism into a finite monoid $M$ that recognizes both $L_0$ and
  $L_1$. Finally, set $\Lbb_0$ and $\Lbb_1$ as the languages of
  well-formed words associated to $L_0$ and $L_1$.

  For all $i \geqslant 1$, $L_0$ is $\siwsi$-separable
  (resp. \bswsi-separable) from $L_1$ iff and only if $\Lbb_0$ is
  $\siwi$-separable  (resp. \bswi-separable) from~$\Lbb_1$.
\end{theorem}

Theorem~\ref{thm:reducs} reduces $\siwsi$-separability
(resp. \bswsi-separability) to $\siwi$-separability
(resp. \bswi-separability). Since, we already know that
\siwd-separability is decidable (see Corollary~\ref{cor:decidsep}), we
get the following corollary:

\begin{corollary} \label{cor:decidseps}
  Given as input two regular languages $L_1,L_2$ it is decidable to test
  whether $L_1$ can be $\siwsd$-separated (resp. $\piwsd$-separated) from
  $L_2$.
\end{corollary}

This terminates our separation results. We now state the membership
reduction. 

\begin{theorem}[\citeN{pzsucc}] \label{thm:reducm}
  Let $L$ be a regular language and let $\alpha: A^* \rightarrow M$
  be a morphism into a finite monoid $M$ that recognizes $L$. Finally,
  set $\Lbb$ as the language of well-formed words associated to $L$.

  For all $i \geqslant 3$, $L$ is $\siwsi$-definable if and only if $\Lbb$
  is $\siwi$-definable.

  For all $i \geqslant 2$, $L$ is $\bswsi$-definable if and only if $\Lbb$
  is $\bswi$-definable.
\end{theorem}

Observe that, in contrast to the separation reduction, the membership
reduction does not work for lower levels in the hierarchy. For
example, it does not work for \bswu and \siwd. This is essentially
because these logics are not powerful enough to express that a word
in $\wfA^*$ is well-formed (this is only possible for logics including
and above \piwd).

\smallskip By combining Theorem~\ref{thm:reducm} with
Corollaries~\ref{cor:decidthree} and~\ref{cor:decid2}, we get the desired
corollary.

\begin{corollary} \label{cor:decids}
  Given as input a regular language $L$, the following problems are
  decidable:
  \begin{itemize}
  \item whether $L$ is definable in \bswsd.
  \item whether $L$ is definable in \dewst.
  \item whether $L$ is definable in \siwst.
  \item whether $L$ is definable in \piwst.
  \end{itemize}
\end{corollary}

\section{Conclusion}
\label{sec:conc}
We solved the separation problem for \siwd using the new notion of \dchains,
and we used our solution to prove decidable characterizations for \bswd,
\dewt, \siwt and \piwt. The main open problem in this field remains to lift up
these results to higher levels in the hierarchy. In particular, we proved that
for any positive integer~$i$, generalizing our separation solution to \siwi (\emph{i.e.}, being able
to compute the \ichains of length~$2$) would yield a decidable
characterization for \siw{i+1}, \piw{i+1} and \dew{i+1}.

Our algorithm for computing \dchains cannot be directly generalized for higher
levels. An obvious reason for this is the fact that it considers \dchains
parametrized by sub-alphabets. This parameter is designed to take care of the
alternation between levels $1$ and $2$, but is not adequate for higher levels.
However, this problem has been circumvented for the next level: a new
algorithm to compute \siwt-\chains has been designed and proved
in~\cite{pseps3}. This requires introducing hybrid objects capturing even more
information than \siwt-\chains and \siwt-\juns, and which are amenable to a
recursive computation. Yet, this difficulty is unlikely to be the only problem. In
particular, we do have an algorithm that avoids using the alphabet, but it
remains difficult to generalize. We leave the presentation of this alternate
algorithm for further work.

Another orthogonal research direction is to solve separation for \bswi levels.
The idea of exploiting the knowledge on some class to solve separation for the
boolean algebra it generates is actually meaningful for other classes than
levels of the alternation hierarchy. Indeed, one can generalize the
relationship between \siwi-\chains with unbounded alternation and separation
for \bswi (as stated in Theorem~\ref{thm:sepbc}) by replacing the class \siwi
with any lattice $\mathcal{L}$ of regular languages. Otherwise stated, one can
generalize the definitions to make generic the link between
$\mathcal{L}$-\chains with unbounded alternation and separation by languages
of $\mathcal{BL}$, the boolean algebra generated by the lattice $\mathcal{L}$.
Even for \bswd, the problem of determining, for two given elements $s_1,s_2$
of the monoid under consideration, whether the set of \chains $(s_1,s_2)^*$
only consists of \siwi-\chains is still wide open. Solving it may provide
intuition for upper levels, but probably requires new concepts.

\bibliographystyle{ACM-Reference-Format-Journals}

\begin{thebibliography}{00}


\ifx \showCODEN    \undefined \def \showCODEN     #1{\unskip}     \fi
\ifx \showDOI      \undefined \def \showDOI       #1{{\tt DOI:}\penalty0{#1}\ }
  \fi
\ifx \showISBNx    \undefined \def \showISBNx     #1{\unskip}     \fi
\ifx \showISBNxiii \undefined \def \showISBNxiii  #1{\unskip}     \fi
\ifx \showISSN     \undefined \def \showISSN      #1{\unskip}     \fi
\ifx \showLCCN     \undefined \def \showLCCN      #1{\unskip}     \fi
\ifx \shownote     \undefined \def \shownote      #1{#1}          \fi
\ifx \showarticletitle \undefined \def \showarticletitle #1{#1}   \fi
\ifx \showURL      \undefined \def \showURL       #1{#1}          \fi

\bibitem[\protect\citeauthoryear{Albert, Baldinger, and Rhodes}{Albert
  et~al\mbox{.}}{1992}]
        {ABR:Undec-Identity:92}
{Douglas Albert}, {Robert Baldinger}, {and} {John Rhodes}. 1992.
\newblock \showarticletitle{Undecidability of the Identity Problem for Finite
  Semigroups}.
\newblock {\em The Journal of Symbolic Logic\/} {57}, 1 (1992), 179--192.
\newblock


\bibitem[\protect\citeauthoryear{Almeida}{Almeida}{1991}]
        {AlmeidaBS1}
{Jorge Almeida}. 1991.
\newblock \showarticletitle{Implicit operations on finite J-trivial semigroups
  and a conjecture of I. Simon}.
\newblock {\em Journal of Pure and Applied Algebra\/} {69}, 3 (1991), 205--218.
\newblock


\bibitem[\protect\citeauthoryear{Almeida}{Almeida}{1995}]
        {AlmeidaBook}
{J. Almeida}. 1995.
\newblock {\em Finite Semigroups and Universal Algebra}.
\newblock World Scientific, Singapore.
\newblock


\bibitem[\protect\citeauthoryear{Almeida}{Almeida}{1999}]
        {MR1709911}
{Jorge Almeida}. 1999.
\newblock \showarticletitle{Some Algorithmic Problems for Pseudovarieties}.
\newblock {\em Publicationes Mathematicae Debrecen\/}  {54} (1999), 531--552.
\newblock


\bibitem[\protect\citeauthoryear{Almeida and Klíma}{Almeida and
  Klíma}{2009}]
        {AK2009}
{Jorge Almeida} {and} {Ondrej Klíma}. 2009.
\newblock \showarticletitle{A Counterexample to a Conjecture Concerning
  Concatenation Hierarchies}.
\newblock {\it Inform. Process. Lett.} {110}, 1 (2009), 4--7.
\newblock


\bibitem[\protect\citeauthoryear{Almeida and Klíma}{Almeida and
  Klíma}{2010}]
        {AK2010}
{Jorge Almeida} {and} {Ondrej Klíma}. 2010.
\newblock \showarticletitle{New Decidable Upper Bound of the 2nd Level in the
  {S}traubing-{T}hérien Concatenation Hierarchy of Star-Free Languages}.
\newblock {\em Discrete Mathematics {\&} Theoretical Computer Science\/} {12},
  4 (2010), 41--58.
\newblock


\bibitem[\protect\citeauthoryear{Almeida and Zeitoun}{Almeida and
  Zeitoun}{1997}]
        {AZ97-J}
{Jorge Almeida} {and} {Marc Zeitoun}. 1997.
\newblock \showarticletitle{{The pseudovariety J is hyperdecidable}}.
\newblock {\em {RAIRO Inform. Théor. Appl.}\/} {31}, 5 (1997), 457--482.
\newblock


\bibitem[\protect\citeauthoryear{Arfi}{Arfi}{1987}]
        {arfi87}
{Mustapha Arfi}. 1987.
\newblock \showarticletitle{Polynomial Operations on Rational Languages}. In
  {\em Proceedings of the 4th Annual Symposium on Theoretical Aspects of
  Computer Science, {STACS'87}} {\em (Lecture Notes in Computer Science)}.
  Springer-Verlag, Berlin, Heidelberg, 198--206.
\newblock


\bibitem[\protect\citeauthoryear{Arfi}{Arfi}{1991}]
        {Arfi_1991}
{Mustapha Arfi}. 1991.
\newblock \showarticletitle{Opérations polynomiales et hiérarchies de
  concaténation}.
\newblock {\em Theoretical Computer Science\/} {91}, 1 (1991), 71--84.
\newblock


\bibitem[\protect\citeauthoryear{Auinger}{Auinger}{2010}]
        {DBLP:journals/ijac/Auinger10}
{Karl Auinger}. 2010.
\newblock \showarticletitle{On the Decidability of Membership in the Global of
  a Monoid Pseudovariety}.
\newblock {\em IJAC\/} {20}, 2 (2010), 181--188.
\newblock


\bibitem[\protect\citeauthoryear{Beauquier and Pin}{Beauquier and Pin}{1989}]
        {Beauquier_1989}
{Danièle Beauquier} {and} {Jean-\'Eric Pin}. 1989.
\newblock \showarticletitle{Factors of words}.
\newblock In {\em Proceedings of the 16th International Colloquium on Automata,
  Languages, and Programming, {ICALP'89}}. Springer-Verlag, Berlin, Heidelberg,
  63--79.
\newblock


\bibitem[\protect\citeauthoryear{Beauquier and Pin}{Beauquier and Pin}{1991}]
        {Beauquier_1991}
{Danièle Beauquier} {and} {Jean-\'Eric Pin}. 1991.
\newblock \showarticletitle{Languages and scanners}.
\newblock {\em Theoretical Computer Science\/} {84}, 1 (1991), 3--21.
\newblock


\bibitem[\protect\citeauthoryear{Boja{\'n}czyk}{Boja{\'n}czyk}{2007}]
        {bojLTT}
{Mikolaj Boja{\'n}czyk}. 2007.
\newblock \showarticletitle{A new algorithm for testing if a regular language
  is locally threshold testable}.
\newblock {\it Inform. Process. Lett.} {104}, 3 (2007), 91--94.
\newblock


\bibitem[\protect\citeauthoryear{Boja{\'n}czyk}{Boja{\'n}czyk}{2009}]
        {bfacto}
{Miko{\l}aj Boja{\'n}czyk}. 2009.
\newblock \showarticletitle{Factorization Forests}. In {\em Proceedings of the
  13th International Conference on Developments in Language Theory, {DLT'09}}
  {\em (Lecture Notes in Computer Science)}. Springer-Verlag, Berlin,
  Heidelberg, 1--17.
\newblock


\bibitem[\protect\citeauthoryear{Boja{\'n}czyk and Place}{Boja{\'n}czyk and
  Place}{2012}]
        {bpopen}
{Miko{\l}aj Boja{\'n}czyk} {and} {Thomas Place}. 2012.
\newblock \showarticletitle{Regular Languages of Infinite Trees that are
  Boolean Combinations of Open Sets}. In {\em Proceedings of the 39th
  International Colloquium on Automata, Languages and Programming, {ICALP'12}}
  {\em (Lecture Notes in Computer Science)}. Springer-Verlag, Berlin,
  Heidelberg, 104--115.
\newblock


\bibitem[\protect\citeauthoryear{Brzozowski}{Brzozowski}{1976}]
        {DBLP:journals/ita/Brzozowski76}
{Janusz~A. Brzozowski}. 1976.
\newblock \showarticletitle{Hierarchies of Aperiodic Languages}.
\newblock {\em {ITA}\/} {10}, 2 (1976), 33--49.
\newblock


\bibitem[\protect\citeauthoryear{Brzozowski and Cohen}{Brzozowski and
  Cohen}{1971}]
        {BrzoDot}
{Janusz~A. Brzozowski} {and} {Rina~S. Cohen}. 1971.
\newblock \showarticletitle{Dot-Depth of Star-Free Events}.
\newblock {\it J. Comput. System Sci.} {5}, 1 (1971), 1--16.
\newblock


\bibitem[\protect\citeauthoryear{Brzozowski and Knast}{Brzozowski and
  Knast}{1978}]
        {BroKnaStrict}
{Janusz~A. Brzozowski} {and} {Robert Knast}. 1978.
\newblock \showarticletitle{The Dot-Depth Hierarchy of Star-Free Languages is
  Infinite}.
\newblock {\it J. Comput. System Sci.} {16}, 1 (1978), 37--55.
\newblock


\bibitem[\protect\citeauthoryear{Brzozowski and Simon}{Brzozowski and
  Simon}{1971}]
        {BSlocalConf}
{Janusz~A. Brzozowski} {and} {Imre Simon}. 1971.
\newblock \showarticletitle{Characterizations of locally testable events}. In
  {\em 12th Annual Symposium on Switching and Automata Theory (swat 1971)}.
  Institute of Electrical {\&} Electronics Engineers ({IEEE}), East Lansing,
  MI, USA, 166--176.
\newblock


\bibitem[\protect\citeauthoryear{Brzozowski and Simon}{Brzozowski and
  Simon}{1973}]
        {BSlocal}
{Janusz~A. Brzozowski} {and} {Imre Simon}. 1973.
\newblock \showarticletitle{Characterizations of locally testable events}.
\newblock {\em Discrete Mathematics\/} {4}, 3 (1973), 243--271.
\newblock


\bibitem[\protect\citeauthoryear{Büchi}{Büchi}{1960}]
        {BuchiMSO}
{Julius~R. Büchi}. 1960.
\newblock \showarticletitle{Weak Second-Order Arithmetic and Finite Automata}.
\newblock {\em Mathematical Logic Quarterly\/} {6}, 1-6 (1960), 66--92.
\newblock


\bibitem[\protect\citeauthoryear{Cho and Huynh}{Cho and Huynh}{1991}]
        {DFA-sf-PSPACE}
{Sang Cho} {and} {Dung~T. Huynh}. 1991.
\newblock \showarticletitle{Finite-automaton aperiodicity is
  {PSPACE}-complete}.
\newblock {\em Theoretical Computer Science\/} {88}, 1 (1991), 99--116.
\newblock


\bibitem[\protect\citeauthoryear{Colcombet}{Colcombet}{2010}]
        {cfacto}
{Thomas Colcombet}. 2010.
\newblock \showarticletitle{Factorization Forests for Infinite Words and
  Applications to Countable Scattered Linear Orderings}.
\newblock {\em Theoritical Computer Science\/} {411}, 4-5 (2010), 751--764.
\newblock


\bibitem[\protect\citeauthoryear{Colcombet}{Colcombet}{2011}]
        {conf/lata/Colcombet11}
{Thomas Colcombet}. 2011.
\newblock \showarticletitle{{Green's} Relations and Their Use in Automata
  Theory.}. In {\em Proceedings of Language and Automata Theory and
  Applications, 5th International Conference (LATA'11)} {\em (Lecture Notes in
  Computer Science)}, Vol. 6638. Springer-Verlag, Berlin Heidelberg, 1--21.
\newblock


\bibitem[\protect\citeauthoryear{Colcombet}{Colcombet}{2015}]
        {tc-handbook15}
{Thomas Colcombet}. 2015.
\newblock The Factorisation Forest Theorem.  (2015).
\newblock
\newblock
\shownote{To appear in the handbook ``Automata: from Mathematics to
  Applications''.}


\bibitem[\protect\citeauthoryear{Cowan}{Cowan}{1993}]
        {Cowan_1993}
{David Cowan}. 1993.
\newblock \showarticletitle{Inverse Monoids of Dot-Depth Two}.
\newblock {\em Internat. J. Algebra Comput.\/} {03}, 04 (1993), 411--424.
\newblock


\bibitem[\protect\citeauthoryear{Czerwi{\'n}ski, Martens, and
  Masopust}{Czerwi{\'n}ski et~al\mbox{.}}{2013}]
        {martens}
{Wojciech Czerwi{\'n}ski}, {Wim Martens}, {and} {Tomá{\v{s}} Masopust}. 2013.
\newblock \showarticletitle{Efficient Separability of Regular Languages by
  Subsequences and Suffixes}. In {\em Proceedings of the 40th International
  Colloquium on Automata, Languages, and Programming, {ICALP'13}} {\em (Lecture
  Notes in Computer Science)}. Springer-Verlag, Berlin, Heidelberg, 150--161.
\newblock


\bibitem[\protect\citeauthoryear{Diekert and Gastin}{Diekert and
  Gastin}{2008}]
        {DGfo}
{Volker Diekert} {and} {Paul Gastin}. 2008.
\newblock \showarticletitle{First-order definable languages}.
\newblock In {\em Logic and Automata: History and Perspectives}, {Jörg Flum},
  {Erich Grädel}, {and} {Thomas Wilke} (Eds.). Texts in Logic and Games,
  Vol.~2. Amsterdam University Press, 261--306.
\newblock


\bibitem[\protect\citeauthoryear{Eilenberg}{Eilenberg}{1976}]
        {EilenbergB}
{Samuel Eilenberg}. 1976.
\newblock {\em Automata, Languages, and Machines}. Vol.~B.
\newblock Academic Press, Inc., Orlando, FL, USA.
\newblock


\bibitem[\protect\citeauthoryear{Elgot}{Elgot}{1961}]
        {ElgotMSO}
{Calvin~C. Elgot}. 1961.
\newblock \showarticletitle{Decision Problems of Finite Automata Design and
  Related Arithmetics}.
\newblock {\it Trans. Amer. Math. Soc.} {98}, 1 (1961), 21--51.
\newblock


\bibitem[\protect\citeauthoryear{Gla{\ss}er and Schmitz}{Gla{\ss}er and
  Schmitz}{2000}]
        {gssig2}
{Christian Gla{\ss}er} {and} {Heinz Schmitz}. 2000.
\newblock \showarticletitle{Languages of Dot-Depth 3/2}. In {\em Proceedings of
  the 17th Annual Symposium on Theoretical Aspects of Computer Science,
  {STACS'00}} {\em (Lecture Notes in Computer Science)}. Springer-Verlag,
  Berlin, Heidelberg, 555--566.
\newblock


\bibitem[\protect\citeauthoryear{Hashiguchi}{Hashiguchi}{1983}]
        {Hashiguchi_1983}
{Kosaburo Hashiguchi}. 1983.
\newblock \showarticletitle{Representation theorems on regular languages}.
\newblock {\it J. Comput. System Sci.} {27}, 1 (1983), 101--115.
\newblock


\bibitem[\protect\citeauthoryear{Henckell}{Henckell}{1988}]
        {henckell:1988}
{Karsten Henckell}. 1988.
\newblock \showarticletitle{Pointlike sets: the finest aperiodic cover of a
  finite semigroup}.
\newblock {\em J. Pure Appl. Algebra\/}  {55} (1988), 85--126.
\newblock


\bibitem[\protect\citeauthoryear{Henckell and Pin}{Henckell and Pin}{2000}]
        {HenckellPinJordered}
{Karsten Henckell} {and} {Jean-\'Eric Pin}. 2000.
\newblock \showarticletitle{Ordered Monoids and J-Trivial Monoids}.
\newblock In {\em Algorithmic Problems in Groups and Semigroups}.
  Springer-Verlag, 121--137.
\newblock


\bibitem[\protect\citeauthoryear{Henckell, Rhodes, and Steinberg}{Henckell
  et~al\mbox{.}}{2010}]
        {DBLP:journals/ijac/HenckellRS10a}
{Karsten Henckell}, {John Rhodes}, {and} {Benjamin Steinberg}. 2010.
\newblock \showarticletitle{Aperiodic Pointlikes and Beyond}.
\newblock {\em Internat. J. Algebra Comput.\/} {20}, 2 (2010), 287--305.
\newblock


\bibitem[\protect\citeauthoryear{Higgins}{Higgins}{1997}]
        {Higgins_1997}
{Peter~M. Higgins}. 1997.
\newblock \showarticletitle{A proof of Simon's theorem on piecewise testable
  languages}.
\newblock {\em Theoretical Computer Science\/} {178}, 1-2 (1997), 257--264.
\newblock


\bibitem[\protect\citeauthoryear{Higgins}{Higgins}{2000}]
        {Higgins:sf:2000}
{Peter~M. Higgins}. 2000.
\newblock \showarticletitle{A new proof of {Schützenberger's} theorem}.
\newblock {\em International Journal of Algebra and Computation\/} {10}, 02
  (2000), 217--220.
\newblock


\bibitem[\protect\citeauthoryear{Howie}{Howie}{1991}]
        {Howie91}
{John~M. Howie}. 1991.
\newblock {\em Automata and Languages}.
\newblock Clarendon Press, Oxford.
\newblock


\bibitem[\protect\citeauthoryear{Immerman}{Immerman}{1999}]
        {Immerman:DC:1999}
{Neil Immerman}. 1999.
\newblock {\em Descriptive Complexity}.
\newblock Springer.
\newblock


\bibitem[\protect\citeauthoryear{Karandikar, Kufleitner, and
  Schnoebelen}{Karandikar et~al\mbox{.}}{2015}]
        {pk-phs-mk-indexSimon}
{Prateek Karandikar}, {Manfred Kufleitner}, {and} {Philippe Schnoebelen}. 2015.
\newblock \showarticletitle{On the index of Simon's congruence for piecewise
  testability}.
\newblock {\it Inform. Process. Lett.} {115}, 4 (2015), 515--519.
\newblock


\bibitem[\protect\citeauthoryear{Klíma}{Klíma}{2011}]
        {Klima:Piecewise-testable-languages-combinatorics:2011:a}
{Ond\v{r}ej Klíma}. 2011.
\newblock \showarticletitle{Piecewise testable languages via combinatorics on
  words}.
\newblock {\em Discrete Mathematics\/} {311}, 20 (2011), 2124--2127.
\newblock


\bibitem[\protect\citeauthoryear{Klíma and Polák}{Klíma and Polák}{2013}]
        {KlimaPTComb}
{Ond{\v{r}}ej Klíma} {and} {Libor Polák}. 2013.
\newblock \showarticletitle{Alternative Automata Characterization of Piecewise
  Testable Languages}.
\newblock In {\em Developments in Language Theory}. Springer-Verlag, 289--300.
\newblock


\bibitem[\protect\citeauthoryear{Knast}{Knast}{1983a}]
        {knast83}
{Robert Knast}. 1983a.
\newblock \showarticletitle{A Semigroup Characterization of Dot-Depth One
  Languages}.
\newblock {\em RAIRO - Theoretical Informatics and Applications\/} {17}, 4
  (1983), 321--330.
\newblock


\bibitem[\protect\citeauthoryear{Knast}{Knast}{1983b}]
        {KnastGraph83}
{Robert Knast}. 1983b.
\newblock \showarticletitle{Some theorems on graph congruences}.
\newblock {\em RAIRO - Theoretical Informatics and Applications\/} {17}, 4
  (1983), 331--342.
\newblock


\bibitem[\protect\citeauthoryear{Kufleitner}{Kufleitner}{2008}]
        {kfacto}
{Manfred Kufleitner}. 2008.
\newblock \showarticletitle{The Height of Factorization Forests}. In {\em
  Proceedings of the 33rd International Symposium on Mathematical Foundations
  of Computer Science, {MFCS'08}} {\em (Lecture Notes in Computer Science)}.
  Springer-Verlag, Berlin, Heidelberg, 443--454.
\newblock


\bibitem[\protect\citeauthoryear{Lallement}{Lallement}{1979}]
        {Lallement:1979:SCA:539871}
{Gérard Lallement}. 1979.
\newblock {\em Semigroups and Combinatorial Applications}.
\newblock John Wiley \& Sons, Inc., New York, NY, USA.
\newblock
\showISBNx{0471043796}


\bibitem[\protect\citeauthoryear{Libkin}{Libkin}{2004}]
        {LibkinFMT:2004}
{Leonid Libkin}. 2004.
\newblock {\em Elements Of Finite Model Theory}.
\newblock Springer.
\newblock


\bibitem[\protect\citeauthoryear{Lucchesi, Simon, Simon, Simon, and
  Kowaltowski}{Lucchesi et~al\mbox{.}}{1979}]
        {LuchesiSimonKowaltowskiBook}
{Cláudio.~L. Lucchesi}, {Imre Simon}, {Istvan Simon}, {Janos Simon}, {and}
  {Tomasz Kowaltowski}. 1979.
\newblock {\em Aspectos teóricos da computação}.
\newblock IMPA, São Paulo.
\newblock
\showURL{
\url{http://www.impa.br/opencms/pt/biblioteca/cbm/11CBM/11_CBM_77_04.pdf}}


\bibitem[\protect\citeauthoryear{Margolis and Pin}{Margolis and Pin}{1985}]
        {Margolis&Pin:Products-group-languages:1985:a}
{Stuart~W. Margolis} {and} {Jean-\'Eric Pin}. 1985.
\newblock \showarticletitle{Products of group languages}.
\newblock In {\em Fundamentals of {Computation} {Theory}}, {Lothar Budach}
  (Ed.). Number 199 in Lecture Notes in Computer Science. Springer, 285--299.
\newblock


\bibitem[\protect\citeauthoryear{{McNaughton}}{{McNaughton}}{1974}]
        {McNaughton:Algebraic-decision-procedures-local:1974:a}
{Robert {McNaughton}}. 1974.
\newblock \showarticletitle{Algebraic decision procedures for local
  testability}.
\newblock {\em Mathematical Systems Theory\/} {8}, 1 (1974), 60--76.
\newblock


\bibitem[\protect\citeauthoryear{McNaughton and Papert}{McNaughton and
  Papert}{1971}]
        {mnpfo}
{Robert McNaughton} {and} {Seymour~A. Papert}. 1971.
\newblock {\em Counter-Free Automata}.
\newblock {MIT} Press.
\newblock


\bibitem[\protect\citeauthoryear{Meyer}{Meyer}{1969}]
        {MeyerFO}
{Albert~R. Meyer}. 1969.
\newblock \showarticletitle{A Note on Star-Free Events}.
\newblock {\em J. ACM\/} {16}, 2 (1969), 220--225.
\newblock


\bibitem[\protect\citeauthoryear{Nerode}{Nerode}{1958}]
        {mn58}
{Anil Nerode}. 1958.
\newblock \showarticletitle{Linear Automaton Transformations}.
\newblock {\it Proc. Amer. Math. Soc.} {9}, 4 (1958), 541--544.
\newblock


\bibitem[\protect\citeauthoryear{Perrin}{Perrin}{1990}]
        {perrinfo}
{Dominique Perrin}. 1990.
\newblock \showarticletitle{Finite Automata}.
\newblock In {\em Formal Models and Semantics}. Elsevier, 1--57.
\newblock


\bibitem[\protect\citeauthoryear{Perrin and Pin}{Perrin and Pin}{1986}]
        {PPOrder}
{Dominique Perrin} {and} {Jean-\'Eric Pin}. 1986.
\newblock \showarticletitle{First-Order Logic and Star-Free Sets}.
\newblock {\it J. Comput. System Sci.} {32}, 3 (1986), 393--406.
\newblock


\bibitem[\protect\citeauthoryear{Pin}{Pin}{1984}]
        {Pin:Varietes-langages-formels:1984:a}
{Jean-\'Eric Pin}. 1984.
\newblock {\em Variétés de langages formels}.
\newblock Masson, Paris.
\newblock
\newblock
\shownote{English translation: 1986, Varieties of formal languages, Plenum,
  New-York.}


\bibitem[\protect\citeauthoryear{Pin}{Pin}{1995a}]
        {Pin_1995}
{Jean-\'Eric Pin}. 1995a.
\newblock \showarticletitle{Finite Semigroups and Recognizable Languages: An
  Introduction}.
\newblock In {\em Semigroups, Formal Languages and Groups}. Springer-Verlag,
  1--32.
\newblock


\bibitem[\protect\citeauthoryear{Pin}{Pin}{1995b}]
        {pinordered}
{Jean-\'Eric Pin}. 1995b.
\newblock \showarticletitle{A variety theorem without complementation}.
\newblock {\em Russian Mathem. (Iz. VUZ)\/}  {39} (1995), 74–83.
\newblock


\bibitem[\protect\citeauthoryear{Pin}{Pin}{1996}]
        {PinLTT}
{Jean-\'Eric Pin}. 1996.
\newblock \showarticletitle{The expressive power of existential first order
  sentences of Büchi's sequential calculus}.
\newblock In {\em Proceedings of the 23rd International Colloquium on Automata,
  Languages, and Programming, {ICALP'96}}. Springer-Verlag, 300--311.
\newblock


\bibitem[\protect\citeauthoryear{Pin}{Pin}{1997}]
        {Pin_1997}
{Jean-\'Eric Pin}. 1997.
\newblock \showarticletitle{Syntactic Semigroups}.
\newblock In {\em Handbook of Formal Languages}. Springer-Verlag, 679--746.
\newblock


\bibitem[\protect\citeauthoryear{Pin}{Pin}{1998}]
        {pinbridges}
{Jean-\'Eric Pin}. 1998.
\newblock \showarticletitle{Bridges for Concatenation Hierarchies}. In {\em
  Proceedings of the 25th International Colloquium on Automata, Languages and
  Programming, {ICALP'98}} {\em (Lecture Notes in Computer Science)}.
  Springer-Verlag, Berlin, Heidelberg, 431--442.
\newblock


\bibitem[\protect\citeauthoryear{Pin}{Pin}{2005}]
        {PinLTTjournal}
{Jean-\'Eric Pin}. 2005.
\newblock \showarticletitle{Expressive power of existential first-order
  sentences of {B}üchi's sequential calculus}.
\newblock {\em Discrete Mathematics\/} {291}, 1-3 (2005), 155--174.
\newblock


\bibitem[\protect\citeauthoryear{Pin}{Pin}{2011}]
        {Pin-ThemeVar2011}
{Jean-\'Eric Pin}. 2011.
\newblock \showarticletitle{Theme and Variations on the Concatenation Product}.
  In {\em Proceedings of the 4th International Conference on Algebraic
  Informatics, {CAI'11}} {\em (Lecture Notes in Computer Science)}.
  Springer-Verlag, Berlin, Heidelberg, 44--64.
\newblock


\bibitem[\protect\citeauthoryear{Pin}{Pin}{2013}]
        {PinIntersection13}
{Jean-\'Eric Pin}. 2013.
\newblock \showarticletitle{An Explicit Formula for the Intersection of Two
  Polynomials of Regular Languages}.
\newblock In {\em Developments in Language Theory}. Springer, Berlin
  Heidelberg, 31--45.
\newblock


\bibitem[\protect\citeauthoryear{Pin}{Pin}{2016a}]
        {Pin:WSPC16}
{Jean-{\'E}ric Pin}. 2016a.
\newblock \showarticletitle{The dot-depth hierarchy, 45 years later}. In {\em
  WSPC Proceedings}.
\newblock
\newblock
\shownote{To appear.}


\bibitem[\protect\citeauthoryear{Pin}{Pin}{2016b}]
        {Pin:Mathematical-Foundations-Automata-Theory:2015:a}
{Jean-\'Eric Pin}. 2016b.
\newblock Mathematical Foundations of Automata Theory.  (2016).
\newblock
\showURL{
\url{http://liafa.jussieu.fr/~jep/MPRI/MPRI.html}}
\newblock
\shownote{In preparation.}


\bibitem[\protect\citeauthoryear{Pin and Straubing}{Pin and Straubing}{1981}]
        {pin-straubing:upper}
{Jean-\'Eric Pin} {and} {Howard Straubing}. 1981.
\newblock \showarticletitle{Monoids of Upper Triangular Boolean Matrices}.
\newblock In {\em Semigroups. Structure and Universal Algebraic Problems},
  {S.~Schwarz G.~Pollák} {and} {O.~Steinfeld} (Eds.). Colloquia Mathematica
  Societatis Janos Bolyal, Vol.~39. North-Holland, Szeged, Hungary, 259--272.
\newblock


\bibitem[\protect\citeauthoryear{Pin and Weil}{Pin and Weil}{1995}]
        {pwdelta_conf}
{Jean-\'Eric Pin} {and} {Pascal Weil}. 1995.
\newblock \showarticletitle{Polynomial closure and unambiguous product}.
\newblock In {\em Proceedings of the 22nd International Colloquium on Automata,
  Languages, and Programming, {ICALP'95}}. Springer-Verlag, 348--359.
\newblock


\bibitem[\protect\citeauthoryear{Pin and Weil}{Pin and Weil}{1996a}]
        {pin:hal-00143953}
{Jean-\'Eric Pin} {and} {Pascal Weil}. 1996a.
\newblock \showarticletitle{Profinite semigroups, {Mal'cev} products and
  identities}.
\newblock {\em Journal of Algebra\/} {182}, 3 (1996), 604--626.
\newblock


\bibitem[\protect\citeauthoryear{Pin and Weil}{Pin and Weil}{1996b}]
        {pin:hal-00143951}
{Jean-\'Eric Pin} {and} {Pascal Weil}. 1996b.
\newblock \showarticletitle{A {Reiterman} theorem for pseudovarieties of finite
  first-order structures}.
\newblock {\em Algebra Universalis\/} {35}, 4 (1996), 577--595.
\newblock


\bibitem[\protect\citeauthoryear{Pin and Weil}{Pin and Weil}{2001}]
        {Pin_2001}
{Jean-\'Eric Pin} {and} {Pascal Weil}. 2001.
\newblock \showarticletitle{A conjecture on the concatenation product}.
\newblock {\em {RAIRO Informatique Théorique}\/} {35}, 6 (2001), 597--618.
\newblock


\bibitem[\protect\citeauthoryear{Pin and Weil}{Pin and Weil}{2002}]
        {pinweilVD}
{Jean-\'Eric Pin} {and} {Pascal Weil}. 2002.
\newblock \showarticletitle{The Wreath Product Principle for Ordered
  Semigroups}.
\newblock {\em Communications in Algebra\/}  {30} (2002), 5677--5713.
\newblock


\bibitem[\protect\citeauthoryear{Pin and Weil}{Pin and Weil}{1997}]
        {pwdelta}
{Jean-{\'Eric} Pin} {and} {Pascal Weil}. 1997.
\newblock \showarticletitle{Polynomial Closure and Unambiguous Product}.
\newblock {\em Theory of Computing Systems\/} {30}, 4 (1997), 383--422.
\newblock


\bibitem[\protect\citeauthoryear{Pippenger}{Pippenger}{1997}]
        {pippengerbook}
{Nicholas Pippenger}. 1997.
\newblock {\em Theories of computability}.
\newblock Cambridge University Press.
\newblock
\showISBNx{978-0-521-55380-3}


\bibitem[\protect\citeauthoryear{Place}{Place}{2015}]
        {pseps3}
{Thomas Place}. 2015.
\newblock \showarticletitle{Separating Regular Languages with Two Quantifier
  Alternations}. In {\em Proceedings of the 30th Annual {ACM/IEEE} Symposium on
  Logic in Computer Science {(LICS'15)}}. IEEE, Kyoto, Japan, 202--213.
\newblock


\bibitem[\protect\citeauthoryear{Place, van Rooijen, and Zeitoun}{Place
  et~al\mbox{.}}{2013a}]
        {pvzltt}
{Thomas Place}, {Lorijn van Rooijen}, {and} {Marc Zeitoun}. 2013a.
\newblock \showarticletitle{Separating Regular Languages by Locally Testable
  and Locally Threshold Testable Languages}. In {\em Proceedings of the 33rd
  IARCS Annual Conference on Foundations of Software Technology and Theoretical
  Computer Science, {FSTTCS'13}} {\em (Leibniz International Proceedings in
  Informatics (LIPIcs))}. Schloss Dagstuhl--Leibniz-Zentrum fuer Informatik,
  Dagstuhl, Germany, 363--375.
\newblock


\bibitem[\protect\citeauthoryear{Place, van Rooijen, and Zeitoun}{Place
  et~al\mbox{.}}{2013b}]
        {pvzmfcs13}
{Thomas Place}, {Lorijn van Rooijen}, {and} {Marc Zeitoun}. 2013b.
\newblock \showarticletitle{Separating Regular Languages by Piecewise Testable
  and Unambiguous Languages}. In {\em Proceedings of the 38th International
  Symposium on Mathematical Foundations of Computer Science, {MFCS'13}} {\em
  (Lecture Notes in Computer Science)}. Springer-Verlag, Berlin, Heidelberg,
  729--740.
\newblock


\bibitem[\protect\citeauthoryear{Place and Zeitoun}{Place and Zeitoun}{2014a}]
        {pzqalt}
{Thomas Place} {and} {Marc Zeitoun}. 2014a.
\newblock \showarticletitle{Going Higher in the First-Order Quantifier
  Alternation Hierarchy on Words}. In {\em Proceedings of the 41st
  International Colloquium on Automata, Languages, and Programming, {ICALP'14}}
  {\em (Lecture Notes in Computer Science)}. Springer-Verlag, Berlin,
  Heidelberg, 342--353.
\newblock


\bibitem[\protect\citeauthoryear{Place and Zeitoun}{Place and Zeitoun}{2014b}]
        {pzfo}
{Thomas Place} {and} {Marc Zeitoun}. 2014b.
\newblock \showarticletitle{Separating Regular Languages with First-order
  Logic}. In {\em Proceedings of the Joint Meeting of the 23rd EACSL Annual
  Conference on Computer Science Logic (CSL'14) and the 29th Annual ACM/IEEE
  Symposium on Logic in Computer Science (LICS'14)}. ACM, New York, NY, USA,
  75:1--75:10.
\newblock


\bibitem[\protect\citeauthoryear{Place and Zeitoun}{Place and Zeitoun}{2015a}]
        {pzsucc}
{Thomas Place} {and} {Marc Zeitoun}. 2015a.
\newblock \showarticletitle{Separation and the Successor Relation}. In {\em
  32nd International Symposium on Theoretical Aspects of Computer Science
  (STACS 2015)} {\em (Leibniz International Proceedings in Informatics
  (LIPIcs))}, {Ernst~W. Mayr} {and} {Nicolas Ollinger} (Eds.), Vol.~30. Schloss
  Dagstuhl--Leibniz-Zentrum fuer Informatik, Dagstuhl, Germany, 662--675.
\newblock
\showISBNx{978-3-939897-78-1}
\showISSN{1868-8969}


\bibitem[\protect\citeauthoryear{Place and Zeitoun}{Place and Zeitoun}{2015b}]
        {PZ:Siglog15}
{Thomas Place} {and} {Marc Zeitoun}. 2015b.
\newblock \showarticletitle{The Tale of the Quantifier Alternation Hierarchy of
  First-Order Logic over Words}.
\newblock {\em SIGLOG news\/} {2}, 3 (2015), 4--17.
\newblock
\showURL{
\url{http://siglog.hosting.acm.org/wp-content/uploads/2015/10/siglog_news_5.pdf}}


\bibitem[\protect\citeauthoryear{Place and Zeitoun}{Place and Zeitoun}{2016}]
        {PZ:FO-Sep16}
{Thomas Place} {and} {Marc Zeitoun}. 2016.
\newblock \showarticletitle{Separating Regular Languages with First-Order
  Logic}.
\newblock {\em Logical Methods in Computer Science\/} (2016).
\newblock
\showURL{
\url{http://arxiv.org/pdf/1402.3277v2}}


\bibitem[\protect\citeauthoryear{Reinhardt}{Reinhardt}{2002}]
        {DBLP:conf/dagstuhl/Reinhardt01}
{Klaus Reinhardt}. 2002.
\newblock \showarticletitle{The Complexity of Translating Logic to Finite
  Automata}. In {\em Automata, Logics, and Infinite Games: {A} Guide to Current
  Research [outcome of a Dagstuhl seminar, February 2001]} {\em (Lecture Notes
  in Computer Science)}, {Erich Grädel}, {Wolfgang Thomas}, {and} {Thomas
  Wilke} (Eds.), Vol. 2500. Springer-Verlag, Berlin, Heidelberg, 231--238.
\newblock


\bibitem[\protect\citeauthoryear{Reiterman}{Reiterman}{1982}]
        {Reiterman:Birkhoff-theorem-finite-algebras:1982:a}
{Jan Reiterman}. 1982.
\newblock \showarticletitle{The {Birkhoff} theorem for finite algebras}.
\newblock {\em Algebra Universalis\/} {14}, 1 (1982), 1--10.
\newblock


\bibitem[\protect\citeauthoryear{Rhodes}{Rhodes}{1999}]
        {Rhodes99-undec}
{John Rhodes}. 1999.
\newblock \showarticletitle{Undecidability, Automata, and Pseudovarities of
  Finite Semigroups}.
\newblock {\em Internat. J. Algebra Comput.\/} {9}, 3-4 (1999), 455--474.
\newblock


\bibitem[\protect\citeauthoryear{Sakarovitch and Simon}{Sakarovitch and
  Simon}{1997}]
        {SakaSimon}
{Jacques Sakarovitch} {and} {Imre Simon}. 1997.
\newblock {\em Combinatorics on Words, Lothaire}.
\newblock Cambridge University Press, Chapter 6, Subwords.
\newblock


\bibitem[\protect\citeauthoryear{Schützenberger}{Schützenberger}{1956}]
        {schsynt}
{Marcel~Paul Schützenberger}. 1955-1956.
\newblock \showarticletitle{Une théorie algébrique du codage}.
\newblock {\em Séminaire Dubreil. Algèbre et théorie des nombres\/}  {9}
  (1955-1956), 1--24.
\newblock
\showURL{
\url{http://eudml.org/doc/111094}}


\bibitem[\protect\citeauthoryear{Schützenberger}{Schützenberger}{1965}]
        {sfo}
{Marcel~Paul Schützenberger}. 1965.
\newblock \showarticletitle{On Finite Monoids Having Only Trivial Subgroups}.
\newblock {\em Information and Control\/} {8}, 2 (1965), 190--194.
\newblock


\bibitem[\protect\citeauthoryear{Schützenberger}{Schützenberger}{1976}]
        {schul}
{Marcel~Paul Schützenberger}. 1976.
\newblock \showarticletitle{Sur le produit de concaténation non ambigu}.
\newblock {\em Semigroup Forum\/}  {13} (1976), 47--75.
\newblock


\bibitem[\protect\citeauthoryear{Simon}{Simon}{1972}]
        {SimonPhD}
{Imre Simon}. 1972.
\newblock {\em Hierarchies of events of dot-depth one}.
\newblock Ph.D. Dissertation. University of Waterloo.
\newblock


\bibitem[\protect\citeauthoryear{Simon}{Simon}{1975}]
        {simon75}
{Imre Simon}. 1975.
\newblock \showarticletitle{Piecewise Testable Events}. In {\em Proceedings of
  the 2nd GI Conference on Automata Theory and Formal Languages}.
  Springer-Verlag, Berlin, Heidelberg, 214--222.
\newblock


\bibitem[\protect\citeauthoryear{Simon}{Simon}{1990}]
        {simonfacto}
{Imre Simon}. 1990.
\newblock \showarticletitle{Factorization Forests of Finite Height}.
\newblock {\em Theoritical Computer Science\/} {72}, 1 (1990), 65--94.
\newblock


\bibitem[\protect\citeauthoryear{Steinberg}{Steinberg}{2001}]
        {Steinberg:delay-pointlikes:2001}
{Benjamin Steinberg}. 2001.
\newblock \showarticletitle{A delay theorem for pointlikes}.
\newblock {\em Semigroup Forum\/} {63}, 3 (2001), 281--304.
\newblock


\bibitem[\protect\citeauthoryear{Stern}{Stern}{1985a}]
        {Stern:Characterizations}
{Jacques Stern}. 1985a.
\newblock \showarticletitle{Characterizations of some classes of regular
  events}.
\newblock {\em Theoretical Computer Science\/}  {35} (1985), 17--42.
\newblock


\bibitem[\protect\citeauthoryear{Stern}{Stern}{1985b}]
        {Stern:Complexity-some-problems-from:1985:a}
{Jacques Stern}. 1985b.
\newblock \showarticletitle{Complexity of some problems from the theory of
  automata}.
\newblock {\em Information and Control\/} {66}, 3 (1985).
\newblock


\bibitem[\protect\citeauthoryear{Stockmeyer}{Stockmeyer}{1974}]
        {Stockmeyer-phd}
{Larry~J. Stockmeyer}. 1974.
\newblock {\em The complexity of decision problems in automata theory and
  logic}.
\newblock Ph.D. Dissertation. Massachusetts Institute of Technology.
\newblock
\showURL{
\url{http://opac.inria.fr/record=b1000295}}
\newblock
\shownote{PHD.}


\bibitem[\protect\citeauthoryear{Stockmeyer and Meyer}{Stockmeyer and
  Meyer}{1973}]
        {Stockmeyer:1973:WPR:800125.804029}
{Larry~J. Stockmeyer} {and} {Albert~R. Meyer}. 1973.
\newblock \showarticletitle{Word Problems Requiring Exponential Time
  (Preliminary Report)}. In {\em Proceedings of the Fifth Annual ACM Symposium
  on Theory of Computing, {STOC '73}}, {Alfred~V. Aho}, {Allan Borodin},
  {Robert~L. Constable}, {Robert~W. Floyd}, {Michael~A. Harrison}, {Richard~M.
  Karp}, {and} {H.~Raymond Strong} (Eds.). ACM, New York, NY, USA, 1--9.
\newblock


\bibitem[\protect\citeauthoryear{Straubing}{Straubing}{1981}]
        {StrauConcat}
{Howard Straubing}. 1981.
\newblock \showarticletitle{A Generalization of the Schützenberger Product of
  Finite Monoids}.
\newblock {\em Theoretical Computer Science\/} {13}, 2 (1981), 137--150.
\newblock


\bibitem[\protect\citeauthoryear{Straubing}{Straubing}{1985}]
        {StrauVD}
{Howard Straubing}. 1985.
\newblock \showarticletitle{Finite Semigroup Varieties of the Form {V
  {\textasteriskcentered} D}}.
\newblock {\em Journal of Pure and Applied Algebra\/}  {36} (1985), 53--94.
\newblock


\bibitem[\protect\citeauthoryear{Straubing}{Straubing}{1986}]
        {StrauDD2Conf}
{Howard Straubing}. 1986.
\newblock \showarticletitle{Semigroups and languages of dot-depth 2}.
\newblock In {\em Proceedings of the 13th International Colloquium on Automata,
  Languages, and Programming, {ICALP'86}}, {Laurent Kott} (Ed.). Lecture Notes
  in Computer Science, Vol. 226. Springer-Verlag, Berlin Heidelberg, 416--423.
\newblock


\bibitem[\protect\citeauthoryear{Straubing}{Straubing}{1988}]
        {StrauDD2}
{Howard Straubing}. 1988.
\newblock \showarticletitle{Semigroups and Languages of Dot-Depth Two}.
\newblock {\em Theoretical Computer Science\/} {58}, 1-3 (1988), 361--378.
\newblock


\bibitem[\protect\citeauthoryear{Straubing}{Straubing}{1994}]
        {bookstraub}
{Howard Straubing}. 1994.
\newblock {\em Finite Automata, Formal Logic and Circuit Complexity}.
\newblock Birkhauser, Basel, Switzerland.
\newblock


\bibitem[\protect\citeauthoryear{Straubing and Th\'erien}{Straubing and
  Th\'erien}{1988}]
        {Straubing&Therien:Partially-ordered-finite-monoids:1988:a}
{Howard Straubing} {and} {Denis Th\'erien}. 1988.
\newblock \showarticletitle{Partially ordered finite monoids and a theorem of
  {I}. {Simon}}.
\newblock {\em Journal of Algebra\/} {119}, 2 (1988), 393--399.
\newblock


\bibitem[\protect\citeauthoryear{Straubing and Weil}{Straubing and
  Weil}{1992}]
        {Straubing_1992}
{Howard Straubing} {and} {Pascal Weil}. 1992.
\newblock \showarticletitle{On a conjecture concerning dot-depth two
  languages}.
\newblock {\em Theoretical Computer Science\/} {104}, 2 (1992), 161--183.
\newblock


\bibitem[\protect\citeauthoryear{Tesson and Th\'erien}{Tesson and
  Th\'erien}{2002}]
        {Tesson02diamondsare}
{Pascal Tesson} {and} {Denis Th\'erien}. 2002.
\newblock \showarticletitle{Diamonds Are Forever: The Variety {DA}}. In {\em
  Semigroups, Algorithms, Automata and Languages}. World Scientific, 475--500.
\newblock


\bibitem[\protect\citeauthoryear{Th\'erien}{Th\'erien}{1981}]
        {TheConcat}
{Denis Th\'erien}. 1981.
\newblock \showarticletitle{Classification of Finite Monoids: The Language
  Approach}.
\newblock {\em Theoretical Computer Science\/} {14}, 2 (1981), 195--208.
\newblock


\bibitem[\protect\citeauthoryear{Th\'erien}{Th\'erien}{2011}]
        {therien:powersurvey}
{Denis Th\'erien}. 2011.
\newblock \showarticletitle{The Power of Diversity}.
\newblock In {\em Descriptional Complexity of Formal Systems}, {Markus Holzer},
  {Martin Kutrib}, {and} {Giovanni Pighizzini} (Eds.). Lecture Notes in
  Computer Science, Vol. 6808. Springer-Verlag, 43--54.
\newblock


\bibitem[\protect\citeauthoryear{Th\'erien and Weiss}{Th\'erien and
  Weiss}{1985}]
        {Therien&Weiss:Graph-congruences-wreath-products:1985:a}
{Denis Th\'erien} {and} {Alex Weiss}. 1985.
\newblock \showarticletitle{Graph congruences and wreath products}.
\newblock {\em J. Pure Appl. Algebra\/}  {36} (1985), 205--215.
\newblock


\bibitem[\protect\citeauthoryear{Th\'erien and Wilke}{Th\'erien and
  Wilke}{1998}]
        {twfodeux}
{Denis Th\'erien} {and} {Thomas Wilke}. 1998.
\newblock \showarticletitle{Over Words, Two Variables Are As Powerful As One
  Quantifier Alternation}. In {\em Proceedings of the 30th Annual ACM Symposium
  on Theory of Computing, {STOC'98}}. Association for Computing Machinery
  ({ACM}), New York, NY, USA, 234--240.
\newblock
\showURL{
\url{http://dx.doi.org/10.1145/276698.276749}}


\bibitem[\protect\citeauthoryear{Thomas}{Thomas}{1982}]
        {ThomEqu}
{Wolfgang Thomas}. 1982.
\newblock \showarticletitle{Classifying Regular Events in Symbolic Logic}.
\newblock {\it J. Comput. System Sci.} {25}, 3 (1982), 360--376.
\newblock


\bibitem[\protect\citeauthoryear{Thomas}{Thomas}{1984}]
        {ThomStrict2}
{Wolfgang Thomas}. 1984.
\newblock \showarticletitle{An application of the {Ehrenfeucht}-{Fraissé} game
  in formal language theory}.
\newblock {\em Mémoires de la Société Mathématique de France\/}  {16}
  (1984), 11--21.
\newblock


\bibitem[\protect\citeauthoryear{Thomas}{Thomas}{1987}]
        {ThomStrict}
{Wolfgang Thomas}. 1987.
\newblock \showarticletitle{A concatenation game and the dot-depth hierarchy}.
\newblock In {\em Computation Theory and Logic}. Springer-Verlag, Berlin,
  Heidelberg, 415--426.
\newblock


\bibitem[\protect\citeauthoryear{Thomas}{Thomas}{1997}]
        {Thomas:Languages-automata-logic:1997:a}
{Wolfgang Thomas}. 1997.
\newblock \showarticletitle{Languages, automata, and logic}.
\newblock In {\em Handbook of formal languages}. Springer.
\newblock


\bibitem[\protect\citeauthoryear{Tilson}{Tilson}{1987}]
        {TilsonCat}
{Bret Tilson}. 1987.
\newblock \showarticletitle{Categories as algebra: {An} essential ingredient in
  the theory of monoids}.
\newblock {\em Journal of Pure and Applied Algebra\/} {48}, 1--2 (1987),
  83--198.
\newblock


\bibitem[\protect\citeauthoryear{Trahtman}{Trahtman}{2001a}]
        {Trahtman:LTTDFA}
{Avraham~N. Trahtman}. 2001a.
\newblock \showarticletitle{An Algorithm to Verify Local Threshold Testability
  of Deterministic Finite Automata}.
\newblock In {\em Automata Implementation}. Number 2214 in Lecture Notes in
  Computer Science. Springer-Verlag, 164--173.
\newblock


\bibitem[\protect\citeauthoryear{Trahtman}{Trahtman}{2001b}]
        {Trahtman:PT-LTT2001}
{Avraham~N. Trahtman}. 2001b.
\newblock \showarticletitle{Piecewise and Local Threshold Testability of
  {DFA}}. In {\em Proc. {FCT}'01} {\em (Lecture Notes in Computer Science)},
  Vol. 2138. Springer-Verlag, London, {UK}, {UK}, 347--358.
\newblock


\bibitem[\protect\citeauthoryear{Trakhtenbrot}{Trakhtenbrot}{1961}]
        {TrakhMSO}
{Boris~A. Trakhtenbrot}. 1961.
\newblock \showarticletitle{Finite Automata and Logic of Monadic Predicates}.
\newblock {\em Doklady Akademii Nauk SSSR\/}  {149} (1961), 326--329.
\newblock
\newblock
\shownote{In Russian.}


\bibitem[\protect\citeauthoryear{Weil}{Weil}{1989a}]
        {Weil-ConcatSurvey1989}
{Pascal Weil}. 1989a.
\newblock \showarticletitle{Concatenation Product: a Survey}.
\newblock In {\em Formal Properties of Finite Automata and Applications}.
  Lecture Notes in Computer Science, Vol. 386. Springer-Verlag, Berlin,
  Heidelberg, 120--137.
\newblock


\bibitem[\protect\citeauthoryear{Weil}{Weil}{1989b}]
        {Weil_1989}
{Pascal Weil}. 1989b.
\newblock \showarticletitle{Inverse monoids of dot-depth two}.
\newblock {\em Theoretical Computer Science\/} {66}, 3 (1989), 233--245.
\newblock


\bibitem[\protect\citeauthoryear{Wilke}{Wilke}{1999}]
        {wfo}
{Thomas Wilke}. 1999.
\newblock \showarticletitle{Classifying Discrete Temporal Properties}. In {\em
  Proceedings of the 16th Annual Conference on Theoretical Aspects of Computer
  Science, {STACS'99}} {\em (Lecture Notes in Computer Science)}.
  Springer-Verlag, Berlin, Heidelberg, 32--46.
\newblock


\bibitem[\protect\citeauthoryear{Zalcstein}{Zalcstein}{1972}]
        {Zalcstein:Locally-testable-languages:1972:a}
{Yechezkel Zalcstein}. 1972.
\newblock \showarticletitle{Locally testable languages}.
\newblock {\it J. Comput. System Sci.} {6}, 2 (1972), 151--167.
\newblock


\end{thebibliography}

\end{document}